\documentclass{amsart}  
\usepackage{amssymb, enumitem, float, caption, mathtools, tikz}
\usetikzlibrary{arrows, calc, decorations.markings, matrix, positioning}
\tikzset{>=latex}
\usepackage[final]{hyperref}

\newcommand{\NN}{\mathbb{N}}

\newcommand{\A}{\mathbb{A}}
\newcommand{\m}[1]{\mathbb{#1}}    
\newcommand{\cl}[1]{\mathcal{#1}}  

\theoremstyle{plain}
  \newtheorem{thm}{Theorem}[section]  \newtheorem*{thm*}{Theorem}
  \newtheorem{claim}{Claim}[thm]      \newtheorem*{claim*}{Claim}
    \newtheorem*{conj*}{Conjecture}
  \newtheorem{cor}[thm]{Corollary}    \newtheorem*{cor*}{Corollary}
  \newtheorem{lem}[thm]{Lemma}        \newtheorem*{lem*}{Lemma}
  \newtheorem{prop}[thm]{Proposition} \newtheorem*{prop*}{Proposition}
\theoremstyle{definition}
  \newtheorem{defn}[thm]{Definition} \newtheorem*{defn*}{Definition}
  \newtheorem{ex}[thm]{Example}      \newtheorem*{ex*}{Example}
\theoremstyle{remark}
    \newtheorem*{rk*}{Remark}
\newcommand{\Case}[1]{\smallskip \textbf{Case #1:}}
\newenvironment{claimproof} {
  \begin{proof}[Proof of claim]
  \renewcommand{\qedsymbol}{\ensuremath{\circ}}
  } {
  \end{proof}
  }

\DeclareMathOperator{\Clo}{Clo}

\DeclareMathOperator{\Rel}{Rel}
\DeclareMathOperator{\Sg}{Sg}

\newcommand{\pmat}[1]{ \begin{pmatrix} #1 \end{pmatrix} }
\newcommand{\ds}[1]{ \displaystyle{#1} }

\numberwithin{equation}{section}  
\renewcommand{\epsilon}{\varepsilon}
\renewcommand{\phi}{\varphi}

\usepackage{booktabs}
 
\setlist[enumerate]{itemsep=0.5em}
\setlist[itemize]{itemsep=0.5em}
\theoremstyle{definition} \newtheorem*{prob}{Problem}
\newcommand{\Proofitem}[1]{\smallskip \textbf{#1:}}
\newcommand{\vect}[1]{ \left< #1 \right> }
\newcommand{\AMConfig}[1]{ \big[\hspace{-0.225em}\big[\ #1\ \big]\hspace{-0.225em}\big] }
\newcommand{\Row}[1]{ [\![ #1 ]\!] }

\newcommand{\M}{\mathcal{M}}
\newcommand{\AM}{A(\M)}
\newcommand{\AAM}{\mathbb{\AM}}
\newcommand{\RR}{\mathbb{R}}
\let\SS\relax \newcommand{\SS}{\mathbb{S}}
\DeclareMathOperator{\Config}{\tt c}

\newcommand{\StateGraph}{\Sigma}
\DeclareMathOperator{\Content}{\tt con}
\DeclareMathOperator{\State}{\tt st}
\DeclareMathOperator{\X}{\tt X}
\newcommand{\Er}{\times}
\renewcommand{\Dot}{\bullet}

\newcommand{\D}{\mathcal{D}}
\renewcommand{\H}{\mathcal{H}}
\newcommand{\N}{\mathcal{N}}

\DeclareMathOperator{\Pol}{Pol}
\DeclareMathOperator{\RClo}{RClo}
\DeclareMathOperator{\RCloDual}{\RClo^{\partial}}
\newcommand{\Entails}{\models}
\newcommand{\EntailsDual}{\models_{\partial}}

\DeclareMathOperator{\maxdeg}{maxdeg}
\DeclareMathOperator{\Approx}{A}
\DeclareMathOperator{\ApproxI}{\Approx_I}

\begin{document}
\title{Finite degree clones are undecidable}
\author{Matthew Moore}
\date{\today}

\address{
  University of Kansas
  Dept.\ of Electrical Engineering and Computer Science;
  Eaton Hall;
  Lawrence, KS 66044;
  U.S.A.}
\email{matthew.moore@ku.edu}

\begin{abstract}
A clone of functions on a finite domain determines and is determined by its
system of invariant relations (=predicates). When a clone is determined by a
\emph{finite} number of relations, we say that the clone is of finite
degree. For each Minsky machine $\M$ we associate a finitely generated clone
$\cl{C}$ such that $\cl{C}$ has finite degree if and only if $\M$ halts,
thus proving that deciding whether a given clone has finite degree is
impossible.
\end{abstract}

\subjclass[2010]{08A30; 03D35}
\keywords{undecidable; finite degree; clones; relations; finitely related; algebra}
\maketitle 

\section{Introduction} \label{sec:intro} 
A clone is a set of operations on a domain which is closed under composition
and contains all projections. Emil Post~\cite{Post_Lattice} in 1941 famously
classified all clones on a 2-element domain (the Boolean clones), of which
there are countably many. In contrast to this, there are continuum many
clones over even a 3-element domain, as proven in 1959 by Janov and
Mu\v{c}nik~\cite{JanovMucnik_ContinuumClones}. The problem under
consideration in this paper has its roots in investigations in the 1970s of
the structure of the lattices of clones over domains of more than 2
elements. Before discussing the history of the problem, however, it will be
useful to establish some background.

There are two common methods of finitely specifying a clone of operations.
The first is to generate the clone from a finite set of functions via
composition and variable manipulations. The second method is to specify the
clone as all operations preserving a given finite set of relations. A
relation $\RR$ on domain $D$ is said to be \emph{preserved} by an operation
$f: D^n \to D$ if $f(r_1,\dots,r_n)\in R$ whenever $r_1,\dots,r_n\in R$. The
\emph{polymorphism clone} on a set $\cl{R}$ of relations over domain $D$ is
\[
  \Pol(\cl{R}) 
  = \bigcup_{n\in \NN} \Big\{ f: D^n \to D
    \mid f \text{ preserves each relation in $\cl{R}$} \Big\}.
\]
We say that a clone $\cl{C}$ is \emph{determined} by $\cl{R}$ if $\cl{C} =
\Pol(\cl{R})$. The supremum of the arities of the relations contained in
$\cl{R}$ is the \emph{degree} of $\cl{R}$, written
\[
  \deg(\cl{R}) = \sup \big\{ \text{arity}(R) \mid R\in \cl{R} \big\},
\]
and the degree of a clone $\cl{C}$ is the infimum of the degrees of all sets
of relations which determine $\cl{C}$,
\[
  \deg(\cl{C})
  = \inf \big\{ \deg(\cl{R}) \mid \cl{C} = \Pol(\cl{R}) \big\}.
\]
Both of these values can be infinite, and we regard them as \emph{total}
functions. Of course, since there are uncountably many clones on domains of
more than 2 elements, there is no enumeration of them and hence no standard
sense in which $\deg(\cdot)$ can be computable. We resolve this complication
by considering only those clones which have finite domain and are generated
by finitely many operations (i.e.\ the clones of finite algebras). The clone
generated by the algebra $\A = \left< A; f_1,\dots,f_n \right>$ is the
smallest clone with domain $A$ containing all the $f_i$. The problem that we
consider in this paper is the following, which we call the \emph{Finite
Degree Problem}:
\begin{itemize}[leftmargin=6em]
  \item[\texttt{Input:}] finite algebra 
    $\A = \left< A; f_1,\dots,f_n \right>$ generating clone $\cl{C}$
  \item[\texttt{Output:}] whether $\deg(\cl{C}) < \infty$.
\end{itemize}
We show that the Finite Degree Problem is undecidable by constructing for
each Minsky machine $\M$ a finite algebra $\AAM$ such that if $\cl{C}$ is
the clone generated by $\AAM$ then $\deg(\cl{C})<\infty$ if and only if $\M$
halts.

It is difficult to determine the precise origin of the Finite Degree
Problem. Questions surrounding the algorithmic computation of the degree of
a clone date back to the 1970s with papers by Romov~\cite{Romov_LocalPostI,
Romov_LocalPostII} and Jablonski\u{\i}~\cite{Jablonskii_UpperNeighborhood}.
The closely related question of deciding whether an algebra admits a natural
duality has been open since the late 1970s, but apparently first appears in
print in 1991 with Davey~\cite{Davey_Duality10Dollars}. The Finite Degree
Problem is likely a contemporary of this problem, but does not appear in
print until 2006 in~\cite{Bergman_ProblemListALV} in which it is credited by
Ralph McKenzie to Mikl\'{o}s Mar\'{o}ti in 2004. 

Investigations into which structures have finite degree and under what
conditions have yielded a host of results over the years, which we now give
a brief overview of. All of the following structures \emph{on a finite
domain} have finite degree:
\begin{itemize}
  \item all bands~\cite{Dolinka_FinRelBands};

  \item many semigroups, but not all of them~\cite{Mayr_FinRelSemigroups,
      DaveyJacksonPitkethly_FiniteDegree};

  \item semilattices, and more generally any clone containing a semilattice
    operation that commutes with the other operations~\cite{Davey_Endodual,
      DJPT_SemilatDual};

  \item clones containing the lattice operations of $\wedge$ and $\vee$, and
    more generally algebras with a near unanimity term (if the algebra
    belongs to a congruence distributive variety then this is an
    equivalence)~\cite{BakerPixley_Interpolation, Barto_CDNU};

  \item groups, rings, and more generally algebras with a cube term (if the
    algebra belongs to a congruence modular variety then this is an
    equivalence)~\cite{AichingerMayrMcKenzie_CubeFinRelated,
    Barto_FinRelCube}.
\end{itemize}
Aside from results for specific structures, necessary conditions for a clone
to have finite degree have also been established. Rosenberg and
Szendrei~\cite{RosenbergSzendrei_DegreesOfClones} and Davey and
Pitkethly~\cite{DaveyPitkethly_CountingRelns} both establish general
algebraic conditions which imply finite degree.

The technique of encoding a model of computation into an algebraic structure
was pioneered by McKenzie~\cite{McKenzie_ResidualBoundUndec,
McKenzie_TarskiUndec}, where it was proven that it is undecidable whether
an algebra is finitely axiomatizable (this is famously known as Tarski's
Problem). Since then, a handful of other authors have used a similar
approach to prove that other algebraic properties are undecidable.
Mar\'{o}ti~\cite{Maroti_NUUndec} proves that it is undecidable whether an
algebra has a near unanimity term defined on all but 2 elements of a finite
domain (it was later discovered that this is decidable without this
restriction, see Mar\'{o}ti~\cite{Maroti_NUDec}). McKenzie and
Wood~\cite{McKenzieWood_TypeSetUncomp} prove certain ``omitting types''
statements about algebras are undecidable. The author~\cite{Moore_DPSCUndec}
proves that the technical property of DPSC is undecidable, thus giving an
alternate proof of the undecidability of Tarski's Problem. Most recently,
Nurakunov and Stronkowski~\cite{NurakunovStronkowski_ProfiniteUndec} prove
that profiniteness is undecidable.

We begin in Section~\ref{sec:Minsky_machines} with a discussion of a simple
of model of computation, the Minsky machine, before continuing on to a brief
survey of the necessary algebraic background and some of the notation used
in the paper in Section~\ref{sec:background}. The algebra $\AAM$ mentioned
above is precisely defined in Section~\ref{sec:AM_defn}, and the exact
manner in which it encodes the computation of the Minsky machine $\M$ is
proven in Section~\ref{sec:encoding}. In Section~\ref{sec:M_doesnt_halt} we
show that $\deg(\AAM) = \infty$ when $\M$ does not halt. The converse is
quite a bit more complicated. Tools necessary for the analysis are developed
in Section~\ref{sec:M_halts_tools}, and the main argument is divided into
cases and addressed in Section~\ref{sec:M_halts_entailment}. Lastly,
Section~\ref{sec:conclusion} contains a statement of the main theorem and a
discussion of related open problems.

A great deal of effort was spent in constructing the algebra $\AAM$ so that
the entire argument would be as straightforward as possible. Much of this
effort took the form of computer experimentation and verification, allowing
for rapid iteration of the definitions. Significant portions of many of the
lemmas and theorems can be verified computationally. The framework that was
used was built specifically for this task, but the majority of it is suited
to general algebraic structures. This computational framework as well as
several examples are available online at the URL below.
\begin{center}
  \url{http://ittc.ku.edu/~moore/preprints/2018_AM.zip}
\end{center}

\section{Minsky machines} \label{sec:Minsky_machines}  
Minsky machines are a simple model of computation for which the halting
problem is undecidable, and were defined in 1961 by Marvin
Minsky~\cite{Minsky_machine, Minsky_Computation}. A \emph{Minsky machine}
has states $\{0,1,\dots,N\}$, where $0$ is the halting state and $1$ is the
initial state, registers $A$ and $B$ which hold non-negative values, and a
finite set of instructions. Minsky machine instructions come in two types:
\begin{itemize}
  \item $(i,R,j)$, interpreted as ``in state $i$, increase register $R$ by
    one and enter state $j$'', and

  \item $(i,R,k,j)$, interpreted as ``in state $i$, if register $R$ is $0$
    then enter state $k$, otherwise decrease $R$ by one and enter state
    $j$''.
\end{itemize}
In order for the instructions to unambiguously describe a Minsky machine
there must be an instruction of the form $(1,\dots)$ and for each state $i$
there must be at most one instruction of the form $(i,\dots)$.

A Minsky machine \emph{configuration} is a triple of the form
$(s,\alpha,\beta)$, where $s$ indicates the state of the machine, $\alpha$
the value of the $A$ register, and $\beta$ the value of the $B$ register.
Formally, a Minsky machine with states $\mathcal{S} = \{0,\dots,N\}$ and
instructions $\mathcal{I}$ is an operation on the space of all possible
configurations
\[
  \M : \mathcal{S}\times \NN \times \NN \to \mathcal{S}\times \NN \times \NN
\]
defined by
\[
  \M(i,\alpha,\beta) = \begin{cases}
    (j,\alpha+1,\beta) & \text{if } (i,A,j)\in \mathcal{I}, \\
    (j,\alpha,\beta+1) & \text{if } (i,B,j)\in \mathcal{I}, \\
    (j,\alpha-1,\beta) & \text{if } (i,A,k,j)\in \mathcal{I}, \alpha\neq 0, \\
    (j,\alpha,\beta-1) & \text{if } (i,B,k,j)\in \mathcal{I}, \beta\neq 0, \\
    (k,\alpha,\beta)   & \text{if } (i,A,k,j)\in \mathcal{I}, \alpha = 0, \\
    (k,\alpha,\beta)   & \text{if } (i,B,k,j)\in \mathcal{I}, \beta = 0, \\
    (0,\alpha,\beta)   & \text{if } i = 0.
  \end{cases}
\]
Since the function $\M$ is determined by $\mathcal{I}$, it is usual to use
the same symbol for both. That is, we indicate that $\M$ has some
instruction, say $(1,A,2)$, by simply writing $(1,A,2)\in \M$.

A single application of the function $\M$ to a configuration represents a
single computational step of the Minsky machine. To indicate multiple steps
in the computation, we can compose $\M$:
\[
  \M^n(i,\alpha,\beta) = \M\circ \dots \circ \M(i,\alpha,\beta).
\]
We say that a Minsky machine $\M$ \emph{halts on input $A=\alpha$,
$B=\beta$} if there is some $n$ such that $\M^n(1,\alpha,\beta) =
(0,\alpha',\beta')$. We say that a Minsky machine $\M$ \emph{halts} (without
reference to input) if it halts on input $A=B=0$. By replacing the halting
state with a new state $k$ and appending instructions $(k,A,k+1,k)$ and
$(k+1,B,0,k+1)$ to the list of instructions, a Minsky machine can be made to
return the registers to $0$ before halting. Thus, we can assume without loss
of generality that all halting machines return both registers to $0$ before
halting. Furthermore, any Minsky machine can be converted to an equivalent
machine with first instruction of the form $(1,R,s)$.

If a given Minsky machine with states $\{0, \dots, N\}$ does not have an
instruction of the form $(k,\dots)$ for some state $k$ then without changing
the halting status of the machine we may add an instruction of the form
$(k,R,k)$ to $\M$. We therefore assume throughout that Minsky machines have
exactly one instruction for each state $k$. 

Let $\StateGraph(\M)$ be the directed graph with vertices $[N]$ and an edge
$i\to j$ if and only if $\M(i,\alpha,\beta) = (j,\alpha',\beta')$ for some
$\alpha,\beta,\alpha',\beta'\in \NN$. We call this the \emph{state graph} of
$\M$. If state $\ell$ is reachable from state $k$ along a (possibly length
$0$) directed path then we write $k\rightsquigarrow \ell$. If there is a
state $k\in \StateGraph(\M)$ such that $1\not\rightsquigarrow k$ then we can
eliminate state $k$ from $\M$ without changing the halting status of $\M$.
We therefore assume that all states are reachable from $1$ in the state
graph.

A further modification of $\M$ allows us to assume that every state has a
path to the halting state $0$. If we have $1\rightsquigarrow \ell
\not\rightsquigarrow 0$ then there is at least one pair of states $i,k\in
\StateGraph(\M)$ such that $i \to k$, $i\rightsquigarrow 0$, and
$k\not\rightsquigarrow 0$. This is only possible if $(i,R,k,j)\in \M$ or
$(i,R,j,k)\in \M$ for some register $R$ and $j\rightsquigarrow 0$. For all
such pairs $i,k$, we do the following:
\begin{itemize}
  \item add a new state $n_k^i$,
  \item if $(i,R,k,j)\in \M$ then replace this instruction with
    $(i,R,n_k^i,j)$ and add the instruction $(n_k^i,R,i,i)$ to $\M$,
  \item if $(i,R,j,k)\in \M$ then replace this instruction with
    $(i,R,j,n_k^i)$ and add the instruction $(n_k^i,R,i)$ to $\M$.
\end{itemize}
The new instructions cause $\M$ to loop upon entering state $n_k^i$. Since
$k\not\rightsquigarrow 0$, modifying $\M$ in this manner does not change its
halting status. After performing this procedure for all $i,k$ as described
above, we next eliminate any states which are not reachable from the initial
state 1. After performing this procedure, we will have $1\rightsquigarrow k
\rightsquigarrow 0$ for all states $k\in \StateGraph(\M)$.

Summarizing, we assume the following about every Minsky machine $\M$ we
consider in this paper:
\begin{itemize}
  \item $\M$ returns both registers to $0$ before halting,

  \item $\M$ has exactly one instruction for each state $k$,

  \item $\M$ begins with an instruction of the form $(1,R,s)$, and

  \item for every state $k$ of $\M$, there are paths in the state graph
    leading from the initial state $1$ to $k$, and from $k$ to the halting
    state $0$.
\end{itemize}
By the discussion in the paragraphs above, the halting problem restricted to
the set of Minsky machines satisfying these is still undecidable.

\section{Algebraic background and notation} \label{sec:background} 
In this section we give a brief background of the algebraic notions used in
the proof. Good references for additional details are McKenzie, McNulty,
Taylor~\cite{MMT_ALVBook} and Burris~\cite{Burris_CourseUnivAlg}.

An \emph{algebra} $\A$ consists of a non-empty set $A$, called the
\emph{universe} of $\A$, and a set of operations $\cl{F}$ on $A$, called the
\emph{fundamental operations} of $\A$. This is typically shortened to $\A =
\left< A; \cl{F} \right>$. From the operations in $\cl{F}$ we can generate
new operations by composition and variable identification. These together
with the projections are the \emph{term operations} of $\A$. A subset
$B\subseteq A$ which is closed under all operations from $\cl{F}$ is called
a \emph{subuniverse}. If $B \neq\emptyset$ then $B$ together with the
operations from $\cl{F}$ restricted to $B$ form a \emph{subalgebra} of $\A$,
written $\m{B}\leq \A$.

The operations of $\A$ extend coordinate-wise to operations of $\A^m$ for any
$m\in \NN$. A subuniverse $C\subseteq \A^m$ is called a \emph{relation} (or
\emph{subpower}) of $\A$. If $D\subseteq A^m$ is a subset then the smallest
relation containing $D$ is called the \emph{subalgebra generated by $D$},
written $\Sg_{\A^m}(D)$. We denote by $\Rel(\A)$ the set of all finitary
relations of $\A$. This set is closed under intersection (of equal arity
relations), product, permutation of coordinates, and projection onto a subset of
coordinates. Another way of saying this is that if relations are viewed as
predicates and $p(x_1,\dots,x_n)$ is a primitive-positive formula in the
language of these predicates then the set of values in $A^n$ for which $p$ is
true forms a relation.

From the discussion above, it is clear that the operations of $\A$ determine
the relations. The opposite is also true for finite $A$: $t$ is a term
operation of $\A$ if and only if it preserves all the relations of $\A$. We
can formalize this by introducing two new operations on sets of term
operations and relations. Let $E$ be a domain, $\cl{G}$ be a set of
operations on $E$, and $\cl{R}$ a set of subsets of powers of $E$. Define
\begin{align*}
  \Rel(\cl{G})
    &= \bigcup_{m\in \NN} \Big\{ R \subseteq E^m \mid R \text{ is closed
      under each operation from $\cl{G}$} \Big\}
    && \text{and}\\
  \Pol(\cl{R})
    &= \bigcup_{m\in \NN} \Big\{ f: E^m \to E \mid f \text{ preserves each
      set in $\cl{R}$} \Big\}.
\end{align*}
These two operations form a Galois connection:
\[
  \cl{R}\subseteq \Rel(\cl{G})
  \qquad\text{if and only if}\qquad
  \cl{G} \subseteq \Pol(\cl{R}).
\]
This relationship is quite famous and was first discovered by
Geiger~\cite{Geiger_RelnConstructs} and by Bodnar\v{c}uk, Kalu\v{z}nin,
Kotov, and Romov~\cite{BodnarcukEtal_GaloisPostAlgebras}. Every Galois
connection defines two closure operators. For $\Rel$ and $\Pol$ these are
$\Clo = \Pol \circ \Rel$ and $\RClo = \Rel \circ \Pol$, the \emph{clone} and
\emph{relational clone}, respectively. If $\A$ is an algebra with
fundamental operations $\cl{F}$ then the set of term operations of $\A$ is
$\Clo(\cl{F})$ and the set of relations of $\A$ is $\Rel(\cl{F}) =
\Rel(\A)$. 

If $\cl{R}$ is a set of relations on $\A$ and $\RR\in \RClo(\cl{R})$ (that
is, $R$ is preserved by every operation which preserves relations in
$\cl{R}$) then we say that $\cl{R}$ \emph{entails} $\RR$ and write
$\cl{R}\Entails \RR$. It is not difficult to prove that $\cl{R}\Entails \RR$
if and only if $\RR$ can be built from the relations in $\cl{R}\cup \{=\}$,
in finitely many steps, by applying the following constructions:
\begin{enumerate}
  \item intersection of equal arity relations,
  \item (cartesian) product of finitely many relations,
  \item permutation of the coordinates of a relation, and
  \item projection of a relation onto a subset of coordinates.
\end{enumerate}
We call these \emph{entailment constructions}. Similarly, for an operation
$f$ on $A$ we write $\cl{R}\Entails f$ if $f\in \Pol(\cl{R})$. We define the
\emph{degree} of $\cl{R}$ to be the supremum of the arities of the relations
in $\cl{R}$,
\[
  \deg(\cl{R})
  = \sup \big\{ \text{arity}(\SS) \mid \SS\in \cl{R} \big\}.
\]
For a clone $\cl{C}$, we define the degree to be the infimum of the degrees
of all sets of relations which determine $\Rel(\cl{C})$,
\[
  \deg(\cl{C})
  = \inf \big\{ \deg(\cl{R}) \mid \cl{R} \Entails \Rel(\cl{C}) \big\}
  = \inf \big\{ \deg(\cl{R}) \mid \Pol(\cl{R}) = \cl{C} \big\}.
\]
Finally, for an algebra $\A$ we define the degree of $\A$ to be the degree
of its clone,
\[
  \deg(\A)
  = \deg(\Clo(\A)).
\]
In general, any of these quantities may be infinite. An algebra $\A$ has
\emph{finite degree} (or is said to be \emph{finitely related}) if $\deg(\A)
< \infty$.

Lastly, we adopt a convention for projections of elements and subsets of
powers intended to increase readability. If $m\in \NN$ and $I\subseteq [m]$
then
\begin{itemize}
  \item for $B\subseteq A$, define $a^{-1}(B) = \{ i\in [m] \mid a(i)\in B
    \}$ and for $b\in A$ define $a^{-1}(b) = a^{-1}(\{b\})$,

  \item denote the projection of $a\in A^m$ to coordinates $I$ by $a(I)\in
    A^I$,

  \item denote the projection of $S\subseteq A^m$ to coordinates $I$ by
    $S(I)\subseteq A^I$, and

  \item define $a(\neq i) = a([m]\setminus i)$ and likewise $a(\neq i,j) =
    a([m]\setminus \{i,j\})$.
\end{itemize}
It is possible to confuse this notation for projection with the notation for
function application, but we will take special care to avoid ambiguous
situations.

\section{The algebra \texorpdfstring{$\AAM$}{A(M)}} \label{sec:AM_defn}  
We begin by defining the underlying set of $\AAM$. Let $\M$ be a Minsky
machine with states $\{0, \dots, N\}$ and define
\[
  M_i = \big\{ \vect{i,c} \mid c\in \{\Dot, \Er,0, A, B\} \big\}
  \qquad\text{and}\qquad
  \AM = \bigcup_{i=0}^N M_i.
\]
We next define several important subsets of $\AM$. Let
\begin{align*}
  & X = \big\{ \left< 0,\Er\right>, \dots, \left<N,\Er\right> \big\},
    && Y = \AM\setminus X, \\
  & D = \big\{ \left< 0,\Dot\right>, \dots, \left<N,\Dot\right> \big\},
    && E = \AM\setminus D, \\
  & C = \AM\setminus ( X \cup D ).
\end{align*}
An easy way to keep these straight is that $X$ contains elements with
second coordinate $\Er$, $D$ contains elements with second coordinate $\Dot$
(``dot''), and $C$ contains elements with neither. The set $Y$ is ``not
$X$'' and $E$ is ``not $D$''. We will now define the operations of $\AAM$.
It will be convenient in the operation definitions which follow to make use
three ``helper'' functions which are \emph{not} operations of $\AAM$. Let
\[
  \X\big( \vect{i,c} \big) = \vect{i,\Er},
  \qquad
  \State\big( \vect{i,c} \big) = i,
  \qquad\text{and}\qquad
  \Content\big( \vect{i,c} \big) = c.
\]
The second two of these are referred to as the \emph{state} and
\emph{content} of an element. We extend both of these functions to elements
of $\AM^m$ in different ways: for $m>1$ and $\alpha\in \AM^m$ define
\[
  \State(\alpha)
  = \big( \State(\alpha(1)), \dots, \State(\alpha(m)) \big)
  \qquad\text{and}\qquad
  \Content(\alpha) = \big\{ \Content(\alpha(i)) \mid i\in [m] \big\}.
\]
The algebra has a semilattice reduct with meet defined as
\[
  \vect{i, c} \wedge \vect{j, d}
  = \begin{cases}
    \vect{i, c}           & \text{if } \vect{i,c} = \vect{j,d}, \\
    \vect{\min(i,j), \Er} & \text{otherwise}.
  \end{cases}
\]
The semilattice operation defines an order: we write $x\leq y$ if and only
if $x\wedge y = x$. The next two operations encode the computation of $\M$
on elements of powers of $\AAM$. Let
\[
  M(x,y) = \begin{cases}
    \vect{j,R}    & \text{if } x = \vect{i,\Dot},\ y = \vect{i,0},\ 
      (i,R,j)\in \M, \\
    \vect{j,0}    & \text{if } x = \vect{i,\Dot},\ y = \vect{i,R},\ 
      (i,R,k,j)\in \M, \\
    \vect{j,\Dot} & \text{if } x = \vect{i,0},\ y = \vect{i,\Dot},\ 
      (i,R,j)\in \M, \\
    \vect{j,\Dot} & \text{if } x = \vect{i,R},\ y = \vect{i,\Dot},\ 
      (i,R,k,j)\in \M, \\
      \vect{j,c}    & \text{if } x = y = \vect{i,c},\ c\neq \Dot,\ 
      \Big[ (i,R,j)\in \M \text{ or } (i,R,k,j)\in \M \Big], \\
    \vect{j,\Er}  & \text{elif } \State(x) = \State(y) = i,\ 
      \Big[ (i,R,j)\in \M \text{ or } (i,R,k,j)\in \M \Big], \\
    \X(y) & \text{otherwise},
  \end{cases}
\]
and
\[
  M'(x) = \begin{cases}
    \vect{k,c}   & \text{if } x = \vect{i,c},\ (i,R,k,j)\in \M,\ c\neq R, \\
    \vect{k,\Er} & \text{elif } \State(x) = i,\ (i,R,k,j)\in \M, \\
    \X(x)        & \text{otherwise.}
  \end{cases}
\]
The next operations are involved with the representation of initial and
halting states of $\M$ in $\AAM$. Define
\[
  I(x,y) = \begin{cases}
    \vect{1,\Dot} & \text{if } x\in D, \\
    \vect{1,0}    & \text{elif } y\in C, \\
    \vect{1,\Er}  & \text{otherwise},
  \end{cases}
  \qquad\qquad
  H(x) = \begin{cases}
    \vect{0,0} & \text{if } x \in \big\{ \vect{0,0}, \vect{0,\Dot} \big\}, \\
    \vect{0,\Er} & \text{otherwise}.
  \end{cases}
\]
The next several operations are technical, but are intimately involved in
entailment and enforce a certain regularity on the structure of subpowers of
$\AAM$. Let
\begin{align*}
  & N_0(x,y,z) = \begin{cases}
    y             & \text{if } x = \vect{0,\Dot},\ \State(y) = \State(z), \\
    z             & \text{if } x = \vect{0,0},\ z\not\in D,\ 
                      \State(y) = \State(z), \\
    \X(y\wedge z) & \text{otherwise},
  \end{cases} \\
  & S(x,y,z) = \begin{cases}
    \left<1, 0\right>  & \text{if } x = \vect{1,0},\ y,z\in M_1, \\
      & \qquad \text{and } (\Content(y), \Content(z))\in \big\{ (\Dot,0), (0,\Dot), (0,0) \big\}, \\
    \left<1,\Er\right> & \text{otherwise},
  \end{cases} \\
  & N_\Dot(u,x,y,z) = \begin{cases}
    x & \text{if } x = y \not\in X,\ \State(x) = \State(y) = \State(z), \\
    x & \text{if } u\in D,\ y\in X,\ \State(x) = \State(y) = \State(z), \\
    y & \text{elif } u\in D,\ x\in X,\ \State(x) = \State(y) = \State(z), \\
    z & \text{elif } u\in D,\ z\in \big\{ x, y \big\},\ 
        \State(x) = \State(y) = \State(z), \\
    \X(x\wedge y\wedge z) & \text{otherwise},
  \end{cases} \\
  & P(u,v,x,y) = \begin{cases}
    x & \text{if } \State(u) = \State(v), \\
    y & \text{otherwise}.
  \end{cases}
\end{align*}
The algebra $\AAM$ is
\[
  \AAM
  = \big< \AM; \wedge, M, M', I, H, N_0, S, N_\Dot, P \big>.
\]
This completes the definition of $\AAM$.

Each operation of $\AAM$ plays an important role in the argument, and each
has been defined to be as simple as possible. Though the argument is
technical, we now attempt to give rough description of the role that each
operation plays.
\begin{itemize}
  \item The semilattice operation $\wedge$ induces an order on the algebra
    that is ``flat'' modulo $X$. That is, if $a\wedge b\not\in X$ then $a =
    b$.

  \item The operations $M$, $M'$, and $I$ encode the computation of $\M$ in
    the relations of $\AAM$. See Example~\ref{ex:encoding}.

  \item The operation $S$ is technical, and is used to produce special term
    operations $z_i(x)$ used elsewhere in the argument. See
    Lemma~\ref{lem:basic_facts_AM}.

  \item The operations $H$ and $N_0$ are responsible for ensuring the
    entailment of certain relations when $\M$ halts. See
    Theorem~\ref{thm:coding}, Corollary~\ref{cor:Sn_halting}, and
    Theorem~\ref{thm:noncomp_halting_entailed}.

  \item The operations $N_{\Dot}$ and $P$ are responsible for the entailment
    of relations which are ``non-computational''. See
    Definition~\ref{defn:computational} and
    Theorem~\ref{thm:noncomp_halting_entailed}.
\end{itemize}

\section{The encoding of computation} \label{sec:encoding} 
In this section we build the tools necessary to prove that the relations of
$\AAM$ encode the computation of $\M$ in a ``faithful'' manner.

\begin{defn} \label{defn:X_abs} 
An $n$-ary operation $f$ of $\AAM$ is said to be \emph{$X$-absorbing} if
for all $a_1,\dots,a_n\in \AM$, if $a_i\in X$ for some $i$ then
$f(a_1,\dots,a_n)\in X$.
\end{defn}   
\begin{defn} \label{defn:synchronized}  
An element $s\in \AM^n$ is said to be \emph{synchronized} if
$\State(s(i))$ is constant over all $i\in [n]$; we refer to the common value
as $\State(s)$ or ``the state of $s$''. A subset $S\subseteq \AM^n$ is said
to be \emph{synchronized} if all of its elements are.
\end{defn}   
\begin{lem} \label{lem:basic_facts_AM}  
Each of the following hold for $\AAM$.
\begin{enumerate}
  \item The state map $\State$ is a homomorphism of $\AAM$. Equivalently,
    if $t$ is a term operation of $\AAM$ and $s_1,\dots,s_m\in \AM^n$ are
    synchronized then $t(s_1,\dots,s_m)$ is synchronized as well.

  \item $\wedge$, $M$, $M'$, $H$, and $S$ are $X$-absorbing.

  \item If $a\in \AM\setminus D = E$ then the unary function $I(a,x)$ is
    $X$-absorbing.

  \item For all states $i$ there is a term operation of $\AAM$ defined by
    \[
      z_i(x) = \begin{cases}
        \vect{i,0} & \text{if } x\in C, \\
        \vect{i,\Er} & \text{otherwise}.
      \end{cases}
    \]

  \item For all states $i$ there is a term $w_i(x)$ in the operations $\{M, M',
    H\}$ satisfying
    \[
      w_i\big( \vect{i,c} \big)
      = \begin{cases}
        \vect{0,0}   & \text{if } c = 0, \\
        \vect{0,\Er} & \text{otherwise},
      \end{cases}
    \]
    for all $c\neq \Dot$.

  \item Let $t(x,y)$ be a term in the operations $\{M,M'\}$ and suppose that
    $t(a,b) = c \not\in X$ for some $a,b,c\in \AM$. Then
    \[
      t\big( \vect{\State(a),0}, \vect{\State(a),0} \big)
      = \vect{\State(c),0}
    \]
    and if $t$ is non-trivial (i.e.\ not a projection) then $\State(a) =
    \State(b)$.
\end{enumerate}
\end{lem}
\begin{proof}
\Proofitem{(1)--(3)}
These can be proven by carefully examining the definitions of the operations
of $\AAM$.

\Proofitem{(4)}
If $(i,R,j)\in \M$ or $(i,R,k,j)\in \M$ then
\[
  M\big( \vect{i,0},\vect{i,0} \big)
  = \vect{j,0},
\]
and if $(i,R,k,j)\in \M$ then
\[
  M'\big( \vect{i,0} \big)
  = \vect{k,0}.
\]
From the assumptions at the end of Section~\ref{sec:Minsky_machines}, every
state can be reached from $1$ in the state graph. Hence, there is a way to
compose operations from $\{ M, M' \}$ to obtain a term operation $f$ such
that
\[
  f\big( \vect{1,0}, \dots, \vect{1,0} \big)
  = \vect{i,0}.
\]
Let $T(x) = S(I(x,x),I(x,x),I(x,x))$. It follows that $z_i(x) = f(T(x),
\dots, T(x))$.

\Proofitem{(5)}
In Section~\ref{sec:Minsky_machines} we assumed that for each state $i$ there is
a directed path in the state graph to the halting state $0$. Similarly to the
proof of item (4), for each state $i$ there is a term $f$ in the operations
$\{M, M'\}$ such that
\[
  f\big( \vect{i,0}, \dots, \vect{i,0} \big) = \vect{0,0}.
\]
From the definitions, for $c\neq \Dot$ we have
\[
  H\big( f\big( \vect{i,c}, \dots, \vect{i,c} \big) \big))
  = \begin{cases}
    \vect{0,0}   & \text{if } c = 0, \\
    \vect{0,\Er} & \text{otherwise}.
  \end{cases}
\]
Hence $w_i(x) = H(f(x,\dots,x))$ satisfies the claim in the Lemma.

\Proofitem{(6)}
The proof is by induction on the complexity of $t$. For the base case where
$t$ is a projection, the conclusion clearly holds. For the inductive step,
there are two cases: $t(x,y) = M(t_1(x,y), t_2(x,y))$ or $t(x,y) =
M'(t_1(x,y))$. Suppose that $t(x,y) = M(t_1(x,y), t_2(x,y))$. Let $t_1(a,b)
= c_1$ and $t_2(a,b) = c_2$. From the definition of $M$, if $t(a,b) = c
\not\in X$ then $c_1, c_2\not\in X$, $\State(c_1) = \State(c_2)$, and $\M$
has an instruction of the form $(\State(c_1),R,\State(c))$ or
$(\State(c_1),R,k,\State(c))$. By the inductive hypothesis, these
observations, and the definition of $M$, it follows that
\begin{align*}
  t\big( \vect{\State(a),0}, \vect{\State(a),0} \big)
  &= M\Big( t_1(\vect{\State(a),0}, \vect{\State(a),0}),
            t_2(\vect{\State(a),0}, \vect{\State(a),0}) \Big) \\
  &= M\big( \vect{\State(c_1),0}, \vect{\State(c_1),0} \big)
  = \vect{\State(c),0},
\end{align*}
as claimed. The case when $t(x,y) = M'(t_1(x,y))$ is similar.
\end{proof}  
\begin{defn} \label{defn:capacity}  
We say that the Minsky machine $\M$ has
\begin{itemize}
  \item \emph{$k$-step capacity} $C$ if 
  \[
    C
    \geq \max\big\{ \alpha+\beta \mid \M^n(1,0,0) = (i,\alpha,\beta) 
      \text{ for some } n\leq k \big\},
  \]

  \item \emph{capacity} $C$ if 
  \[
    C
    \geq \max\big\{ \alpha+\beta \mid \M^n(1,0,0) = (i,\alpha,\beta) 
      \text{ for some } n\in \NN \big\},
  \]

  \item and \emph{halts with capacity} $C$ if it has capacity $C$ and halts.
\end{itemize}
We say that the relation $\RR\leq\AAM^m$ has 
\begin{itemize}
  \item \emph{capacity} $C$ if
    \[
      \Big| \big\{ i \mid \exists r\in R\cap Y^m \ r(i)\in D \big\} \Big|
      > C
    \]
  \item and \emph{weak capacity} $C$ if
    \[
      \Big| \big\{ i \mid \exists r\in R\ r(i)\in D \big\} \Big| 
      > C.
    \]
\end{itemize}
We say that some elements $\sigma_1$, $\dots$, $\sigma_{C+1}\in R$
\emph{witness} $\RR$ having capacity (resp.\ weak capacity) $C$ if there are
distinct elements 
\[
  i_1,\dots,i_{C+1}
  \in \big\{ i \mid \exists r\in R\ r(i)\in D \big\}
\]
such that $\sigma_j(i_j)\in D$ and $\sigma_j\in R\cap Y^m$ (resp.\
$\sigma_j\in R$).
\end{defn}   

Observe that for each $m\in \NN$ the halting problem is decidable for Minsky
machines with capacity $m$ since there are a finite (though quite large)
number of configurations. At first glance, the definitions of capacity for
machines and relations seem to be at odds. As we will see, however, a
relation with capacity $C$ can encode any Minsky machine computation with
capacity $C$.

\begin{defn} \label{defn:computational} 
If the relation $\RR\leq \AAM^m$ is synchronized and
\[
  \big| r^{-1}(D) \big|
  = \big| \{ i\in [m] \mid r(i) \in D \} \big|
  \leq 1
\]
(i.e.\ $r$ does not have two coordinates with content $\Dot$) for all $r\in
R$ then we call $\RR$ \emph{computational}.
\end{defn}   
\begin{defn} \label{defn:halting} 
$R \subseteq A(M)^m$ is \emph{halting} if it contains an element $r\in R$
such that
\[
  r
  \in \big\{ \vect{0,0}, \vect{0,\Dot} \big\}^m \setminus \big\{ \vect{0,0} \big\}^m.
\]
Such an $r$ is called a \emph{halting vector} of $R$. If $R$ is not halting
then we say that $R$ is \emph{non-halting}.
\end{defn}   

The easiest way to see how relations of $\AAM$ encode computation is to work
through an example.

\begin{ex} \label{ex:encoding} 
Consider the Minsky machine
\[
  \M 
  = \Big\{
    (1,A,2),\ 
    (2,B,3),\ 
    (3,A,4,3),\ 
    (4,B,0,4)
  \Big\}.
\]
Recall that a configuration of $\M$ is a triple $(k,\alpha,\beta)$ where $k$
is a state, $\alpha\in \NN$ is the value of register $A$, and $\beta\in \NN$
is the value of register $B$. We regard $\M$ as an operation on the set of
configurations. Running through the computation on the initial configuration
$(1,0,0)$, we have the table below.
\[ \begin{array}{c|ccccccc}
  n           & 0       & 1       & 2       & 3       & 4       & 5       & 6       \\ \midrule
  \M^n(1,0,0) & (1,0,0) & (2,1,0) & (3,1,1) & (3,0,1) & (4,0,1) & (4,0,0) & (0,0,0)
\end{array} \]
(We assumed in Section~\ref{sec:Minsky_machines} that all $\M$ zeroed out
the registers on halting.) Let us see how this is encoded in $\AAM^3$. For
$i\in [3]$ define elements $\sigma_i\in \AM^3$ and a subalgebra $\SS\leq
\AAM^3$ by
\[
  \sigma_i(j) = \begin{cases}
    \vect{1,\Dot} & \text{if } i = j \\
    \vect{1,0} & \text{otherwise},
  \end{cases}
  \qquad\text{and}\qquad 
  \SS = \Sg_{\AAM^3}\big\{ \sigma_1, \sigma_2, \sigma_3 \big\}.
\]
The relation $\SS$ is computational and has capacity $2$. Let $s\in S\cap
Y^3$. The value $\State(s)$ will correspond to the state of the computation,
and the values
\[
  \alpha 
    = \big| s^{-1}\big( \vect{\State(s), A} \big) \big|
  \qquad\text{and}\qquad
  \beta 
    = \big| s^{-1}\big( \vect{\State(s), A} \big) \big|
\]
will correspond to the value of registers $A$ and $B$, respectively. Observe
that the elements $\sigma_1, \sigma_2, \sigma_3$ correspond to the
configuration $(1,0,0)$. We will need some notation. For $k$ a state and
distinct indices $i_1, i_2, i_3\in [3]$, define elements of $\AM^3$
\[
  \AMConfig{
    \underbrace{k}_{\State} 
    \mid \underbrace{i_1}_{\Dot}
    \mid \underbrace{i_2}_{A}
    \mid \underbrace{i_3}_{B}
  }(j) 
  = \begin{cases}
    \vect{k, \Dot} & \text{if } j = i_1, \\
    \vect{k, A} & \text{if } j = i_2, \\
    \vect{k, B} & \text{if } j = i_3, \\
    \vect{k, 0} & \text{otherwise}.
  \end{cases}
\]
Additionally, define
\[
  \AMConfig{ k \mid i_1 \mid i_2 \mid \emptyset }(j)
  = \begin{cases}
    \vect{k, \Dot} & \text{if } j = i_1, \\
    \vect{k, A} & \text{if } j = i_2, \\
    \vect{k, 0} & \text{otherwise},
  \end{cases}
\]
and define $\AMConfig{ k \mid i_1 \mid \emptyset \mid i_3 }$ and 
$\AMConfig{ k \mid i_1 \mid \emptyset \mid \emptyset }$ similarly. In the
computations to follow below, the coordinates $i,j,k\in [3]$ are all
distinct.

First, observe that $\AMConfig{1 \mid i \mid \emptyset \mid \emptyset} =
\sigma_i$. We have
\small \begin{align*}
  \AMConfig{2 \mid i \mid j \mid \emptyset} 
    &= M\Big( \AMConfig{1 \mid j \mid \emptyset \mid \emptyset},\ 
              \AMConfig{1 \mid i \mid \emptyset \mid \emptyset} \Big),
  \qquad \text{e.g.}, \\
  \AMConfig{2 \mid 3 \mid 2 \mid \emptyset}
    &= \pmat{ \vect{2,0}    \\
              \vect{2,A}    \\
              \vect{2,\Dot} }
    = M\pmat{ \vect{1,0},    & \vect{1,0}    \\
              \vect{1,\Dot}, & \vect{1,0}    \\
              \vect{1,0},    & \vect{1,\Dot} },
\end{align*} \normalsize
corresponding to the configuration $\M^1(1,0,0) = (2,1,0)$. Next,
\small \begin{align*}
  \AMConfig{ 3 \mid i \mid j \mid k } 
    &= M\Big( \AMConfig{ 2 \mid k \mid j \mid \emptyset },\ 
              \AMConfig{ 2 \mid i \mid j \mid \emptyset } \Big),
    \qquad \text{e.g.}, \\
  \AMConfig{ 3 \mid 1 \mid 3 \mid 2 }
    &= \pmat{ \vect{3,\Dot} \\
              \vect{3,B}    \\
              \vect{3,A}    }
    = M\pmat{ \vect{2,0},    & \vect{2,\Dot} \\
              \vect{2,\Dot}, & \vect{2,0}    \\
              \vect{2,A},    & \vect{2,A}    },
\end{align*} \normalsize
corresponding to the configuration $\M^2(1,0,0) = (3,1,1)$. Next,
\small \begin{align*}
  \AMConfig{ 3 \mid i \mid \emptyset \mid j } 
    &= M\Big( \AMConfig{ 3 \mid k \mid i \mid j },\ 
              \AMConfig{ 3 \mid i \mid k \mid j } \Big),
    \qquad \text{e.g.}, \\
  \AMConfig{ 3 \mid 2 \mid \emptyset \mid 1 }
    &= \pmat{ \vect{3,B}    \\
              \vect{3,\Dot} \\
              \vect{3,0}    }
    = M\pmat{ \vect{3,B},    & \vect{3,B}    \\
              \vect{3,A},    & \vect{3,\Dot} \\
              \vect{3,\Dot}, & \vect{3,A}    },
\end{align*} \normalsize
corresponding to the configuration $\M^3(1,0,0) = (3,0,1)$. Next,
\small \begin{align*}
  \AMConfig{ 4 \mid i \mid \emptyset \mid j } 
    &= M'\Big( \AMConfig{ 3 \mid i \mid \emptyset \mid j } \Big),
    \qquad \text{e.g.}, \\
  \AMConfig{ 4 \mid 3 \mid \emptyset \mid 2 }
    &= \pmat{ \vect{4,0}    \\
              \vect{4,B}    \\
              \vect{4,\Dot} }
    = M'\pmat{ \vect{3,0}    \\
               \vect{3,B}    \\
               \vect{3,\Dot} },
\end{align*} \normalsize
corresponding to the configuration $\M^4(1,0,0) = (4,0,1)$. Next,
\small \begin{align*}
  \AMConfig{ 4 \mid i \mid \emptyset \mid \emptyset } 
    &= M\Big( \AMConfig{ 4 \mid j \mid \emptyset \mid i },\ 
              \AMConfig{ 4 \mid i \mid \emptyset \mid j } \Big),
    \qquad \text{e.g.}, \\
  \AMConfig{ 4 \mid 1 \mid \emptyset \mid \emptyset }
    &= \pmat{ \vect{4,\Dot} \\
              \vect{4,0}    \\
              \vect{4,0}    }
    = M\pmat{ \vect{4,B},    & \vect{4,\Dot} \\
              \vect{4,\Dot}, & \vect{4,B}    \\
              \vect{4,0},    & \vect{4,0}    },
\end{align*} \normalsize
corresponding to the configuration $\M^5(1,0,0) = (4,0,0)$. Finally, we have
\small \begin{align*}
  \AMConfig{ 0 \mid i \mid \emptyset \mid \emptyset } 
    &= M'\Big( \AMConfig{ 4 \mid i \mid \emptyset \mid \emptyset } \Big),
    \qquad \text{e.g.}, \\
  \AMConfig{ 0 \mid 2 \mid \emptyset \mid \emptyset }
    &= \pmat{ \vect{0,0}    \\
              \vect{0,\Dot} \\
              \vect{0,0}    }
    = M'\pmat{ \vect{4,0}    \\
               \vect{4,\Dot} \\
               \vect{4,0}    },
\end{align*} \normalsize
corresponding to the halting configuration $\M^6(1,0,0) = (0,0,0)$.

Since $\SS$ can witness the halting of $\M$, the relation $\SS$ will have a
lot of ``non-computational'' vectors in $Y^3$. In general, if a relation
does \emph{not} witness the halting of $\M$ then this will not be the case.
\end{ex}   

Now that we have some intuition for how computation is encoded, let us
continue exploring the structure of the relations of $\AAM$.

\begin{lem} \label{lem:basic_facts_relns} 
Let $\RR\leq \AAM^m$ be a relation and let $a,b,c,d\in R$.
\begin{enumerate}
  \item If $\RR$ is computational non-halting then $N_0(a,b,c)\not\in Y^m$
    or $N_0(a,b,c) = c$.

  \item If $\RR$ is computational then $N_{\Dot}(a,b,c,d)\leq b$ or
    $N_{\Dot}(a,b,c,d)\leq c$.

  \item If $\RR$ is synchronized then $P(a,b,c,d)\in \{c,d\}$.

  \item $S(a,b,c)\leq I(a,a)$ and if $S(a,b,c)\not\in X^m$ then
    $S(a,b,c)\leq a$.

  \item $\ds{ \bigcap \big\{ \RR \mid \RR\leq \AAM^m \text{ synchronized} \big\}
      = \big\{ r\in X^m \mid r \text{ is synchronized} \big\} }$.
\end{enumerate}
\end{lem}
\begin{proof}
\Proofitem{(1)--(4)}
These items follow directly from the definitions. The most complicated one
is item (2), so we will leave the others to the reader. If $b, c, d$ do not
share the same state then $N_{\Dot}(a,b,c,d) = \X(b\wedge c\wedge d) \leq
b$. Assume now that $b,c,d$ share the same state. If
$\Dot\not\in\Content(a)$ then $N_\Dot(a,b,c,d) = b$ or $N_\Dot(a,b,c,d) =
\X(b\wedge c\wedge d)\leq b$, so also assume that there is $k$ with $a(k)\in
D$ and $a(\neq k)\in E^{m-1}$ (we use $R$ being computational here). Hence
$N_\Dot(a,b,c,d)(\neq k) \leq (b\wedge c)(\neq k)$, so we just need to show
that $N_\Dot(a,b,c,d)(k)$ is less than equal to $b$ or $c$. The
possibilities are
\[
  N_\Dot(a,b,c,d)(k) = \begin{cases}
    b = c & \text{if } b = c\not\in X, \\
    b & \text{if } c\in X, \\
    c & \text{if } b\in X, \\
    b = d & \text{if } b = d\not\in X, \\
    c = d & \text{if } c = d\not\in X, \\
    \X(b\wedge c\wedge d) & \text{otherwise}.
  \end{cases}
\]
In all cases we have $N_\Dot(a,b,c,d)(k)$ less than equal to $b$ or $c$, so
we are finished.

\Proofitem{(5)}
Let $r\in R$ and let $s = H(I(r,r))$. It is not hard to see that $s(i) =
\vect{0,\Er}$ for all $i$. We also have that $z_k(s)(i) = \vect{k,\Er}$ for
all $i$, from Lemma~\ref{lem:basic_facts_AM} item (4). The conclusion
follows immediately.
\end{proof}  
\begin{defn} \label{defn:Sm}  
Define elements $\sigma_i\in \AM^m$ for $i\in [m]$ by
\begin{equation}  \label{eqn:sigma}
  \sigma_i(j) = \begin{cases}
    \vect{1,\Dot} & \text{if } i = j, \\
    \vect{1,0} & \text{otherwise}.
  \end{cases}
\end{equation}
Let $\Sigma_m = \{ \sigma_1,\dots, \sigma_m \}$ and define the \emph{$m$-th
sequential relation} of $\AAM$ to be
\[
  \SS_m = \Sg_{\AAM^m}(\Sigma_m).
\]
\end{defn}   
\begin{defn} \label{defn:config}  
Let $k$ be a state of $\M$ and $\alpha,\beta,m\in \NN$ be such that
$\alpha+\beta < m$. Let $P_m$ be the set of permutations on $[m]$ and define
$\Config(k,\alpha,\beta)\subseteq \AM^m$ to be
\begin{multline*}
  \Config(k,\alpha,\beta) = \\
  \bigcup_{p\in P_m}
    \Bigg\{ p\Big( \big(
        \vect{k,\Dot},
        \underbrace{ \vect{k,A}, \dots, \vect{k,A} }_{\alpha},
        \underbrace{ \vect{k,B}, \dots, \vect{k,B} }_{\beta},
        \underbrace{ \vect{k,0}, \dots, \vect{k,0} }_{m-\alpha-\beta-1}
    \big) \Big) \Bigg\}
\end{multline*}
(the permutation $p$ acts on a tuple by permuting coordinates). This is the
set of vectors encoding the Minsky machine configuration $(k,\alpha,\beta)$.
\end{defn}   
\begin{lem} \label{lem:Sm}  
Let $\SS_m\leq \AAM^m$ be as in Definition~\ref{defn:Sm}.
\begin{enumerate}
  \item $\SS_m$ is computational and has capacity $m-1$.

  \item If $p$ is a permutation on $[m]$ then $p(S_m) = S_m$.

  \item $\Config(k,\alpha,\beta)\cap S_m\neq\emptyset$ if and only if
    $\Config(k,\alpha,\beta)\subseteq S_m$.

  \item If $S_m\cap C^m\neq \emptyset$ then $\SS_m$ is halting.
\end{enumerate}
\end{lem}
\begin{proof}
\Proofitem{(1)}
The generators $\Sigma_m$ are witnesses to $\SS_m$ having capacity $m-1$.
Furthermore, $\Sigma_m$ is synchronized, so by
Lemma~\ref{lem:basic_facts_AM} item (1), $\SS_m$ must be as well. Let $s\in
S_m$ be such that $|s^{-1}(D)| \geq 2$. Examining the operations of $\AAM$,
we can see that any such element must have been generated by elements of
$\Sigma_m$ with more than one coordinate in $D$. $\Sigma_m$ contains no such
vectors.

\Proofitem{(2), (3)}
The generators $\Sigma_m$ are closed under $p$, so $\SS_m$ must be as well.
Applying item (2) for all permutations of $[m]$ proves item (3).

\Proofitem{(4)}
The generating set $\Sigma_m$ contains no vectors in $C^m$. A careful
analysis of the operations of $\AAM$ shows that the least complexity term
operation generating an element in $S_m\cap C^m$ from $\Sigma_m$ is of the
form $H(t_1(\overline{\sigma}))$ or $N_0(t_1(\overline{\sigma}),
t_2(\overline{\sigma}), t_3(\overline{\sigma}))$, where
$t_1(\overline{\sigma})(\ell)\in D$ for some $\ell$. Looking at the
definitions, we can see that $t_1(\overline{\sigma})$ is a halting vector in
either of these cases.
\end{proof}  

We are now ready to prove the main result of this section, The Coding
Theorem, which proves that $\SS_m$ encodes the computation of $\M$.

\begin{thm}[The Coding Theorem] \label{thm:coding}  
Let $\M$ be a Minsky machine.
\begin{enumerate}
  \item If $\M^n(1,0,0) = (k,\alpha,\beta)$ and $\M$ has $n$-step capacity
    $m-1$ then $\Config(k,\alpha,\beta)\subseteq S_m$.

  \item If $\Config(k,\alpha,\beta)\subseteq S_m$ and $\M$ does not halt
    with capacity $m-1$ then for some $n$ we have $\M^n(1,0,0) =
    (k,\alpha,\beta)$ and $\M$ has $n$-step capacity $m-1$.
\end{enumerate}
\end{thm}
\begin{proof}
For the first item, we refer the reader to Example~\ref{ex:encoding} and
Lemma~\ref{lem:Sm}.

For the second item, suppose that $\Config(k,\alpha,\beta)\subseteq S_m$ and
$\M$ does not halt with capacity $m-1$. We will analyze the generation of
$\SS_m = \Sg_{\AAM^m}(\Sigma_m)$. Let $G_0 = \Sigma_m$ and
\[
  G_n
  = \Big\{ F(\overline{g}) \mid F \text{ a fundamental $\ell$-ary
      operation},\ \overline{g}\in G_{n-1}^\ell \Big\}
    \cup G_{n-1}.
\]
Observe that $S_m = \bigcup G_n$, so $\Config(k,\alpha,\beta)\cap G_n \neq
\emptyset$ for some least $n$. A key observation for what follows is that
since $\Sigma_m$ is closed under coordinate permutation, so is $G_n$, so
$\Config(k,\alpha,\beta)\cap G_n\neq \emptyset$ implies
$\Config(k,\alpha,\beta)\subseteq G_n$. After proving the next claim, we
will be done.

\begin{claim*}
If $n$ is minimal such that $\Config(k,\alpha,\beta)\subseteq G_n$ then
$\M^n(1,0,0) = (k,\alpha,\beta)$ and $\M$ has $n$-step capacity $m-1$.
\end{claim*}
\begin{claimproof}
The proof shall be by induction on $n$. Observe that $\Config(1,0,0) =
\Sigma_m = G_0$, $\M^0(1,0,0) = (1,0,0)$, and $\M$ has $0$-step capacity
$m-1\geq 0$. This establishes the basis of the induction.

Suppose now that $n>0$ and let $s\in \Config(k,\alpha,\beta)\subseteq G_n$.
This implies that $s = F(\overline{g})$ for some $\ell$-ary fundamental
operation $F$ and $\overline{g}\in G_{n-1}^{\ell}$. We break into cases
depending on which fundamental operation $F$ is.

\Case{$F\in \{\wedge, N_{\Dot}, P\}$}
These operations have the property that if $s = F(\overline{g})$ then $s\leq
g_i$ for some $g_i$ amongst the $\overline{g}$. Since $s\in Y^m$, this
implies $s = g_i$, so $s\in G_{n-1}$ and hence
$\Config(k,\alpha,\beta)\subseteq G_{n-1}$, a contradiction.

\Case{$F\in \{H, S\}$}
These operations have ranges contained entirely in $E$. Since $\Dot\in
\Content(s)$, $s$ cannot be the output of such an operation.

\Case{$F = I$}
If $s = I(a,b)$ then $s\in \Config(1,0,0)$, and we are back in the base
case.

\Case{$F = N_0$}
If $s = N_0(a,b,c)$ then $N_0(a,b,c)\in Y^m$. If $a$ is not a halting vector
then we have $s = c$ by Lemma~\ref{lem:basic_facts_relns} item (1), so $s$
and hence $\Config(k,\alpha,\beta)$ are contained in $G_{n-1}$, a
contradiction. If $a$ is a halting vector then from the definition of $N_0$
we have that $a\in \Config(0,0,0)$, so $\Config(0,0,0)\subseteq G_{n-1}$,
and by the inductive hypothesis we have that $\M^{n-1}(1,0,0) = (0,0,0)$.
Hence $\M$ halts in $n-1$ steps with capacity $m-1$, contradicting the
hypotheses.

\Case{$F \in \{M,M'\}$}
Let $s = M(a,b)$. If $a\in C^m$ then by Lemma~\ref{lem:Sm} item (4), we have
that $G_{n-1}$ contains a halting vector. This gives rise to a contradiction
as in the case for $F = N_0$. If $a\not\in C^m$ then since $s\in Y^m$ we
have that $a(\ell)\in D$ for some $\ell$, from the definition of $M$. Also
from the definition, there is some instruction $(i,R,k)\in \M$ or
$(i,R,j,k)\in \M$ such that $a,b\in \Config(i, \alpha+\epsilon, \beta+\tau)$
where the different possibilities for $(\epsilon,\tau)$ correspond to the
different possibilities for the instruction. In any case, by the inductive
hypothesis we have that $\M^{n-1}(1,0,0)=(i, \alpha+\epsilon, \beta+\tau)$
and $\M$ has $(n-1)$-step capacity $m-1$. We therefore have
\[
 \M^{n}(1,0,0)
 = \M(i,\alpha+\epsilon,\beta+\tau)
 = (k,\alpha,\beta).
\]
Since $\alpha + \beta \leq m-1$ and $\M$ has $(n-1)$-step capacity $m-1$, it
follows that $\M$ has $n$-step capacity $m-1$. The analysis for $M'$ is
similar.
\renewcommand{\qedsymbol}{\ensuremath{\circ} \quad \ensuremath{\square}}
\end{claimproof}
\let\qed\relax
\end{proof}  
\begin{cor} \label{cor:Sn_halting}  
The following are equivalent.
\begin{enumerate}
  \item $\M$ halts with capacity $m-1$,

  \item $\SS_m$ is halting,

  \item every computational $\RR\leq \AAM^\ell$ with capacity $m-1$ is
    halting.
\end{enumerate}
\end{cor}
\begin{proof}
We begin by proving the equivalence of the first two items. Suppose that
$\M$ halts with capacity $m-1$. By Theorem~\ref{thm:coding}, this implies
that $\Config(0,0,0)\subseteq S_m$ (recall that we assumed in
Section~\ref{sec:Minsky_machines} that $\M$ would zero the registers before
halting). Any element of $\Config(0,0,0)$ is a halting vector, so $\SS_m$ is
halting. For the converse, suppose that $\SS_m$ has a halting vector $s$. It
follows that $s\in \Config(0,0,0)$ and hence, by Lemma~\ref{lem:Sm} item
(3), that $\Config(0,0,0)\subseteq S_m$. Towards a contradiction assume that
$\M$ does not halt with capacity $m-1$. By Theorem~\ref{thm:coding}, for
some $n$ we have $\M^n(1,0,0) = (0,0,0)$ and $\M$ has $n$-step capacity
$m-1$. This is a contradiction.

We next prove the equivalence of items (2) and (3). Suppose that $\SS_m$ is
halting and that $\RR\leq\AAM^\ell$ has capacity $m$ witnessed by
$(\sigma_i')_{i\in \cl{I}}$ with $|\cl{I}| = m$, say $\sigma_i'(i)\in D$ and
$\sigma_i'(\neq i)\in C^{\ell-1}$. Let $\sigma_i = I(\sigma_i',\sigma_i')$
and observe that $(\sigma_i)_{i\in \cl{I}}$ satisfies
Equation~\eqref{eqn:sigma} from the definition of $\SS_m$:
\[ \tag{\ref{eqn:sigma} redux}
  \sigma_i(j) = \begin{cases}
    \vect{1,\Dot} & \text{if } i = j, \\
    \vect{1,0} & \text{otherwise}.
  \end{cases}
\]
Let $\SS = \Sg_{\AAM^\ell}( \{\sigma_i \mid i\in \cl{I}\} )$. We have that
$\SS(\cl{I}) = \SS_m$ and if $s\in S\cap Y^\ell$ then $\Content(s(i)) = 0$
for all $i\in [\ell]\setminus [m]$. Combining these yields a halting vector
for $S$. We have that $\SS\leq \RR$, so $\RR$ must also be halting. The
converse is clear since $\SS_m$ has capacity $m-1$, by Lemma~\ref{lem:Sm}
item (1).
\end{proof}  

If $\M$ halts then $\SS_m$ is halting for some $m$. Projecting on a
single coordinate, it follows that
\[
  \m{T} = \Sg_{\AAM}\big\{ \vect{1,\Dot}, \vect{1,0} \big\}
\]
is also halting (i.e.\ $\vect{0,\Dot} \in T$). Independent of the halting
status of $\M$, let us consider this relation. Whether or not $\m{T}$ is
halting is a decidable property. If $\m{T}$ is non-halting then it is not
possible for $\M$ to halt (the converse does not hold, of course).
\textbf{We therefore assume from this point onward that $\M$ is such that
$\m{T}$ is halting}.

\section{If \texorpdfstring{$\M$}{M} does not halt} \label{sec:M_doesnt_halt}  
Recall from Section~\ref{sec:background} that $\Rel(\AAM)$ is the set of all
finitary relations of $\AAM$. If $\cl{R}\subseteq \Rel(\AAM)$ is a set of
relations then we have $\cl{R}\Entails \RR$ if and only if $\RR$ can be
obtained from relations in $\cl{R}\cup \{ = \}$, in finitely many steps, by
applying the following constructions:
\begin{enumerate}
  \item intersection of equal arity relations,
  \item (cartesian) product of finitely many relations,
  \item permutation of the coordinates of a relation, and
  \item projection of a relation onto a subset of coordinates.
\end{enumerate}
A close analysis of various projections of relations is called for, so we
remind the reader of the convention for projections adopted in
Section~\ref{sec:background}: for $m\in \NN$ and $I\subseteq [m]$,
\begin{itemize}
  \item denote the projection of $a\in \AM^m$ to coordinates $I$ by $a(I)\in
    \AM^I$,

  \item denote the projection of $S\subseteq \AM^m$ to coordinates $I$ by
    $S(I)\subseteq \AM^I$, and

  \item define $a(\neq i) = a([m]\setminus i)$ and likewise $a(\neq i,j) =
    a([m]\setminus \{i,j\})$.
\end{itemize}
Finally, for $n\in \NN$ we define $\Rel_{\leq n}(\AAM)$ to be the set of at
most $n$-ary relations of $\AAM$.

The next theorem shows that the relations built using the entailment
constructions above must have a certain form. This theorem is essentially
Theorem 3.3 from Zadori~\cite{Zadori_DualityFinRelns} and we refer the
interested reader to that paper for the proof.

\begin{thm} \label{thm:entailment_form} 
Let $\A$ be an algebra and let $\cl{R}$ be a set of relations on
$\A$. Then $\cl{R}\Entails \SS$ if and only if
\[
  \SS
  = \pi \Bigg( \bigcap_{i\in I} \mu_i \Big( \prod_{j\in J_i} \RR_{ij} \Big) \Bigg)
\]
for some finite index sets $I$ and $(J_i)_{i\in I}$, where the $\RR_{ij}\in
\cl{R} \cup \{=\}$, $\pi$ is a coordinate projection, and the $\mu_i$ are
coordinate permutations.
\end{thm} 

We now take a close look at relations of this form.

\begin{lem} \label{lem:entailed_must_halt}  
Suppose that
\[
  \sigma_1,\dots,\sigma_m
  \in \pi \Bigg( \bigcap_{i\in I} \mu_i \Big( \prod_{j\in J_i} \RR_{ij} \Big) \Bigg)
  = \SS
  \leq \AAM^m,
\]
where $\sigma_1,\dots,\sigma_m$ are the generators of $\SS_m$
(cf.\ Definition~\ref{defn:Sm}), $\pi$ is a projection, the $\mu_i$ are
permutations, the $\RR_{ij}$ are a finite collection of members of
$\Rel_{\leq n}(\AAM)$, and $n < m$. Then $S\cap C^m\neq\emptyset$.
\end{lem}
\begin{proof}
We begin by establishing some notation. Let
\[
  \m{B}
  = \bigcap_{i\in I} \mu_i \Big( \prod_{j\in J_i} \RR_{ij} \Big)
  \leq \AAM^M.
\]
Without loss of generality assume that $P = [m]$ is the set of coordinates
that $\pi$ projections onto and let $Q = [M]\setminus P = \{m+1,\dots,M\}$.
Define $K_{ij}\subseteq [M]$ to be the coordinates of $\RR_{ij}$ in the
permuted product $\mu_i \big( \prod_{j\in J_i} \RR_{ij} \big)$. Let
$K_{ij}^P = K_{ij}\cap P$ and $K_{ij}^Q = K_{ij}\cap Q$. Observe that
\begin{itemize}
  \item $[M] = P\sqcup Q$ (the disjoint union),
  \item $|K_{ij}| \leq n < m$ and $K_{ij} = K_{ij}^P \sqcup K_{ij}^Q$ for
    all $i\in I$ and $j\in J_i$, and
  \item for every $i\in I$ and $j\in J_i$ we have $\ds{P = \bigsqcup_{i\in
        J_i} K_{ij}^P}$ and $\ds{Q = \bigsqcup_{i\in J_i} K_{ij}^Q}$.
\end{itemize}
We have $\sigma_1,\dots,\sigma_m\in S$, and since $B(P) = S$ there must be
elements $\tau_1,\dots,\tau_m\in B$ such that $\tau_\ell(P) = \sigma_\ell$
for all $\ell$. Take each $\tau_\ell$ to be minimal (under the semilattice
order) with this property and such that $\State(\tau_\ell) = 1$ (such
$\tau_\ell$ exist --- just use operation $I$).
\begin{claim} \label{claim:entailed_halt__tau_either_or}
Let $q\in Q$. Either 
\begin{enumerate}
  \item $\tau_\ell(q) = \tau_k(q)\in \big\{ \vect{1,0}, \vect{1,\Er}
    \big\}$ for all $\ell,k\in [m]$ or
  \item there is a unique $\ell\in [m]$ such that $\tau_\ell(q) =
    \vect{1,\Dot}$ and for all $k\in [m]\setminus \{\ell\}$,  $\tau_k(q) =
    \vect{1,0}$.
\end{enumerate}
\end{claim}
\begin{claimproof}
Observe that if $\Content(\tau_\ell(q)) \in \{ A, B \}$ then the element
\[
  \tau_\ell \wedge I(\tau_\ell, \tau_\ell)
\]
will be properly less than $\tau_\ell$ while still having projection on $P$
to $\sigma_\ell$, contradicting the minimality of $\tau_\ell$. Hence it must
be that $\Content(\tau_\ell(q))\in \{\Dot, 0, \Er\}$.

For all distinct $k,\ell\in [m]$ define
\[
  t'_{k\ell}
  = M(\tau_k, \tau_\ell),
  \qquad
  t_{k\ell}
  = I(t'_{k\ell}, t'_{k\ell}),
  \qquad\text{and}\qquad
  t_\ell
  = \bigwedge_{k\in [m]\setminus \{\ell\}} t_{k\ell}.
\]
The machine $\M$ begins with an instruction of the form $(1,R,s)$. From the
definition of $M$ and $I$ it therefore follows that $t_{k\ell}(P) =
\sigma_\ell$ for all $k\in [m]$. Thus $t_\ell(P) = \sigma_\ell$ and hence
(by minimality) $t_\ell(q) = \tau_\ell(q)$. Let us fix an $\ell\in [m]$ to
consider. By the observation at the end of the previous paragraph
$\Content(t_\ell(q))\in \{\Dot, 0, \Er\}$, giving us three cases to examine.

\Case{$t_\ell(q) = \vect{1,\Dot}$}
Note that $\tau_\ell(q) = t_\ell(q)$. Suppose towards a contradiction that
$\tau_k(q)\neq \vect{1,0}$ for some $k\in[m]\setminus \{\ell\}$. It follows that
$\Content(\tau_k(q))\in \{\Dot, \Er\}$, so we have
\[
  t'_{k\ell}(q) 
  = M(\tau_k, \tau_\ell)(q)
  = \left\{ \begin{aligned}
    & M\big( \vect{1,\Dot}, \vect{1,\Dot} \big) & \text{if } \Content(\tau_k(q)) = \Dot, \\
    & M\big( \vect{1,\Er}, \vect{1,\Dot} \big)  & \text{if } \Content(\tau_k(q)) = \Er
  \end{aligned} \right\}
  = \vect{s, \Er}
\]
for some state $s$. This yields $t_\ell(q) = \vect{1,\Er}$, which is a
contradiction since $t_\ell(q) = \vect{1,\Dot}$. Hence $\tau_k(q) = \vect{1,0}$
for all $k\in [m]\setminus \{\ell\}$, which is item (2) from the claim.

\Case{$t_\ell(q) = \vect{1,0}$}
This implies that $t_{k\ell}(q) = \vect{1,0}$ for all $k$, and so
$\Content(\tau_k(q))\in \{0,\Dot\}$ for all $k$ (from the definition of
$M$). If there were two distinct $k_1, k_2\in [m]$ such that $\tau_{k_1}(q)
= \tau_{k_2}(q) = \vect{1,\Dot}$, then we would have (as in the previous case
above) that $t_{k_2k_1}(q) = \vect{1,\Er}$ and hence $t_{k_1}(q) =
\vect{1,\Er}$. This contradicts the observation from the start of this case that
$\Content(\tau_k(q))\in \{0,\Dot\}$ for all $k$ since $\tau_{k_1}(q) =
t_{k_1}(q)$. It follows that there is at most one $k$ such that $\tau_k(q) =
\vect{1,\Dot}$ and that for all other $k'\neq k$ we have $\tau_{k'}(q) =
\vect{1,0}$. This is either item (1) or (2) of the claim.

\Case{$t_\ell(q) = \vect{1,\Er}$}
In this case, by the minimality of $\tau_\ell$ we have $\tau_\ell(q) =
\vect{1,\Er}$. It follows from the definition of $M$ that for all $k \neq
\ell$ we have $t_{\ell k}(q) = \vect{1,\Er}$ (note the order of subscripts),
and so $t_k(q) = \vect{1,\Er}$ for all $k$. Using minimality again yields
$\tau_k(q) = \vect{1,\Er}$ for all $k$, giving us item (1) of the claim.
\end{claimproof}

From Claim~\ref{claim:entailed_halt__tau_either_or} above, we can partition
$Q$ into two pieces,
\begin{align*}
  Q_= 
  &= \Big\{ q\in Q \mid 
      \text{Claim~\ref{claim:entailed_halt__tau_either_or} item (1) holds} 
    \Big\}
  & \text{and} \\
  Q_{\neq}
  &= \Big\{ q\in Q \mid 
      \text{Claim~\ref{claim:entailed_halt__tau_either_or} item (2) holds}
    \Big\}.
\end{align*}
Let $K_{ij}^{Q_{\neq}} = K_{ij}\cap Q_{\neq}$. Fix an $i\in I$ and for each
$j\in J_i$ choose an $\ell_j\in [m]\setminus K_{ij}^P$ (such $\ell_j$ exist
for all $j$ since $|K_{ij}| < m$). Let $L = \big\{ \ell_j \mid j\in J_i
\big\}$ be the set of these choices.

Observe that for any sets $Z_1, Z_2$ and any $a,b\in Z_1\times Z_2$ there is
an element $c\in Z_1\times Z_2$ with $c(1) = a(1)$ and $c(2) = b(2)$.
Applying this observation to the elements $\{ \tau_{\ell_j} \mid \ell_j\in L
\} \subseteq \mu_i\big( \prod_{j\in J_i} \RR_{ij} \big)$ yields an element
$\alpha_i^L\in \mu_i\big( \prod_{j\in J_i} \RR_{ij} \big)$ such that
$\alpha_i^L(K_{ij}) = \tau_{\ell_j}(K_{ij})$ for all $j\in J_i$. Expanding
upon this, we have
\begin{itemize}
  \item $\alpha_i^L(K_{ij}^Q) = \tau_{\ell_j}(K_{ij}^Q)$ for all $j \in J_i$ and

  \item for all $p \in P$ and for the unique $j\in J_i$ such that $p\in
    K_{ij}$,
    \[
     \alpha_i^L(p)
     = \tau_{\ell_j}(p)
     = \sigma_{\ell_j}(p)
     = \vect{1,0}
    \]
    (this follows from $\ell_j\in [m]\setminus K_{ij}^P$).
\end{itemize}
It follows from this that $\alpha_i^L(P)\in C^m$, so to prove the lemma it
suffices to show that there is some system of choices $(L_i)_{i\in I}$ such
that for all $i,i' \in I$ we have $\alpha_i^{L_i} = \alpha_{i'}^{L_{i'}}$.
That is, the element $\alpha_i^{L_i}$ does not depend on $i$ and thus lies
in the intersection $\bigcap_{i\in I} \mu_i\big( \prod_{j\in J_i} \RR_{ij}
\big)$.

\begin{claim} \label{claim:entailed_halt__Qeq}
Fix an $i\in I$. For all $q\in Q_=$, all choices of $L$ as above, and all
$\ell\in [m]$, we have $\alpha_i^L(q) = \tau_\ell(q)\in \big\{ \vect{1,0},
\vect{1,\Er} \big\}$.
\end{claim}
\begin{claimproof}
For each $q\in Q_=$ there is a unique $j$ such that $q\in K_{ij}^Q$, so by
the construction of $\alpha_i^L$ we have $\alpha_i^L(q) = \tau_{\ell_j}(q)$
for some unique $\ell_j\in L$. Since $q\in Q_=$, by
Claim~\ref{claim:entailed_halt__tau_either_or} the conclusion follows.
\end{claimproof}

By Claim~\ref{claim:entailed_halt__tau_either_or}, for each $q\in Q_{\neq}$
there is a unique $k_q\in [m]$ such that $\tau_{k_q}(q) = \vect{1,\Dot}$. It
follows that $Q_{\neq}$ can be partitioned,
\[
  Q_{\neq}
    = \bigsqcup_{k\in [m]} Q_{\neq}^k
  \qquad\text{where}\qquad
  Q_{\neq}^k
    = \big\{ q\in Q_{\neq} \mid \tau_k(q) = \vect{1,\Dot} \big\}.
\]

\begin{claim} \label{claim:entailed_halt__Qneq}
Fix an $i\in I$. For all $k\in [m]$ and all $j\in J_i$ such that $Q_{\neq}^k\cap
K_{ij} \neq \emptyset$ there exists some $k\in [m]\setminus K_{ij}^P$ such that
\[
  \Content\big( \tau_k(K_{ij}^P \cup K_{ij}^{Q_{\neq}}) \big)
  = \{0\}.
\]
\end{claim}
\begin{claimproof}
Let $K = [m]\setminus K_{ij}^P$ and observe that for every $k\in K$ we have
$\Content\big( \tau_k(K_{ij}^P) \big) = \{0\}$ and $\Content\big(
\tau_k(K_{ij}^{Q_{\neq}}) \big) \subseteq \{ \Dot, 0\}$ by
Claims~\ref{claim:entailed_halt__tau_either_or} and
\ref{claim:entailed_halt__Qeq}. As mentioned before the statement of the
claim, for each $q\in Q_{\neq}$ there is a unique $k_q\in [m]$ such that
$\tau_{k_q}(q) = \vect{1,\Dot}$ and $\tau_k(q) = \vect{1,0}$ for all $k\neq
k_q$. Towards a contradiction, let us assume that for all $k\in K$ we have
$\Dot\in \Content\big( \tau_k(K_{ij}^{Q_{\neq}}) \big)$. It follows that
\[
  K
  = \Big\{ k_q \mid q\in K_{ij}^{Q_{\neq}} \Big\}.
\]
We therefore have $|K| = m - |K_{ij}^P|$ (from the start of the proof of the
claim) and
$|K|\leq |K_{ij}^{Q_{\neq}}|$. Hence $m - |K_{ij}^P| \leq
|K_{ij}^{Q_{\neq}}|$, so
\[
  m
  \leq \big| K_{ij}^P \big| + \big| K_{ij}^{Q_{\neq}} \big|
  \leq \big| K_{ij}^P \big| + \big| K_{ij}^Q \big|
  = \big| K_{ij} \big|
  \leq n
  < m,
\]
a contradiction.
\end{claimproof}

Consider $\alpha_i^L$ for some fixed $i\in I$ and fixed $L$. Suppose that
for some $\ell_h\in L$ we have $\Dot\in \Content\big(
\alpha_i^L(K_{ih}^{Q_{\neq}}) \big)$. The set $Q_{\neq}^{\ell_h}$ has a
covering
\[
  Q_{\neq}^{\ell_h}
  \subseteq \bigsqcup_{j\in J_i^h} K_{ij}^{Q_{\neq}}
  \qquad\text{where}\qquad 
  J_i^h 
    = \big\{ j\in J_i \mid Q_{\neq}^{\ell_h}\cap K_{ij}^{Q_{\neq}} \neq\emptyset \big\}.
\]
For each $K_{ij}^{Q_{\neq}}$ in this covering, replace $\ell_j$ in $L$ with
some $k_j$ satisfying the conclusion of the
Claim~\ref{claim:entailed_halt__Qneq}. $\ell_h$ will be replaced in this
process, along with possibly others. After this replacement, the number of
$\ell_k\in L$ such that $\Dot\in \Content\big( \alpha_i^L(K_{ik}^{Q_{\neq}})
\big)$ will have decreased by Claim~\ref{claim:entailed_halt__tau_either_or}
and the construction of $Q_{\neq}$ and $L$. 

Repeat the above procedure on the newly obtained $\alpha_i^L$ until
$\Dot\not\in \Content\big( \alpha_i^L(Q_{\neq}) \big)$ and call the final
result $\alpha_i$. For a fixed $i$, we thus have constructed an element
$\alpha_i$ such that
\begin{itemize}
  \item $\alpha_i(p) = \vect{1,0}$ for all $p\in P$,
  \item $\alpha_i(Q_=) = \tau_1(Q_=) = \cdots = \tau_m(Q_=)$ (by
    Claim~\ref{claim:entailed_halt__Qeq}), and
  \item $\alpha_i(q) = \vect{1,0}$ for all $q\in Q_{\neq}$ (by
    Claim~\ref{claim:entailed_halt__Qneq} and construction).
\end{itemize}
The description of $\alpha_i$ above does not depend on $i$, so $\alpha_i$ is
a common element in the intersection $\bigcap_{i\in I} \mu_i\big(
\prod_{j\in J_i} \RR_{ij} \big)$. It follows that $\alpha_i(P)\in S\cap
C^m$.
\end{proof}  

\begin{thm} \label{thm:M_not_halt_high_degree} 
The following hold for any Minsky machine $\M$.
\begin{enumerate}
  \item If $\M$ does not halt with capacity $m$ then $m < \deg(\AAM)$.
  \item If $\M$ does not halt then $\AAM$ is not finitely related.
\end{enumerate}
\end{thm}
\begin{proof}
For item (1), suppose that $\deg(\AAM) \leq m$. This implies in particular
that
\[
  \Rel_{\leq m}(\AAM) \Entails \SS_{m+1}.
\]
By Theorem~\ref{thm:entailment_form} there is some projection $\pi$,
permutations $\mu_i$, and a finite collection of relations $\RR_{ij}\in
\Rel_{\leq m}(\AAM)$ such that
\[
  \SS_{m+1}
  = \pi \Bigg( \bigcap_{i\in I} \mu_i \Big( \prod_{j\in J_i} \RR_{ij} \Big) \Bigg).
\]
By Lemma~\ref{lem:entailed_must_halt}, this implies that $S_{m+1}\cap
C^{m+1}\neq \emptyset$, and by Lemma~\ref{lem:Sm} and
Corollary~\ref{cor:Sn_halting}, this implies that $\M$ halts with capacity
$m$, a contradiction. Item (2) follows from item (1).
\end{proof}  

\section{If \texorpdfstring{$\M$}{M} halts --- tools}\label{sec:M_halts_tools} 
The argument showing that $\AAM$ is finitely related when $\M$ halts is
quite long and intricate. This section develops the necessary machinery.
Throughout this section and the next (Section~\ref{sec:M_halts_entailment}),
we assume that $\M$ halts with capacity $\kappa$. We begin by highlighting
some important relations of $\AAM$.

The strategy for the main proof is to show that for some suitably chosen
$k$, we have $\Rel_{\leq k}(\AAM) \Entails \Rel_{\leq n}(\AAM)$ for all $n$.
We therefore consider an arbitrary $m$-ary operation $f$ which preserves
$\Rel_{\leq k}(\AAM)$, arbitrary $\RR\leq \AAM^n$, and arbitrary
$r_1,\dots,r_m\in R$ and endeavor to show that $f(r_1,\dots, r_m)\in R$. The
relations which we define below will play an important role in analyzing the
behavior of $f$ on $R$, and following each definition we attempt to give the
reader some intuition for how they can be used.

\begin{defn} \label{defn:mu_chi}  
Let
\begin{align*}
  &\mu = \left\{
      \pmat{ a \\ a }, \pmat{ a \\ \X(a) } \mid a\in E
    \right\}
    \subseteq \AM^2, \\[1em]
  &\chi = \left\{ 
      \pmat{ a_1 \\ a_2 \\ a_2 }, \pmat{ \X(a_1) \\ a_2 \\ \X(a_2) } 
      \mid (a_1,a_2)\in E^2 \text{ synchronized}
    \right\}
    \subseteq \AM^3.
\end{align*}
Operations which preserve $\mu$ are monotone on $E$ (see
Lemma~\ref{lem:f_E_monotone}). The property that $\chi$ describes is more
subtle. Let $f$ be an operation and consider an evaluation of the form
\[
  f\pmat{ a_1, & \cdots & \X(a_k), & \cdots & a_m \\
          b_1, & \cdots & b_k,     & \cdots & b_m }
    = \pmat{ \alpha_1 \\
             \alpha_2 }
\]
where each $(a_i,b_i)\in E^2$ is synchronized and $\alpha_1\not\in X$. If
$f$ preserves $\mu$ and $\chi$ then we can conclude that replacing $b_k$
with $\X(b_k)$ in the second line of input does not change the output of
$f$:
\[
  f\pmat{ a_1, & \cdots & \X(a_k), & \cdots & a_m \\
          b_1, & \cdots & b_k,     & \cdots & b_m \\
          b_1, & \cdots & \X(b_k), & \cdots & b_m }
    = \pmat{ \alpha_1 \\
             \alpha_2 \\
             \alpha_2 }
\]
(the input vectors are elements of $\chi$, so the output is in $\chi$ as
well). The details of this are contained in
Lemma~\ref{lem:not_chi_entailed_helper}.
\end{defn}   
\begin{defn} \label{defn:deltas}  
Define three subsets of $\AM^3$,
\begin{align*}
  &\Delta_{\forall}
    = \left\{
      \pmat{ z \\ z \\ z },
      \pmat{ z \\ a \\ z },
      \pmat{ z \\ b \\ z },
      \pmat{ a \\ z \\ z },
      \pmat{ a \\ a \\ a },
      \pmat{ a \\ b \\ z },
      \pmat{ b \\ z \\ z },
      \pmat{ b \\ a \\ z },
      \pmat{ b \\ b \\ b }, 
    \right. \\
    &\left. \qquad \quad\;
      \pmat{ \X(z) \\ c_1 \\ c_2 },
      \pmat{ c_1 \\ \X(z) \\ c_2 }
      \mid \begin{gathered}
          (a,b,z,c_1,c_2)\in E^5 \text{ synchronized}, \\
          \Content(a) = A,\ \Content(b) = B,\ \Content(z) = 0
        \end{gathered}
    \right\},
\end{align*}
\begin{align*}
  &\Delta_{\exists A}
    = \left\{
      \pmat{ z \\ z \\ z },
      \pmat{ z \\ a \\ a },
      \pmat{ z \\ b \\ z },
      \pmat{ a \\ z \\ a },
      \pmat{ a \\ a \\ a },
      \pmat{ a \\ b \\ a },
      \pmat{ b \\ z \\ z },
      \pmat{ b \\ a \\ a },
      \pmat{ b \\ b \\ b }, 
    \right. \\
    &\left. \qquad \quad\ \ \,
      \pmat{ \X(z) \\ c_1 \\ c_2 },
      \pmat{ c_1 \\ \X(z) \\ c_2 }
      \mid \begin{gathered}
          (a,b,z,c_1,c_2)\in E^5 \text{ synchronized}, \\
          \Content(a) = A,\ \Content(b) = B,\ \Content(z) = 0
        \end{gathered}
    \right\},
\end{align*}
\begin{align*}
  &\Delta_{\exists B}
    = \left\{
      \pmat{ z \\ z \\ z },
      \pmat{ z \\ a \\ z },
      \pmat{ z \\ b \\ b },
      \pmat{ a \\ z \\ z },
      \pmat{ a \\ a \\ a },
      \pmat{ a \\ b \\ b },
      \pmat{ b \\ z \\ b },
      \pmat{ b \\ a \\ b },
      \pmat{ b \\ b \\ b }, 
    \right. \\
    &\left. \qquad \quad\ \ \,
      \pmat{ \X(z) \\ c_1 \\ c_2 },
      \pmat{ c_1 \\ \X(z) \\ c_2 }
      \mid \begin{gathered}
          (a,b,z,c_1,c_2)\in E^5 \text{ synchronized}, \\
          \Content(a) = A,\ \Content(b) = B,\ \Content(z) = 0
        \end{gathered}
    \right\}.
\end{align*}
As an example of how $\Delta_{\exists A}$ can be used, consider an
evaluation of an operation $f$,
\[
  f\pmat{ \vect{i, A}, & \vect{i, B}, & \vect{i, 0}, & \vect{i, B}, & \vect{i, A} \\
          \vect{i, 0}, & \vect{i, B}, & \vect{i, B}, & \vect{i, A}, & \vect{i, A} }
  = \pmat{ \vect{j, A} \\
           \vect{j, A} }.
\]
We can add a row to this evaluation in such a way that the input vectors are
in $\Delta_{\exists A}$, and if $f$ preserves $\Delta_{\exists A}$ then the
output will be in $\Delta_{\exists A}$ and therefore equal to $\vect{j,A}$:
\[
  f\pmat{ \vect{i, A}, & \vect{i, B}, & \vect{i, 0}, & \vect{i, B}, & \vect{i, A} \\
          \vect{i, 0}, & \vect{i, B}, & \vect{i, B}, & \vect{i, A}, & \vect{i, A} \\
          \vect{i, A}, & \vect{i, B}, & \vect{i, 0}, & \vect{i, A}, & \vect{i, A} }
  = \pmat{ \vect{j, A} \\
           \vect{j, A} \\
           \vect{j, A} }.
\]
Let us call this new third row the ``added row for $\Delta_{\exists A}$''.
Similar manipulations can be performed using $\Delta_{\forall}$ and
$\Delta_{\exists B}$. Doing this for the $2$-line evaluation at the start
and writing just the ``added'' rows, we obtain
\[
  f\pmat{ \vect{i, 0}, & \vect{i, B}, & \vect{i, 0}, & \vect{i, 0}, & \vect{i, A} \\
          \vect{i, A}, & \vect{i, B}, & \vect{i, 0}, & \vect{i, A}, & \vect{i, A} \\
          \vect{i, 0}, & \vect{i, B}, & \vect{i, B}, & \vect{i, B}, & \vect{i, A} }
  = \pmat{ \vect{j, A} \\
           \vect{j, A} \\
           \vect{j, A} }.
\]
The first row is the added row for $\Delta_{\forall}$, the second for
$\Delta_{\exists A}$, and the third for $\Delta_{\exists B}$. The subpower
$\Gamma$ defined next can be used to further manipulate the input. This
technique is discussed in detail in the proof of
Theorem~\ref{thm:few_D_entailment}.
\end{defn}   
\begin{defn}  \label{defn:gamma}  
Define a subset of $\AM^4$,
\begin{align*}
  &\Gamma
    = \left\{
      \pmat{ z \\ z \\ z \\ z },
      \pmat{ a \\ a \\ a \\ a },
      \pmat{ b \\ b \\ b \\ b },
      \pmat{ z \\ a \\ z \\ \alpha },
      \pmat{ z \\ a \\ b \\ \gamma },
      \pmat{ z \\ z \\ b \\ \beta },
      \pmat{ c_1 \\ c_2 \\ c_3 \\ c_1\wedge c_2\wedge c_3 },
    \right. \\
    &\left. \qquad\ \,
      \pmat{ \X(z) \\ c_1 \\ c_2 \\ c_3 },
      \pmat{ c_1 \\ \X(z) \\ c_2 \\ c_3 },
      \pmat{ c_1 \\ c_2 \\ \X(z) \\ c_3 }
      \mid \begin{gathered}
          (a,b,z,c_1,c_2,c_3)\in E^6 \text{ synchronized}, \\
          \Content(a) = A,\ \Content(b) = B,\ \Content(z) = 0, \\
          \alpha\in \{z,a\},\ \beta\in \{z,b\},\ \gamma\in \{z,a,b\}
        \end{gathered}
    \right\}.
\end{align*}
As an example of how $\Gamma$ can be used, consider the ``added row''
evaluation that we ended the discussion of the $\Delta_{\forall}$,
$\Delta_{\exists A}$, $\Delta_{\exists B}$ relations with:
\[
  f\pmat{ \vect{i, 0}, & \vect{i, B}, & \vect{i, 0}, & \vect{i, 0}, & \vect{i, A} \\
          \vect{i, A}, & \vect{i, B}, & \vect{i, 0}, & \vect{i, A}, & \vect{i, A} \\
          \vect{i, 0}, & \vect{i, B}, & \vect{i, B}, & \vect{i, B}, & \vect{i, A} }
  = \pmat{ \vect{j, A} \\
           \vect{j, A} \\
           \vect{j, A} }.
\]
If $f$ preserves $\Gamma$ then a row can be added to this evaluation so
that the input vectors will be in $\Gamma$ and the output will remain
unchanged:
\[
  f\pmat{ \vect{i, 0},      & \vect{i, B}, & \vect{i, 0},     & \vect{i, 0},      & \vect{i, A} \\
          \vect{i, A},      & \vect{i, B}, & \vect{i, 0},     & \vect{i, A},      & \vect{i, A} \\
          \vect{i, 0},      & \vect{i, B}, & \vect{i, B},     & \vect{i, B},      & \vect{i, A} \\
          \vect{i, \alpha}, & \vect{i, B}, & \vect{i, \beta}, & \vect{i, \gamma}, & \vect{i, A} }
  = \pmat{ \vect{j, A} \\
           \vect{j, A} \\
           \vect{j, A} \\
           \vect{j, A} }
\]
where $\alpha\in \{0,A\}$, $\beta\in \{0,B\}$, and $\gamma\in \{0,A,B\}$.
Note that different choices of $\alpha, \beta, \gamma$ result in the first
three rows of the original evaluation of $f$ in
Definition~\ref{defn:deltas}. As a result, if $f$ preserves $\Gamma$ then
the behavior of $f$ on the three rows above determines the behavior of $f$
on many other rows. This technique is discussed in detail in the proof of
Theorem~\ref{thm:few_D_entailment}.
\end{defn}   
\begin{lem} \label{lem:relns} 
The subpowers $\mu$, $\chi$, $\Delta_{\forall}$, $\Delta_{\exists A}$,
$\Delta_{\exists B}$ of Definitions~\ref{defn:mu_chi} and~\ref{defn:deltas}
are relations of $\AAM$.
\end{lem}
\begin{proof}
It is a straightforward (though tedious) procedure to verify that these are
all relations. We will sketch the proof for $\Delta_{\exists A}$ and leave
the others to the reader.

It suffices to show that if $F$ is an $\ell$-ary fundamental operation and
$g_1,\dots,g_\ell\in \Delta_{\exists A}$ then 
\[
  \alpha 
  = F(g_1,\dots,g_\ell)
  \in \Delta_{\exists A}.
\]
There are a few observations that we can make.
\begin{itemize}
  \item $\Delta_{\exists A} \subseteq E^3$ (i.e.\ $\Delta_{\exists A}$ has
    no elements with content $\Dot$). This simplifies the definitions of
    many of the operations of $\AAM$.

  \item If $d\in \Delta_{\exists A}$ has $d(1,2)\in Y^2$ then $d\in Y^3$.

  \item If $\Er\in \{\Content(\alpha(1)), \Content(\alpha(2))\}$ then
    $\alpha\in \Delta_{\exists A}$ since the elements $c_1$ and $c_2$ are
    unconstrained. Hence, we may assume that $\alpha(1,2)\in Y^2$.

  \item If $d\in \Delta_{\exists A}\cap Y^3$ then $d(1,2)$ uniquely
    determines $d(3)$.
\end{itemize}
The proof can be done by cases depending on which operation $F$ is, and all
of the cases are straightforward using the observations above.
\end{proof} 
\begin{lem} \label{lem:relns_reduct}  
The subpower $\Gamma$ of Definition~\ref{defn:gamma} is closed under all
operations of $\AAM$ except for $I$.
\end{lem}
\begin{proof}
As in the previous lemma, the proof is straightforward after making a few
observations. We will therefore provide only a sketch of it. It is
enough to show that if $F$ is an $\ell$-ary fundamental operation and
$g_1,\dots,g_\ell\in \Gamma$ then
\[
  \alpha 
  = F(g_1,\dots,g_\ell)
  \in \Gamma.
\]
Observe the following.
\begin{itemize}
  \item $\Gamma \subseteq E^4$ (i.e.\ $\Gamma$ has no elements with content
    $\Dot$). This simplifies the definitions of many of the operations of
    $\AAM$.

  \item If $d\in \Gamma$ has $d(4)\in X$ and $d(1,2,3)\in Y^3$ then
    \[
      \big| \big\{ \Content(d(i)) \mid i\in \{1,2,3\} \big\} \big| \geq 2.
    \]
    In particular, if $d(1) = d(2) = d(3)\in Y$ then $d(4) = d(1)$.

  \item If $d\in \Gamma\cap Y^4$ then $d(4)\in \big\{ d(1), d(2), d(3) \big\}$.
\end{itemize}
The proof can be done by cases depending on which operation $F$ is. All of
these cases are straightforward using these observations.
\end{proof} 

\begin{lem} \label{lem:f_E_monotone} 
Assume that there is $\ell$ such that
\begin{itemize}
  \item $\Rel_{\leq 2}(\AAM) \Entails f$ and $f$ is $n$-ary,

  \item $G = \{g_1,\dots,g_n\} \subseteq E$ and $g_\ell\in C$,

  \item $f(g_1,\dots, g_\ell, \dots, g_n) = \alpha\in Y$, and

  \item $f(g_1,\dots, \X(g_\ell), \dots, g_n)\in Y$.
\end{itemize}
Then $f(g_1,\dots, \X(g_\ell), \dots, g_n) = \alpha$.
\end{lem}
\begin{proof}
The function $f$ respects binary relations, so in particular it respects
$\mu$ from Definition~\ref{defn:mu_chi}. Consider
\[
  f\pmat{ g_1, & \cdots & g_\ell,     & \cdots & g_n \\
          g_1, & \cdots & \X(g_\ell), & \cdots & g_n }
  = \pmat{\alpha \\ \beta}.
\]
The hypotheses on $G$ mean that all the argument vectors are in $\mu$, so
the output must be as well. By hypothesis $\beta\not\in X$, so the only
possibility for $(\alpha,\beta)\in \mu$ is if $\beta = \alpha$, as claimed.
\end{proof}  

We next analyze some metrics which can be defined on relations. A major
component of the argument in Section~\ref{sec:M_halts_entailment} is proving
that entailment by lower arity relations is guaranteed when these metrics
are small or large enough.

\begin{defn} \label{defn:D_H} 
Let $\RR\leq \AAM^m$ be computational and define
\begin{align*}
  \D(\RR)
    &= \big\{ i\in [m] \mid R(i)\cap D\neq\emptyset \big\}, \\
  \H(\RR)
    &= \big\{ i\in [m] \mid \RR(\neq i) \text{ is halting} \big\}.
\end{align*}
We call $\D(\RR)$ the \emph{dot} part of $\RR$ and $\H(\RR)$ the
\emph{approximately halting} part of $\RR$. When the relation is clear, we
will sometimes use $\D$ for $\D(\RR)$ and $\H$ for $\H(\RR)$.
\end{defn}   
\begin{lem} \label{lem:capacity}  
Let $\RR\leq \AAM^m$ be computational.
\begin{enumerate}
  \item Let $I = \H(\RR)\cap \D(\RR)$. There are vectors 
    $(\sigma_i)_{i\in I}$ in $R$ satisfying Equation~\eqref{eqn:sigma}:
    \[ \tag{\ref{eqn:sigma} redux}
      \sigma_i(j) = \begin{cases}
        \vect{1,\Dot} & \text{if } i = j, \\
        \vect{1,0} & \text{otherwise}.
      \end{cases}
    \]

  \item If $\D(\RR)\neq\emptyset$ then $\RR$ is halting if and only if
    $R\cap C^m\neq \emptyset$.

  \item $\RR$ has capacity $|\D(\RR)\cap \H(\RR)|-1$ and this is the largest
    capacity it has.

  \item If $\RR$ is non-halting then $|\D(\RR)\cap \H(\RR)|\leq \kappa$.

  \item If $R\cap C^m\neq\emptyset$ and $\RR$ has weak capacity $k$ then
    $\RR$ has capacity $k$.
\end{enumerate}
\end{lem}
\begin{proof}
\Proofitem{(1)}
Let $\tau'_i\in R$ be such that $\tau'_i(\neq i)$ is a halting vector and
let $\sigma'_i\in R$ be such that $\sigma'_i(i) \in D$. Define $\sigma_i =
I(\sigma'_i,H(\tau'_i))$. It is easy to check that $\sigma_i$ satisfies
Equation~\eqref{eqn:sigma}.

\Proofitem{(2)}
If $\RR$ is halting then there is some vector $r\in R$ such that $r(i) =
\vect{0,\Dot}$ and $r(\neq i)\in \{ \vect{0,0} \}^{m-1}$. It follows that
$H(r) = (\vect{0,0}, \dots, \vect{0,0})\in C^m$. For the other direction, if
$c'\in R\cap C^m$ then let $c = I(c',c') = (\vect{1,0}, \dots, \vect{1,0})$.
Since $\D(\RR)\neq\emptyset$, $\RR$ has non-negative weak capacity (see
Definition~\ref{defn:capacity}). Let $\sigma'$ be a witness to $\RR$ having
weak capacity $0$, say $\sigma'(i) \in D$. Let $\sigma = I(\sigma',c)$ so
that $\sigma(i) = \vect{1,\Dot}$ and $\sigma(\neq i) = (\vect{1,0}, \dots,
\vect{1,0})$. We assumed at the end of Section~\ref{sec:encoding} that
$\m{T} = \Sg_{\AAM}\big( \big\{ \vect{1,0}, \vect{1,\Dot} \big\} \big)$ was
halting, so $\RR(i)$, containing this subalgebra, must halt. This means that
there is a term $t$ in the operations $\{ M,M' \}$ such that $t(\sigma,c)(i)
= \vect{0,\Dot}$. From the definitions of $\sigma$ and $c$ and by
Lemma~\ref{lem:basic_facts_AM} item (6), this implies
\[
  t(\sigma,c)(j)
  = t\big( \vect{1,0}, \vect{1,0} \big)
  = \vect{0,0}
\]
for all $j\neq i$. Hence $t(\sigma,c)(i) = \vect{0,\Dot}$ and
$t(\sigma,c)(\neq i)\in \{ \vect{0,0} \}^{m-1}$, so $t(\sigma,c)$ is a
halting vector and $\RR$ is therefore halting.

\Proofitem{(3)}
Item (1) implies that $\RR$ has capacity $|\D(\RR)\cap \H(\RR)| - 1$ (the
$\sigma_i$ are witnesses). Suppose now that we have a vector $r\in R\cap
Y^m$ such that $r(j)\in D$. It follows that $j\in \D(\RR)$ and that $r(\neq
j)\in C^{m-1}$. By item (2) we have that $\RR(\neq j)$ is halting and thus
$j\in \H(\RR)$. Therefore $j\in \D(\RR)\cap \H(\RR)$.

\Proofitem{(4)} This follows from item (3) (recall that $\M$ halts with
capacity $\kappa$).

\Proofitem{(5)}
Let $c\in R\cap C^m$ and let $\tau_i$ be a witness to $\RR$ having weak
capacity $0$, say $\tau_i(i) \in D$. Define $\sigma_i = I(\tau_i,c)$ and
observe that $\sigma_i\in Y^m$ satisfies equation~\eqref{eqn:sigma}. Doing
this for all $k$ witnesses of $\RR$'s weak capacity yields witnesses to
$\RR$ having capacity $k$.
\end{proof} 

The set $\Gamma$ from Definition~\ref{defn:gamma} will play an important
role in the argument for entailment. Since $\Gamma$ is closed under all
operations except for $I$ by Lemma~\ref{lem:relns_reduct}, it will be
necessary to understand a bit about how $I$ can interact with the other
operations. The next lemma and proposition are our first steps in this
direction.

\begin{defn} \label{defn:RI}  
Let $\RR\leq \AAM^m$. Define  $\ds{ \RR_I = \Sg_{\AAM^m}\big( I(R\cap Y^m,
R\cap Y^m) \big) }$.
\end{defn}   
\begin{lem} \label{lem:RI}  
Let $\RR\leq \AAM^m$ be computational.
\begin{enumerate}
  \item If $p$ is a permutation on $[m]$ which restricts to a permutation on 
    \[
      K 
      = \big\{ i \mid \exists r\in R\cap Y^m \text{ such that } r(i)\in D \big\}
    \]
    then $p(R_I) = R_I$.

  \item If $\D(\RR)\neq \emptyset$ then $\RR$ is halting if and only if
    $\RR_I$ is halting.

  \item If $|\D(\RR)|\geq 2$ then $\D(\RR_I) = \D(\RR)\cap \H(\RR)$.

  \item Let $\D_I = \D(\RR_I)$. Then $\RR_I(\D_I) = \SS_{|\D_I|}$. In
    particular there are elements $(\sigma_i)_{i\in \D_I}$ in $R_I$
    satisfying Equation~\eqref{eqn:sigma} and
    \[
      \RR_I
      = \Sg_{\AAM^m}\big\{ \sigma_i \mid i\in \D_I \big\}.
    \]
\end{enumerate}
\end{lem}
\begin{proof}
\Proofitem{(1)}
Let $a,b\in R\cap Y^m$. From the definition of $I$ and since $a,b\in Y^m$,
\[
  I(a,b)(j)
  = \begin{cases}
    \vect{1,\Dot} & \text{if } a(j)\in D, \\
    \vect{1,\Er}  & \text{if } a(j)\not\in D \text{ and } b(j)\in D, \\
    \vect{1,0}    & \text{otherwise}.
  \end{cases}
\]
This is a typical element of $I(R\cap Y^m, R\cap Y^m)$. Observe that the
first two cases in the equation imply $j\in K$. For all pairs
$i,j\in K$, choose $r_i,r_j\in R\cap Y^m$ such that $r_i(i),r_j(j)\in D$ and
define elements $s_{ij} = I(r_i,r_j)$. From the description of elements of
$I(R\cap Y^m, R\cap Y^m)$, we have
\[
  \big\{ s_{ij} \mid i,j\in K \big\}
  = I(R\cap Y^m, R\cap Y^m).
\]
The set on the left is closed under the permutation $p$, so $I(R\cap Y^m,
R\cap Y^m)$ must be as well. These are the generators of $\RR_I$, so the
conclusion follows.

\Proofitem{(2)}
Since $\RR_I\leq \RR$ and $\D(\RR_I)\subseteq \D(\RR)$, if $\RR_I$ is
halting then so is $\RR$. Conversely, if $\RR$ is halting then $R\cap
C^m\neq\emptyset$ by Lemma~\ref{lem:capacity} item (5). Pick any $r\in R\cap
C^m$. It follows that $I(r,r)\in C$ and
\[
  I(r,r)
  \in I(R\cap Y^m, R\cap Y^m),
\]
so $R_I\cap C^m\neq\emptyset$. Therefore $\RR_I$ is halting.

\Proofitem{(3)} 
Suppose that $i\in \D(\RR_I)$. Since $\RR_I\leq \AAM^m$, it follows that
there is a generator $g = I(a,b)$, $a,b\in R\cap Y^m$, with $g(i)\in D$.
This implies that $a(i)\in D$, so $i\in \D(\RR)$ and $a(\neq i)\in C^{m-1}$.
Since $|\D(\RR)|\geq 2$, we have $\D(\RR(\neq i))\neq \emptyset$ and
$\RR(\neq i)\cap C^{m-1}\neq\emptyset$. By Lemma~\ref{lem:capacity} item (2)
$\RR(\neq i)$ must halt, so $i\in \H(\RR)$. For the reverse inclusion,
suppose that $i\in \D(\RR)$ and $i\in \H(\RR)$. By Lemma~\ref{lem:capacity}
item (1), we have that there is an element $\sigma_i\in R$ such that
$\sigma_i(i) = \vect{1,\Dot}$ and $\sigma_i(\neq i)\in \{ \vect{1,0}
\}^{m-1}$. Hence $\sigma_i\in R_I$, so $i\in \D(\RR_I)$.

\Proofitem{(4)}
This follows from items (2) and (3) above, Lemma~\ref{lem:capacity} item
(1), and the definition of $\SS_k$.
\end{proof}  
\begin{prop} \label{prop:I_avoidance} 
Let $\RR\leq \AAM^m$ be computational non-halting. If $t$ is a $k$-ary term
operation and $\overline{r}\in R^k$ is such that $t(\overline{r}) \in Y^m$
then
\begin{enumerate}
  \item $t(\overline{r})\in R_I$ or

  \item there is a term operation $s$ without operation $I$ in its term tree
    such that $s(\overline{r}) = t(\overline{r})$.
\end{enumerate}
\end{prop}
\begin{proof}
Let $t(\overline{r}) = \alpha$ and assume that (2) is not the case, so if we
have $s(\overline{r}) = \alpha$ then $s$ has $I$ in its term tree. We will
prove that $\alpha\in R_I$. The proof shall be by induction on the
complexity of $t$. If we have $\alpha = I(a, b)$ for some $a,b\in R\cap Y^m$
then $\alpha\in R_I$ by definition. This establishes the basis of the
induction. Assume now that $t$ is not a projection, so $t$ can be written as
\[
  t(\overline{x})
  = F\big( f_1(\overline{x}),\dots,f_\ell(\overline{x}) \big),
\]
where $F$ is an $\ell$-ary fundamental operation and the $f_i$ are other
$k$-ary term operations. We will proceed by cases depending on which
operation $F$ is.

\Case{$F\in \{\wedge, N_0, N_\Dot, P\}$}
Since $\RR$ is computational and non-halting, such $F$ have the property
that $F(\overline{a})\leq a_i$ for some $a_i$ amongst the $\overline{a}$, by
the various parts of Lemma~\ref{lem:basic_facts_relns}. Therefore, if
$\alpha = F(f_1(\overline{r}),\dots,f_n(\overline{r}))$ then $\alpha\leq
f_j(\overline{r})$ for some $j$. Since $\alpha\in Y^m$, this implies that
$f_j(\overline{r}) = \alpha$. As (2) does not hold, $f_j$ must have $I$ in
its term tree, so by the inductive hypothesis we have that
$\alpha = f_j(\overline{r})\in R_I$.

\Case{$F\in \{M', H\}$}
In this case, $F$ is $X$-absorbing and unary, by
Lemma~\ref{lem:basic_facts_AM} item (2). It follows that
$F(f_1(\overline{r})) = \alpha\in Y^m$ implies $f_1(\overline{r})\in Y^m$
and that $I$ is in the term tree of $f_1$. Therefore the inductive
hypothesis applies and $f_1(\overline{r})\in R_I$. Hence $\alpha\in R_I$.

\Case{$F = M$}
Since $\alpha\in Y^m$, by Lemma~\ref{lem:basic_facts_AM} item (2) we have
$f_1(\overline{r}), f_2(\overline{r})\in Y^m$. The term operation $t$ has
$I$ in its term tree, so one of the $f_i$ does as well. By the inductive
hypothesis, one of $f_i(\overline{r})$ is in $R_I$. If $\D(\RR) = \emptyset$
then $f_1(\overline{r}) = f_2(\overline{r})$, so both belong to $R_I$. If
$\D(\RR)\neq\emptyset$ then $R\cap C^m = \emptyset$ by
Lemma~\ref{lem:capacity} item (2). It follows from this and the definition
of $M$ that there are coordinates $j,k$ such that
\begin{itemize}
  \item $f_1(\overline{r})(k)$, $f_2(\overline{r})(j)$, $\alpha(j) \in D$,

  \item $f_1(\overline{r})(j) = f_2(\overline{r})(k)$, and

  \item $f_1(\overline{r})(\ell) = f_2(\overline{r})(\ell)$ for $\ell\not\in
    \{j,k\}$.
\end{itemize}
That is, $f_1(\overline{r})$ and $f_2(\overline{r})$ equal under the
coordinate transposition swapping $j$ and $k$. By Lemma~\ref{lem:RI} item
(1), one of them being in $R_I$ implies the other is in $R_I$ as well.
Therefore $\alpha\in R_I$.

\Case{$F\in \{I,S\}$}
From the definitions and Lemma~\ref{lem:basic_facts_relns} item (4), we have
that $\alpha = I(\alpha,\alpha)$ in this case. Thus $\alpha\in R_I$.

\smallskip

This completes the case analysis, the induction, and the proof.
\end{proof} 

The next proposition and subsequent definition establishes the biggest tool
we have for analyzing the halting status of a relation. It is absolutely
essential to the proofs in the next section.

\begin{prop} \label{prop:inherent_nonhalt}  
Suppose that the relation $\RR\leq \AAM^m$ is computational non-halting.
There exists $\N\subseteq [m]$ such that
\begin{enumerate}
  \item $\RR(\N)$ is non-halting,

  \item $|\N\cap \D(\RR)| \leq \kappa$,

  \item if $\D(\RR)\neq\emptyset$ then $\N\cap \D(\RR)\neq \emptyset$, and

  \item $\big( [m]\setminus \D(\RR) \big) \subseteq \N$.
\end{enumerate}
\end{prop}
\begin{proof}
If $\D(\RR) = \emptyset$ then take $\N = [m]$. It is not hard to see that
$\N$ satisfies (1)--(4). Assume now that $\D(\RR)\neq\emptyset$ and let
$\N'$ be minimal such that $\N'\cap \D(\RR)\neq\emptyset$ and $\RR(\N')$ is
non-halting. Since $\RR$ is already non-halting, there is at least one such
$\N'$. We begin by proving that $|\N'\cap \D(\RR)| \leq \kappa$.

Suppose that we have distinct $i_1,\dots,i_{\kappa+1}\in \N'\cap \D(\RR)$.
By the minimality of $\N'$, we have that $\RR(\N'\setminus \{i_k\})$ is
halting for each $k$. Therefore 
\[
  i_1,\dots,i_{\kappa+1}
  \in \H\big( \RR(\N') \big) \cap \D\big( \RR(\N') \big),
\]
so $\RR(\N')$ is halting by Lemma~\ref{lem:capacity} item (4), a
contradiction. Hence $|\N'\cap \D(\RR)|\leq \kappa$. Let 
\[
  \N = \N' \cup \big( [m]\setminus \D(\RR) \big).
\]
It is easy to see that $|\N\cap \D(\RR)| = |\N' \cap \D(\RR)|\leq \kappa$.
Suppose towards a contradiction that $\RR(\N)$ is halting. By
Lemma~\ref{lem:capacity} item (2) we have $\RR(\N)\cap C^{\N}\neq\emptyset$.
It follows that $\RR(\N')\cap C^{\N'}\neq \emptyset$, and so $\RR(\N')$ is
halting, contradicting the choice of $\N'$. Therefore $\RR(\N)$ is
non-halting, and we are done.
\end{proof}  
\begin{defn} \label{defn:inherent_nonhalt}  
For each $\RR$ that is computational non-halting we fix a set of indices
$\N(\RR)$ satisfying the conclusion of
Proposition~\ref{prop:inherent_nonhalt}. We call $\N(\RR)$ the
\emph{inherently non-halting} part of $\RR$. As with $\D$ and $\H$, if the
relation is clear then we will sometimes use $\N$ instead of $\N(\RR)$.
\end{defn}   
\begin{lem} \label{lem:inherent_properties} 
Let $\RR\leq \AAM^m$ be computational non-halting and suppose that $r\in R$.
\begin{enumerate}
  \item If $\D(\RR)\neq\emptyset$ then there is $j\in \N(\RR)$ with $r(j)\in
    D\cup X$.

  \item If $i\not\in \N(\RR)$ and $r(i)\in D$ then there is $j\in \N(\RR)$
    with $r(j)\in X$.

  \item $\D(\RR_I) \subseteq \N(\RR) \cap \D(\RR)$.

  \item $\H(\RR)\subseteq \N(\RR)$.

  \item If $|\D(\RR)| \leq 1$ then $\N(\RR) = [m]$.
\end{enumerate}
\end{lem}
\begin{proof}
\Proofitem{(1)}
We have that $\RR(\N(\RR))$ is non-halting. Since $\N(\RR)\cap
\D(\RR)\neq\emptyset$, this implies $\RR(\N(\RR))\cap C^{\N(\RR)} =
\emptyset$ by Lemma~\ref{lem:capacity} item (2). Therefore
$r(\N(\RR))\not\in C^{\N(\RR)}$. The conclusion follows.

\Proofitem{(2)} 
This follows from item (1). If we have $r(i)\in D$ and $r(j)\in D\cup X$ for
$j\neq i$ then $r(j)\in X$ since $\RR$ is computational.

\Proofitem{(3)} 
We already have $\D(\RR_I)\subseteq \D(\RR)$, so we only need to show
$\D(\RR_I)\subseteq \N(\RR)$. Let $i\in \D(\RR_I)$. The only way this is
possible is if there is a generator $g = I(a,b)$, $a,b\in R\cap Y^m$, with
$g(i), a(i)\in D$. If $i\not\in \N(\RR)$ then by item (2) above there is
$j\in \N(\RR)$ with $a(j)\in X$, contradicting $a\in Y^m$.

\Proofitem{(4)}
Let $i\in \H(\RR)$, so that $\RR(\neq i)$ is halting. If $i\not\in \N(\RR)$
then $\RR(\N) = \RR(\neq i)(\N)$, so we have that $\RR(\N)$ is halting as
well, contradicting Proposition~\ref{prop:inherent_nonhalt} item (1). Hence
$i\in \N(\RR)$.

\Proofitem{(5)}
If $\D(\RR) = \emptyset$ then $\N(\RR) = [m]$ from the proof of
Proposition~\ref{prop:inherent_nonhalt}. If $\D(\RR) = \{i\}$ then $i\in
\N(\RR)$ since $\D(\RR)\cap \N(\RR)\neq\emptyset$. Since we also have
$([m]\setminus \D(\RR))\subseteq \N(\RR)$, the conclusion follows.
\end{proof}  

We have now built enough tools to attack the main problem.

\section{If \texorpdfstring{$\M$}{M} halts --- entailment}\label{sec:M_halts_entailment}  
As with the previous section, we assume throughout that $\M$ halts with
capacity $\kappa$. The overall structure of the argument will be to consider
a relation $\RR\in \Rel_{\leq m}(\AAM)$, and proceed by cases. These cases
are laid out in the proof of the main entailment theorem, which we begin the
section with (after introducing some notation). The proof references the
theorems later in this section, but it is useful at the outset to see the
overall strategy.

\begin{defn} \label{defn:approx_element}  
Let $\RR\leq \AAM^m$ and $\alpha\in \AM^m$.
\begin{itemize}
  \item We say that $\Approx(\alpha,i)$ holds for $\RR$ if
    $\alpha(\neq i)\in R(\neq i)$.

  \item We say that $\ApproxI(\alpha,i)$ holds for $\RR$ if
    $\Approx(\alpha,i)$ holds for $\RR_I$.

  \item If $\RR$ has $\Approx(\alpha,i)$ then fix an element $\alpha_i\in
    R$ such that $\alpha_i(\neq i) = \alpha(\neq i)$, and likewise if
    $\ApproxI(\alpha,i)$ holds. If $\RR$ has both $\Approx(\alpha,i)$ and
    $\ApproxI(\alpha,i)$ then take $\alpha_i\in R_I\subseteq R$.
\end{itemize}
If the relation $\RR$ is clear, we will use $\Approx(\alpha,i)$ and
$\ApproxI(\alpha,i)$ without reference to the relation.
\end{defn}   
\begin{cor} \label{cor:finitely_related} 
If $\M$ halts then $\deg(\AAM) \leq \kappa+15$.
\end{cor}
\begin{proof}
We will show that $\Rel_{\leq \kappa+15}(\AAM)\Entails \Rel_{\leq m}(\AAM)$
by induction on $m$. The base case of $m = \kappa+15$ is included in the
hypotheses. Suppose now that $m\geq \kappa+16$, $\RR\leq \AAM^m$,
$\Rel_{\leq m-1}(\AAM) \Entails f$, and $f(r_1,\dots,r_n) = \alpha$ for some
$r_1,\dots, r_n\in R$. We endeavor to prove $\alpha\in R$. Let $G =
\{r_1,\dots,r_n\}$. Without loss of generality we may assume that $\RR =
\Sg_{\AAM^m}(G)$.

If $\RR$ is not computational or is halting then
Theorem~\ref{thm:noncomp_halting_entailed} yields $\Rel_{\leq m-1}\Entails
\RR$, so $\alpha\in R$. Therefore we assume that
\begin{enumerate}[label=\textbf{(\arabic*)}, series=assumptions]
  \item $\RR$ is both computational and non-halting, so $|\N(\RR)\cap
    \D(\RR)|\leq \kappa$ by Proposition~\ref{prop:inherent_nonhalt}.
\end{enumerate}
If $\Er\in \Content(\alpha)$ then Theorem~\ref{thm:X_entailment} yields
$\alpha\in R$. Therefore we assume that 
\begin{enumerate}[label=\textbf{(\arabic*)}, resume=assumptions]
  \item $\alpha\in Y^m$.
\end{enumerate}
By the inductive hypothesis $\RR$ has $\Approx(\alpha, k)$ for all $k \in
[m]$. If there are distinct $i,j\not\in \N$ such that $\RR$ has
$\ApproxI(\alpha,i)$ and $\ApproxI(\alpha,j)$ then
Theorem~\ref{thm:I_entailment} yields $\alpha\in R$. Therefore we assume
that 
\begin{enumerate}[label=\textbf{(\arabic*)}, resume=assumptions]
  \item there is at most one $i\not\in \N(\RR)$ such that
    $\ApproxI(\alpha,i)$.
\end{enumerate}
If $|[m] \setminus \D(\RR)|\geq 11$ then Theorem~\ref{thm:few_D_entailment}
yields $\alpha\in R$. Therefore we assume that
\begin{enumerate}[label=\textbf{(\arabic*)}, resume=assumptions]
  \item $\big| [m] \setminus \D(\RR) \big|\leq 10$, so $|\N|\leq \kappa+10$
    by Proposition~\ref{prop:inherent_nonhalt}.
\end{enumerate}
Finally, our list of assumptions agrees with the hypotheses of
Theorem~\ref{thm:main_entailment}, so $\alpha\in R$.
\end{proof}  

Having established the overall strategy we will be pursuing, we prove our
first entailment theorem --- entailment for non-computational or halting
relations.

\begin{thm} \label{thm:noncomp_halting_entailed} 
If $m\geq 3$ and $\RR\leq \AAM^m$ fails to be computational or is halting
then $\Rel_{\leq m-1}(\AAM) \Entails \RR$.
\end{thm}
\begin{proof}
Towards a contradiction, suppose that $\Rel_{\leq m-1}(\AAM) \not\Entails
\RR$. This implies that there is some $n$-ary function $f$ and
$\overline{r}\in R^n$ such that $\Rel_{\leq m-1}(\AAM) \Entails f$ and
$f(\overline{r}) = \alpha\not\in R$. Since $\RR(\neq i)\in \Rel_{\leq
m-1}(\AAM)$, we have that $\alpha(\neq i)\in R(\neq i)$, so
$\Approx(\alpha,i)$ holds and we have elements $\alpha_i\in R$ for all $i$
(cf.\ Definition~\ref{defn:approx_element}).
There are three cases to consider: $\RR$ is not synchronized, there is $r\in
R$ with $|r^{-1}(D)| \geq 2$, or $\RR$ is halting.

\Case{$\RR$ is not synchronized}
In this case there is an $r\in R$ with a non-constant state. For each state
$i$ of $\M$ let
\[
  K_i = \big\{ j \mid \State(r(j)) = i \big \}
  \qquad\text{and}\qquad
  L_i = [m]\setminus K_i
\]
and pick some $K_k\neq \emptyset$. Let $s = z_k(r)$ from
Lemma~\ref{lem:basic_facts_AM} item (4). It follows that for $a,b\in R$
\[
  P(r,s,a,b)(j)
  = \begin{cases}
    a(j) & \text{if } j\in K_k, \\
    b(j) & \text{otherwise}.
  \end{cases}
\]
That is, $\RR$ is obtained by some permutation of the coordinates of
$\RR(K_k)\times \RR(L_k)$, so $\RR$ is entailed by lower-arity relations.

\Case{there is $r\in R$ with $|r^{-1}(D)|\geq 2$}
Assume that $\RR$ is synchronized. Let us choose distinct $i,j$ such that
$r(i),r(j)\in D$ and let $k$ be distinct from $i$ and $j$ (we use $m\geq 3$
here). From the definition of $N_\Dot$ it follows that
\[
  N_\Dot(r, \alpha_i, \alpha_j, \alpha_k)(\ell)
  = \begin{cases}
    \alpha(i) & \text{if } \ell = i, \\
    \alpha(j) & \text{if } \ell = j, \\
    \alpha(\ell) & \text{otherwise},
  \end{cases}
\]
so $\alpha = N_\Dot(r,\alpha_i,\alpha_j,\alpha_k)$ and hence $\alpha\in R$.

\Case{$\RR$ is halting}
Let us assume that $\RR$ is computational and that $r\in R$ is a halting
vector. That is, $r(\ell) = \vect{0,\Dot}$ for some $\ell$ and $r(\neq \ell)
\in \{ \vect{0,0} \}^{m-1}$. It is not possible for there to be two
coordinates $i$ at which $\alpha(i)\in D$ or $r(i)\in D$ since $\RR$ is
computational. If $\alpha(\ell)\in D$ or $\alpha\in C^m$ then by definition
of $N_0$,
\[
  \alpha
  = N_0\big( r, \alpha_i, \alpha_{\ell} \big)
\]
for some $i\neq \ell$, and hence $\alpha\in R$. The other possibility is
that there is some $k\neq \ell$ with $\alpha(k)\in D$. Let $s = H(r)$ and
$\beta' = I(\alpha_\ell,s)$. We have $s\in \{\vect{0,0}\}^m$, $\beta'(k) =
\vect{1,\Dot}$, and $\beta'(j) = \vect{1,0}$ for all $j\neq k$. From
$\beta'$ and $s$ we can obtain a halting vector $r'$ such that $r'(k)\in D$
and $r'(j) = \vect{0,0}$ (we use that $\m{T}$ from the end of
Section~\ref{sec:encoding} is halting here). As before,
\[
  \alpha 
  = N_0\big( r', \alpha_i,  \alpha_k \big)
\]
for some $i\neq k$, so $\alpha\in R$.
\end{proof}  

\subsection{Entailment for \texorpdfstring{$\RR_I$}{RR\_I}}
We now prove an entailment theorem for the relation $\RR_I$. Once proven, by
Proposition~\ref{prop:I_avoidance} we will be able to use the set $\Gamma$
from Definition~\ref{defn:gamma} and Lemma~\ref{lem:relns} in the other
cases.

\begin{lem}  \label{lem:RI_off_dot} 
Suppose that $\RR\leq \AAM^m$ is computational and let $\D_I = \D(\RR_I)$
and $L = [m] \setminus \D_I$. Then $\alpha\in R_I$ if and only if
\[
  \alpha(\D_I)\in R_I(\D_I)
  \qquad\text{and}\qquad
  \alpha(L)\in \left\{ 
    \pmat{ \vect{\State(\alpha), 0} \\ \vdots \\ \vect{\State(\alpha), 0} },
    \pmat{ \vect{\State(\alpha), \Er} \\ \vdots \\ \vect{\State(\alpha), \Er} }
  \right\}.
\]
\end{lem}
\begin{proof}
We begin by building some tools. Define 
\[
  Q
  = \bigcup_{i \text{ a state of } \M} \left\{ 
    \pmat{ \vect{i, 0} \\ \vdots \\ \vect{i, 0} },
    \pmat{ \vect{i, \Er} \\ \vdots \\ \vect{i, \Er} }
  \right\}
  \subseteq \AM^L.
\]
Examining the definitions of the operations of $\AAM$, observe that $Q$ is
a subuniverse of $\AAM^L$. From Lemma~\ref{lem:RI} item (4) we have that
$\RR_I(\D_I) = \SS_{|\D_I|}$ and there are elements $(\sigma_i)_{i\in \D_I}$
satisfying Equation~\eqref{eqn:sigma},
\[ \tag{\ref{eqn:sigma} redux}
  \sigma_i(j)
  = \begin{cases}
    \vect{1,\Dot} & \text{if } j = i, \\
    \vect{1,0}    & \text{otherwise},
  \end{cases}
\]
that generate $\RR_I$. Let $x = z_1(\sigma_i) \wedge z_2(\sigma_i)$ for the
term operations $z_j$ defined in Lemma~\ref{lem:basic_facts_AM} item (4) and
observe that $x = (\vect{1,\Er}, \dots, \vect{1,\Er})$. Suppose that $\D_I =
\{ i_1,\dots, i_k\}$. For $\ell\in \D_I$ and $j\leq k$ define the sequence
of elements $\tau_\ell^j\in R_I$ by
\[
  \tau_\ell^0
  = x
  \qquad\text{and}\qquad
  \tau_\ell^j 
  = N_{\Dot}\big( \sigma_{i_j}, \tau_\ell^{j-1}, \sigma_\ell, \sigma_\ell \big).
\]
Let $\tau_\ell = \tau_\ell^k$. It is not hard to see that $\tau_\ell(\D_I) =
\sigma_\ell(\D_I)$ and $\tau_\ell(L) = x(L) \in X^L$. We are now ready to
prove the lemma.

Suppose that $\alpha\in R_I$. It is immediate that $\alpha(\D_I)\in
R_I(\D_I)$. Furthermore, if $\sigma_i$ is one of the generators of $\RR_I$
then $\sigma_i(L)\in Q$. Since $Q$ is a subuniverse, this implies that
$\alpha(L)\in Q$. This completes the ``only if'' portion of the proposition.
The ``if'' portion will give us more difficulty.

Suppose that $\alpha(\D_I)\in R_I(\D_I)$ and $\alpha(L)\in Q$. It follows
that there is a term operation $t$ such that
\[
  \alpha(\D_I) = t\big( \overline{\sigma} \big)(\D_I),
\]
where $\overline{\sigma}$ are the generators of $\RR_I$. Let $\beta =
t(\overline{\sigma})$. Clearly $\beta\in R_I$, so if $\alpha\not\in R_I$
then $\alpha\neq \beta$. Since the $\overline{\sigma}$ are all equal with
content $0$ on coordinates $L$, $\alpha\neq \beta$ implies that one of
$\Content(\alpha(L)), \Content(\beta(L))$ is $\{\Er\}$ and the other is
$\{0\}$. Thus there are two cases to consider.

For the first case, suppose that $\Content(\alpha(L)) = \{\Er\}$ and
$\Content(\beta(L)) = \{0\}$. By Lemma~\ref{lem:basic_facts_relns} item (5)
we have
\[
  \alpha(L)
  = \X(\alpha(L))
  = \X(t(\overline{\sigma}))(L)
  = t\big( \X(\overline{\sigma})(L) \big)
  = t(\overline{\tau})(L)
\]
for the elements $\overline{\tau} = (\tau_i)_{i\in \D_I}$ defined at the
start of the proof. Since $\overline{\tau}(\D_I) = \overline{\sigma}(\D_I)$,
we have that $\alpha = t( \overline{\tau} )$ and hence $\alpha\in R_I$. 

For the second case, suppose that $\Content(\alpha(L)) = \{0\}$ and
$\Content(\beta(L)) = \{\Er\}$. After proving the next claim, we will be
done.

\begin{claim*}
If $a$ and $b$ are such that $b\in R_I$, $a(\D_I) = b(\D_I)$,
$\Content(b(L)) = \{\Er\}$, and $\Content(a(L)) = \{0\}$ then $a\in R_I$.
\end{claim*}
\begin{claimproof}
Let $G_0 = \big\{ \sigma_i \mid i\in \D_I \big\}$ be the generators of
$\RR_I$ and
\[
  G_n
  = \Big\{ F(\overline{g}) \mid F \text{ a fundamental $\ell$-ary
      operation},\ \overline{g}\in G_{n-1}^\ell \Big\}
    \cup G_{n-1}.
\]
Suppose towards a contradiction that the claim is false. Choose a
counterexample $a,b$ with $b\in G_n$ such that $n$ is minimal. When $b\in
G_0$ the claim's hypothesis fails, so it holds vacuously. Assume that $n>0$,
so
\[
  b = F\big( g_1,\dots,g_\ell \big)
\]
for some $\ell$-ary operation $F$ and elements $g_1,\dots,g_{\ell}\in
G_{n-1}$. If one of the $g_i$ has $g_i(L)\in X^L$ then by the inductive
hypothesis there is an element $g_i'\in G_{n-1}$ with $g_i'(\D_I) =
g_i({\D_I})$ and $\Content(g_i'(L))\in \{0\}$. Let $b'$ be the result of
replacing $g_i$ with $g_i'$ in the arguments of $F$. There are two
possibilities for $b'(L)$: either $\Content(b'(L)) = \{0\}$ (and so $b' =
a$) or $b'(L) = b(L)\in X^L$. In the first possibility we conclude that $a\in
R_I$, a contradiction, and in the second possibility we conclude $b' = b$.
We may therefore  assume without loss of generality that
$\Content(g_i)(L)\in \{0\}$ for all $i$.

Looking through the definitions of the operations, we can see that if
$\Content(g_i(L))\in \{0\}$ for all $i$ and $b(L) = F(\overline{g})(L)\in
X^L$ then it must be that $b({\D_I}) = F(\overline{g})({\D_I})\in X^{\D_I}$
(this follows from $\RR$ being synchronized). That is, $b\in X^m$ and thus
$a({\D_I}) = b({\D_I})\in X^{\D_I}$. We now have
\[
  a
  = \bigwedge_{i\in \D_I} z_{\State(a)}(\sigma_i),
\]
where $z_{\State(a)}$ is the term operation from
Lemma~\ref{lem:basic_facts_AM} item (4). It follows that $a\in R_I$, and we
are done.
\renewcommand{\qedsymbol}{\ensuremath{\circ} \quad \ensuremath{\square}}
\end{claimproof}
\let\qed\relax
\end{proof}  
\begin{prop} \label{prop:RI_off_inherent} 
Let $\RR\leq \AAM^m$ be computational non-halting. If $r\in R_I$, $i\not\in
\N(\RR)\cap \D(\RR)$, and $\Content(r(i))\neq 0$ then $r(j)\in X$ for all
$j\not\in \N(\RR)\cap \D(\RR)$.
\end{prop}
\begin{proof}
This follows immediately from the observation that $\D(\RR_I)\subseteq
\N(\RR)\cap \D(\RR)$ (Lemma~\ref{lem:inherent_properties} item (3)) and an
application of Lemma~\ref{lem:RI_off_dot}.
\end{proof}  
\begin{thm} \label{thm:I_entailment} 
Let $\RR \leq \AAM^m$ be computational non-halting. If $\alpha\in Y^m$ is
such that
\begin{itemize}
  \item $\Approx(\alpha,i)$ for all $i$, and

  \item there are distinct $k,\ell\not\in \N(\RR)$ such that
    $\ApproxI(\alpha,k)$ and $\ApproxI(\alpha,\ell)$
\end{itemize}
then $\alpha\in R_I$.
\end{thm}
\begin{proof}
We have $\alpha_k,\alpha_\ell\in R_I$ from
Definition~\ref{defn:approx_element}. Proposition~\ref{prop:RI_off_inherent}
implies that
\[
  \Content(\alpha_k(\ell)) = \Content(\alpha(\ell)) = 0
  \qquad\text{and}\qquad
  \Content(\alpha_\ell(k)) = \Content(\alpha(k)) = 0.
\]
Furthermore, if $\Content(\alpha_k(k))\neq 0$ then $\alpha_k(\ell)\in X$ by
the same proposition, a contradiction. It follows that $\alpha = \alpha_k$.
\end{proof} 

\subsection{Entailment when \texorpdfstring{$|\D|$}{|D|} is small}
We next show how relations with small $|\D|$ are entailed. The key to the
argument is to first prove that the generating set of such relations has a
specific form, and then to use the relations from
Definitions~\ref{defn:mu_chi},~\ref{defn:deltas}, and~\ref{defn:gamma}.

\begin{defn} \label{defn:chi_compat}  
Let $G\subseteq \AM^m$ and $\RR = \Sg_{\AAM^m}(G)$. We say that $G$ is
\emph{$\chi$-compatible} over $K\subseteq [m]$ if
\[
  G(K)
  \subseteq \left\{
    \pmat{\vect{i,\Er} \\ \vdots \\ \vect{i,\Er}}
    \pmat{\vect{i,a_1} \\ \vdots \\ \vect{i,a_m}}
    \mid
    a_1,\dots,a_m\in \{ A,B,0 \},\ i \text{ a state}
  \right\}.
\]
If $K$ is not specified then we take $K = [m]\setminus \D(\RR)$ (the
non-dot coordinates of $\RR$). Note that $K = \emptyset$ is allowed.
\end{defn}   
\begin{lem} \label{lem:not_chi_entailed_helper} 
Assume that
\begin{itemize}
  \item $\Rel_{\leq 3}(\AAM) \Entails f$ and $f$ is $n$-ary,

  \item $G = \{g_1,\dots,g_n\} \subseteq E^2$ and $\RR = \Sg_{\AAM^2}(G)$ is
    synchronized, and

  \item $f(g_1,\dots,g_n) = \alpha\in Y^2$.
\end{itemize}
If $G$ is not $\chi$-compatible then there is $g_{\ell}\in G$ such that
$g_{\ell}\not\in X^2$ and
\[
  f( g_1, \dots, \X(g_\ell), \dots, g_n) = \alpha.
\]
\end{lem}
\begin{proof}
$G$ is not $\chi$-compatible, so there is $\ell$ such that (modulo permuting
coordinates) we have $g_\ell(1)\in X$ while $g_\ell(2)\in Y$. Therefore
$\X(g_{\ell}(1)) = g_{\ell}(1)$ and $\X(g_{\ell}(2))\neq g_{\ell}(2)$.
Consider
\[
  f\pmat{ g_1(1), & \cdots & \X(g_\ell(1)), & \cdots & g_n(1) \\
          g_1(2), & \cdots & g_\ell(2),     & \cdots & g_n(2) \\
          g_1(2), & \cdots & \X(g_\ell(2)), & \cdots & g_n(2) }
  = \pmat{\alpha(1) \\ \alpha(2) \\ \beta}.
\]
Each of the input vectors lies in the relation $\chi$ from
Definition~\ref{defn:mu_chi}, so the output lies in $\chi$ as well. The
relation $\chi$ has the property that if $r\in \chi$ and $r(1), r(2) \not\in
X$ then $r(3)\not\in X$. Since $\alpha(1), \alpha(2) \not\in X$ we have
$\beta\not\in X$, so by the definition of $\chi$ we now have $\alpha(2) =
\beta$. Projecting the above equality onto coordinates $\{1,3\}$ yields the
conclusion of the lemma.
\end{proof}  
\begin{prop} \label{prop:not_chi_entailed} 
Assume that $m \geq 4$ and
\begin{itemize}
  \item $\Rel_{\leq m-1}(\AAM) \Entails f$ and $f$ is $n$-ary,

  \item $G = \{g_1,\dots,g_n\} \subseteq \AM^m$ and $\RR = \Sg_{\AAM^m}(G)$
    is computational,

  \item $f(g_1,\dots,g_n) = \alpha$, and

  \item $K \subseteq [m]\setminus \D(\RR)$ and $\alpha(K)\in Y^K$.
\end{itemize}
If $G$ is not $\chi$-compatible over $K$ then $\alpha\in R$.
\end{prop}
\begin{proof}
Observe that $G$ being $\chi$-compatible over $K$ means that $G(K)$ is
$\chi$-compatible and that $G(K)\subseteq E^K$. If $|K| \leq 1$ then $G$ is
always $\chi$-compatible. Assume therefore that $|K| \geq 2$.

The proof is by induction on the number of coordinates which are $Y$ (i.e.\
not in $X$) in $G$:
\[
  \sum_{i=1}^m \Big| \big\{ k \mid g_k(i)\in Y,\ g_k\in G \big\} \Big|.
\]
If this quantity is $0$ then $G(K)\subseteq X^K$. Choose some $k\in K$.
Since $X\leq \AAM$ we have $f(g_1,\dots,g_n)(k) = \alpha(k)\in X$,
contradicting $\alpha(K)\in Y^K$. This establishes the basis of the
induction.

If $G$ fails to be $\chi$-compatible then there is $g_\ell\in G$ and
coordinates $j, k\in K$ such that $g_{\ell}(j)\in X$ while $g_{\ell}(k)\in
Y$. Define
\[
  \widehat{g}_{\ell}(i) 
  = \begin{cases}
    \X(g_{\ell}(k)) & \text{if } i = k, \\
    g_{\ell}(i) & \text{otherwise},
  \end{cases}
  \qquad\text{and}\qquad
  \cl{E} = \big\{ g_1, \dots, \widehat{g}_{\ell}, \dots, g_n \big\}.
\]
Since $G(\{j,k\})$ is not $\chi$-compatible,
Lemma~\ref{lem:not_chi_entailed_helper} implies that 
\[
  f(g_1,\dots, \widehat{g}_\ell, \dots, g_n) = \alpha.
\]
The arguments have $1$ fewer coordinates in $Y$, so $\alpha\in
\Sg_{\AAM^m}(\cl{E})$. Hence there is a term operation $t$ that generates
$\alpha$ from $\cl{E}$. Consider the equation
\[
  t\pmat{ g_1(j), & \cdots & \X(g_\ell(j)), & \cdots & g_n(j) \\
          g_1(k), & \cdots & g_\ell(k),     & \cdots & g_n(k) \\
          g_1(k), & \cdots & \X(g_\ell(k)), & \cdots & g_n(k) }
  = \pmat{ \alpha(j) \\ \gamma \\ \alpha(k) }.
\]
Projecting the arguments on coordinates $\{1,2\}$ yields $G(\{j,k\})$ and on
$\{1,3\}$ yields $\cl{E}(\{j,k\})$. Since $t$ is a term operation and all
the input vectors lie in $\chi$, the output must as well. The relation
$\chi$ has the property that if $r\in \chi$ and $r(3)\not\in X$ then
$r(2)\not\in X$. Since $\alpha(k)\not\in X$, we have $\gamma\not\in X$, and
by the definition of $\chi$ we conclude that $\alpha(k) = \gamma$.
Projecting on coordinates $\{1,2\}$ now yields $t(g_1,\dots,g_n) = \alpha$,
so $\alpha\in R$. This completes the induction and the proof.
\end{proof}  
\begin{thm} \label{thm:few_D_entailment}  
Assume that $m$ is such that
\begin{itemize}
  \item $\Rel_{\leq m-1}(\AAM)\Entails f$ and $f$ is $n$-ary,

  \item $G = \{g_1,\dots,g_n\}\subseteq \AM^m$ and $\RR = \Sg_{\AAM^m}(G)$
    is computational,

  \item $f(g_1,\dots,g_n) = \alpha\in Y^m$,

  \item $\big| [m]\setminus \D(\RR) \big| \geq 11$, and

  \item there is at most one $k\not\in \N(\RR)$ such that
    $\ApproxI(\alpha,k)$.
\end{itemize}
Then $\alpha\in R$.
\end{thm}
\begin{proof}
By Proposition~\ref{prop:not_chi_entailed}, if $G$ is not $\chi$-compatible
then $\alpha\in R$. Assume therefore that $G$ is $\chi$-compatible and let
\[
  K
  = [m]\setminus \Big\{ k \mid k\in \D \text{ or }
    \big[ k\not\in \N \text{ and } \ApproxI(\alpha,k) \big] \Big\}.
\]
The hypotheses of the theorem mean that $|K|\geq 10$ and that $G$ is
$\chi$-compatible on $K$. Since $\alpha(K)\in Y^K$, it follows that one of
the sets
\[
  \alpha^{-1}(\{0\}) \cap K,
  \qquad \qquad
  \alpha^{-1}(\{A\}) \cap K,
  \qquad \qquad
  \alpha^{-1}(\{B\}) \cap K
\]
contains $4$ elements. Let us suppose that $\alpha^{-1}(\{A\}) \cap K$ has
$4$ elements, call them 1, 2, 3, 4. The argument that follows applies
equally well to the other possibilities. We will closely examine $f$
evaluated on these coordinates. 

We have $f(g_1,\dots,g_n) = \alpha$. Evaluation at a coordinate is just
evaluation on a ``row'' of this equation. For $i\in [n]$, define the length
$n$ tuples $\Row{i} = (g_1(i),\dots,g_n(i))$ and note that $\Row{i}(j) =
g_j(i)$. For distinct $i_1,\dots, i_k\in [n]$, define the length $n$ tuples
\begin{align*}
  \Row{i_1|\cdots|i_k}_{\forall} (j)
  &= \begin{cases}
    g_j(i_1) & \text{if } g_j(i_1) = \cdots = g_j(i_k), \\
    \vect{\State(g_j),0} & \text{otherwise},
  \end{cases} \\
  \Row{i_1|\cdots|i_k}_{\exists A} (j)
  &= \begin{cases}
    g_j(i_1) & \text{if } g_j(i_1) = \cdots = g_j(i_k), \\
    g_j(i_\ell) & \text{if } \Content(g_j(i_\ell)) = A \text{ for some }
      i_\ell\in \{i_1,\dots,i_k\}, \\
    \vect{\State(g_j),0} & \text{otherwise},
  \end{cases} \\
  \Row{i_1|\cdots|i_k}_{\exists B} (j)
  &= \begin{cases}
    g_j(i_1) & \text{if } g_j(i_1) = \cdots = g_j(i_k), \\
    g_j(i_\ell) & \text{if } \Content(g_j(i_\ell)) = B \text{ for some }
      i_\ell\in \{i_1,\dots,i_k\}, \\
    \vect{\State(g_j),0} & \text{otherwise}.
  \end{cases}
\end{align*}
We claim that
\[
  f\pmat{ \Row{1|2|3|4}_{\forall} \\
          \Row{1|2|3|4}_{\exists A} \\
          \Row{1|2|3|4}_{\exists B} }
  = \pmat{ \vect{\State(\alpha),A} \\
           \vect{\State(\alpha),A} \\
           \vect{\State(\alpha),A} }
  = \alpha(\{1,2,3\}).
\]
Using the relations $\Delta_\forall$, $\Delta_{\exists A}$, and
$\Delta_{\exists B}$ it is not difficult to see that this is true. It is,
however, most easily seen by working through an example. See
Figure~\ref{fig:delta_existsA} for an example showing
$f(\Row{1|2|3|4}_{\exists A}) = \vect{\State(\alpha),A}$.

\begin{figure}  
\scriptsize
\begin{align*}
  &f\pmat{ \Row{1} \\ \Row{2} \\ \Row{3} \\ \Row{4} }
  = f\pmat{ \Er & A & B & 0 & B & A \\
            \Er & 0 & B & B & A & A \\
            \Er & A & B & 0 & B & A \\
            \Er & 0 & B & B & 0 & A }
  = \pmat{ A \\ A \\ A \\ A } \\[1em]
  &\Longrightarrow f\pmat{ \Row{1} \\ \Row{2} \\ \Row{1|2}_{\exists A} }
  = f\pmat{ \Er & A & B & 0 & B & A \\
            \Er & 0 & B & B & A & A \\
            \Er & A & B & 0 & A & A }   
  = \pmat{ A \\ A \\ c_1 } \\[1em]
  &\Longrightarrow f\pmat{ \Row{1|2}_{\exists A} \\ \Row{3} \\ \Row{1|2|3}_{\exists A} }
  = f\pmat{ \Er & A & B & 0 & A & A \\
            \Er & A & B & 0 & B & A \\
            \Er & A & B & 0 & A & A }
  = \pmat{ A \\ A \\ c_2 } \\[1em]
  &\Longrightarrow f\pmat{ \Row{1|2|3}_{\exists A} \\ \Row{4} \\ \Row{1|2|3|4}_{\exists A} }
  = f\pmat{ \Er & A & B & 0 & A & A \\
            \Er & 0 & B & B & 0 & A \\
            \Er & A & B & 0 & A & A }
  = \pmat{ A \\ A \\ c_3 }
\end{align*}
\caption{The argument showing $f(\Row{1|2|3|4}_{\exists A}) =
  \vect{\State(\alpha),A}$. For brevity, we show only the content of the
  vectors ($\RR$ is synchronized, so state in a vector is constant). In all
  cases, $c_i = A$ since the argument columns are in $\Delta_{\exists A}$
  and $\Delta_{\exists A}\Entails f$.}
\label{fig:delta_existsA}
\end{figure}   

Define vectors $h_i\in \AAM^{m-1}$ by
\[
  h_i(j)
  = \begin{cases}
    \Row{1|2|3|4}_{\forall}(i) & \text{if } j = 1, \\
    \Row{1|2|3|4}_{\exists A}(i) & \text{if } j = 2, \\
    \Row{1|2|3|4}_{\exists B}(i) & \text{if } j = 3, \\
    g_i(j) & \text{if } j\not\in \{1,2,3,4\},
  \end{cases}
\]
let $\cl{E} = \{ h_1,\dots, h_n \}$, and let $\SS =
\Sg_{\AAM^{m-1}}(\cl{E})$. From the previous paragraph, we have that
\[
  f(h_1,\dots, h_n)(j)
  = \begin{cases}
    \alpha(j) & \text{if } j\in \{1,2,3\}, \\
    \alpha(j) & \text{if } j\not\in \{1,2,3,4\},
  \end{cases}
\]
so $f(h_1,\dots,h_n) = \alpha(\neq 4)$. Since $\Rel_{\leq m-1}(\AAM)\Entails
f$, we have that $f$ preserves $\SS$. Therefore $\alpha(\neq 4)\in S$, so
there is a term operation $t$ such that $t(h_1,\dots,h_n) = \alpha(\neq 4)$. 

We chose $1,2,3,4$ from $K$, and $K$ does not include any coordinates $k$
for which $\ApproxI(\alpha,k)$ holds for $\RR$. Since $K$ is also 
disjoint from $\D(\RR)$, we have that $\SS_I = \RR_I(\neq 4)$. Therefore
$\alpha(\neq 4)\not\in R_I$, and so by Proposition~\ref{prop:I_avoidance},
we can assume that the term operation $t$ does not have $I$ in its term tree
and hence respects the relation $\Gamma$ from Definition~\ref{defn:gamma} by
Lemma~\ref{lem:relns_reduct}.

\begin{figure}  
\scriptsize
\begin{align*}
  &\Row{1} = \pmat{ \Er & A & B & 0 & B & A } \\
  &\Row{2} = \pmat{ \Er & 0 & B & B & A & A } \\
  &\Row{3} = \pmat{ \Er & A & B & 0 & B & A } \\
  &\Row{4} = \pmat{ \Er & 0 & B & B & 0 & A } \\
  &\Longrightarrow t\pmat{ \Row{1|2|3|4}_{\forall} \\
                           \Row{1|2|3|4}_{\exists A} \\
                           \Row{1|2|3|4}_{\exists B} \\
                           \Row{1} }
  = t\pmat{ \Er & 0 & B & 0 & 0 & A \\
            \Er & A & B & 0 & A & A \\
            \Er & 0 & B & B & B & A \\
            \Er & A & B & 0 & B & A }
  = \pmat{ A \\ A \\ A \\ c }
\end{align*}
\caption{The argument showing $t(\Row{1}) = \vect{\State(\alpha),A}$. For
  brevity, we show only the content of the vectors. We have $c = A$ 
  because the argument columns are in $\Gamma$ and the term operation $t$
  preserves $\Gamma$ since it does not have $I$ in its term tree.}
\label{fig:gamma}
\end{figure}   

We will use $\Gamma$ to show that $t(\Row{j}) = \alpha(j)$ for $j\in
\{1,2,3,4\}$. As $t(h_1,\dots,h_n)(j) = \alpha(j)$ for $j\not\in
\{1,2,3,4\}$ already, this will finish the proof. Again, this is most easily
seen by example --- see Figure~\ref{fig:gamma}. The vectors $h_1(\{1,2,3\}),
\dots, h_n(\{1,2,3\})$ make up the first three rows of typical elements of
$\Gamma$. Carefully examining $\Gamma$, we see that we can complete the
$h_i(\{1,2,3\})$ to elements of $\Gamma$ in many ways while keeping $t$
constant on this new row. Due to how $\Row{1|2|3|4}_{\forall}$,
$\Row{1|2|3|4}_{\exists A}$, and $\Row{1|2|3|4}_{\exists B}$ were defined,
there are completions that equal each of $\Row{1}$, $\Row{2}$, $\Row{3}$,
and $\Row{4}$. Thus $t(\Row{j}) = \vect{\State(\alpha),A}$ and hence
$t(g_1,\dots,g_n) = \alpha$, so $\alpha\in R$, as claimed.
\end{proof}  

The same approach used to prove the above theorem can also be used to prove
entailment when $\Er\in \Content(\alpha)$. We do this in the next theorem.

\begin{thm} \label{thm:X_entailment} 
Assume that $m\geq 11$,
\begin{itemize}
  \item $\Rel_{\leq m-1}(\AAM)\Entails f$ and $f$ is $n$-ary,

  \item $G = \{g_1,\dots,g_n\}\subseteq \AM^m$,

  \item $\RR = \Sg_{\AAM^m}(G)$ is computational non-halting,

  \item $f(g_1,\dots,g_n) = \alpha$ and $\Er\in \Content(\alpha)$.
\end{itemize}
Then $\alpha\in R$.
\end{thm}
\begin{proof}
If $G$ is not $\chi$-compatible then $\alpha\in R$ by
Proposition~\ref{prop:not_chi_entailed}. Assume therefore that $G$ is
$\chi$-compatible and assume towards a contradiction that $\alpha\not\in R$.
We have that $\Rel_{\leq m-1}(\AAM)\Entails f$ and $\RR\leq\AAM^m$, so $\RR$
has $\Approx(\alpha,i)$ for all $i$. In
Definition~\ref{defn:approx_element}, we fixed elements $\alpha_i\in R$
witnessing this. We will make use of these elements in the argument to
follow.

Suppose that there are two distinct coordinates $k,\ell$ such that
$\alpha(k), \alpha(\ell)\in X$. In this case
\[
  \alpha = \alpha_k \wedge \alpha_\ell,
\]
so $\alpha\in R$. Therefore there must be a unique coordinate $k$ such that
$\alpha(k)\in X$ and $\alpha_k\in Y^m$. We will use this coordinate in the
following analysis.

Suppose that there is $\ell$ such that $\alpha(\ell) \in D$. It
follows from the definition that
\[
  \alpha = N_{\Dot}( \alpha_k, \alpha_\ell, \alpha_k, \alpha_k ),
\]
so $\alpha\in R$. Therefore $\alpha\in E^m$. Since $\alpha_k\in Y^m$ and
$\RR$ is non-halting, it must be that $\D(\RR(\neq k)) = \emptyset$, by
Lemma~\ref{lem:capacity} item (2).

Suppose that $R(k)\cap D\neq \emptyset$. Choose $d'\in R$ such that $d'(k)
\in D$ and let $d = I(d',\alpha_k)$. It follows that $d(k)\in D$ and $d(\neq
k)\in C^{m-1}$, so $z_0(d)(k) \in X$ and $z_0(d)(\neq k) = \vect{0,0}$ by
Lemma~\ref{lem:basic_facts_AM} item (4). Using $N_0$ we now have
\[
  N_0(z_0(d),\alpha_k,\alpha_k)(j) = 
  \begin{cases}
    \alpha_k(j) & \text{if } j\neq k \\
    \vect{\State(\alpha),\Er} & \text{otherwise}.
  \end{cases}
\]
Since $\alpha(k)\in X$, it follows that $N_0(z_0(d),\alpha_k,\alpha_k) =
\alpha$ and hence $\alpha\in R$. Therefore it must be that $R(k)\cap D =
\emptyset$. Combining this with the previous paragraph, we have $\D(\RR) =
\emptyset$.

At this point, the analysis becomes quite similar to that performed in
Theorem~\ref{thm:few_D_entailment}. Let $K = [m] \setminus \{k\}$ and find 4
distinct values, call them $1,2,3,4\in K$, such that $\alpha$ has a common
value on these coordinates (we use $|K|\geq 10$ here). Using $f$ and $G$,
produce the row tuples $\Row{1|2|3|4}_{\forall}$, $\Row{1|2|3|4}_{\exists
A}$, and $\Row{1|2|3|4}_{\exists B}$. As before, we have
\[
  f\pmat{ \Row{1|2|3|4}_{\forall} \\
          \Row{1|2|3|4}_{\exists A} \\
          \Row{1|2|3|4}_{\exists B} }
  = \pmat{ \alpha(1) \\ \alpha(1) \\ \alpha(1) }
  = \alpha(\{1,2,3\}).
\]
Form $h_i\in \AM^{m-1}$ from the $g_i$ by
\[
  h_i(j)
  = \begin{cases}
    \Row{1|2|3|4}_{\forall}(i) & \text{if } j = 1, \\
    \Row{1|2|3|4}_{\exists A}(i) & \text{if } j = 2, \\
    \Row{1|2|3|4}_{\exists B}(i) & \text{if } j = 3, \\
    g_i(j) & \text{if } j\not\in \{1,2,3,4\},
  \end{cases}
\]
let $\cl{E} = \{ h_1,\dots, h_n \}$, and let $\SS =
\Sg_{\AAM^{m-1}}(\cl{E})$. From the previous paragraph, we have
\[
  f(h_1,\dots, h_n)(j)
  = \begin{cases}
    \alpha(j) & \text{if } j\in \{1,2,3\}, \\
    \alpha(j) & \text{if } j\not\in \{1,2,3,4\},
  \end{cases}
\]
so $f(h_1,\dots,h_n) = \alpha(\neq 4)$. Since $\Rel_{\leq m-1}(\AAM)\Entails
f$, we have that $f$ preserves $\SS$. Therefore $\alpha(\neq 4)\in S$, so
there is a term operation $t$ such that $t(h_1,\dots,h_n) = \alpha(\neq 4)$.
There is a difficulty in continuing as we did in the proof of
Theorem~\ref{thm:few_D_entailment}, however: we cannot assume that $I$ does
not appear in the term tree of $t$ since $\alpha\not\in Y^m$, and so we
cannot make use of the relation $\Gamma$. It turns out that this difficulty
is not insurmountable.

\begin{claim*}
There is a term operation $s$ without $I$ in its term tree such that
\[
  s(h_1,\dots,h_n)(\neq k) \in C^{m-2}
  \qquad\text{and}\qquad
  s(h_1,\dots,h_n)(k) \in X.
\]
\end{claim*}
\begin{claimproof}
We begin by making some observations. From Lemma~\ref{lem:RI_off_dot}, we have
that
\[
  S_I
  = \bigcup_{i \text{ a state}} 
  \left\{ \pmat{ \vect{i, 0} \\ \vdots \\ \vect{i, 0} }
          \pmat{ \vect{i, \Er} \\ \vdots \\ \vect{i, \Er} } \right\}.
\]
We will say that the element $a\in S$ \emph{avoids} $I$ if there is a term
operation $s$ without $I$ in its term tree such that $s(\overline{h}) = a$.
From Proposition~\ref{prop:I_avoidance} and our observation about $S_I$
above, we have that if $b\in S\cap Y^{m-1}$ and $b$ does not avoid $I$ then
$\Content(b) = \{0\}$. We are now ready to prove the claim. 

As usual, we will proceed by induction. Let $G_0 = \cl{E}$ be the generators
of $\SS$ and
\[
  G_n
  = \Big\{ F(\overline{b}) \mid F \text{ a fundamental $\ell$-ary
      operation},\ \overline{b}\in G_{n-1}^\ell \Big\}
    \cup G_{n-1}.
\]
Choose $n$ minimal such that there is $a\in G_n$ with $a(\neq k)\in C^{m-1}$
and $a(k)\in X$ (from the paragraph prior to the claim, we know that
$t(\overline{h})$ is such an element). If $a$ avoids $I$ then we are done,
so assume that $a$ does not avoid $I$. We will prove that there exists an
element $a'\in S$ which avoids $I$ and has $a'(\neq k)\in C^{m-2}$ and
$a'(k)\in X$. If $a\in G_0 = \cl{E}$ then $a$ avoids $I$, so we are done.
Assume that $n>0$, so
\[
  a = F\big( b_1,\dots,b_\ell \big)
\]
for some $\ell$-ary operation $F$ and elements $b_1,\dots,b_{\ell}\in
G_{n-1}$. We proceed by cases depending on which operation $F$ is. The cases
for $F = I$ and $F = P$ are quite straightforward (using $\D(\SS) =
\emptyset$ for $F = I$), and so we omit them.

\Case{$F \in \{\wedge, M, N_{\Dot}\}$}
If $a = b \wedge c$ then $b(\neq k) = c(\neq k) = a(\neq k)\in C^{m-1}$, so
by minimality of $n$ we have $b(k),c(k)\in C$ and $b(k)\neq c(k)$. Hence
$b,c\in Y^{m-1}$. The element $a$ does not avoid $I$, so one of $b$ or $c$
does not. Without loss of generality, suppose that $b$ does not avoid $I$.
From Proposition~\ref{prop:I_avoidance} and the observation about $S_I$
above, it follows that $\Content(b) = \{0\}$. Since $c(\neq k) = b(\neq k)$
and $c(k)\neq b(k)$, it must be that $\Content(c(\neq k)) = \{0\}$ and
$\Content(c(k))\neq 0$. It follows that $c\not\in S_I$ and thus $c$ avoids $I$.
Let $w_{\State(c)}$ be the term from Lemma \ref{lem:basic_facts_AM} item (5).
The element $a' = w_{\State(c)}(c)$ therefore avoids $I$ and has $a'(\neq k)\in
C^{m-2}$ and $a'(k)\in X$. The analysis for $F = M$ is almost identical and $F =
N_{\Dot}$ similarly reduces since $\D(\SS) = \emptyset$.

\Case{$F\in \{M', H\}$}
If $a = M'(b)$ then $b$ does not avoid $I$ and by minimality of $n$ we have
$b\in C^{m-1}$. By Proposition~\ref{prop:I_avoidance} and the observations
about $S_I$ above, it must be that $\Content(b) = \{0\}$, but then
$M'(b)(\neq k)\in C^{m-1}$ and $M'(b)(k)\in X$ is impossible. The case for
$F = H$ is similar.

\Case{$F = N_0$}
If $a = N_0(b,c,d)$ then $b(\neq k)\in \{ \vect{0,0} \}^{m-1}$, $d(\neq k) =
a(\neq k)\in C^{m-2}$, and either $d(k)\in X$ or $b(k) \neq \vect{0,0}$. The
possibility where $d(k)\in X$ contradicts the minimality of $n$, so it must be
that $b(k)\neq \vect{0,0}$. It follows that $b\not\in S_I$. If $b(k)\in X$ then
the minimality of $n$ is contradicted again, so it must be that $b\in Y^{m-1}$
and hence avoids $I$. The element $a' = w_0(b)$ (where $w_0$ is from Lemma
\ref{lem:basic_facts_AM} item (5)) therefore avoids $I$ and has $a'(\neq k)\in
C^{m-2}$ and $a'(k)\in X$.

\Case{$F = S$}
If $a = S(b,c,d)$ then $b(\neq k) = c(\neq k) = d(\neq k) = a(\neq k) \in \{
\vect{1,0} \}^{m-2}$ and one of $b(k),c(k),d(k)$ is not equal to $\vect{1,0}$.
The analyses for each of these possibilities are quite similar, so we will
examine $b(k)\neq \vect{1,0}$ and leave the others to the reader. By the
minimality of $n$ we have $b(k)\not\in X$, so $b\in Y^{m-1}$ and
$\Content(b(k))\neq 0$. Thus $b\not\in S_I$ and we have that $b$ avoids $I$. The
element $a' = w_1(b)$ satisfies the claim, where $w_1$ is the term from Lemma
\ref{lem:basic_facts_AM} item (5).

\smallskip

In all cases, we have produced an element $a'\in S$ which avoids $I$ and has
$a'(\neq k)\in C^{m-2}$ and $a'(k)\in X$, proving the claim.
\end{claimproof}

Apply the above claim to the term operation $t$ to produce a new term
operation $s$ without $I$ in its term tree such that $s(h_1,\dots,h_n)(\neq
k) \in C^{m-2}$ and $s(h_1,\dots,h_n)(k)\in X$. Since $s$ does not have $I$
in its term tree, it respects $\Gamma$, and so as in
the proof of Theorem~\ref{thm:few_D_entailment} we obtain
\[
  s(g_1,\dots,g_n)(\neq k)\in C^{m-1}
  \qquad\text{and}\qquad
  s(g_1,\dots,g_n)(k)\in X.
\]
Let $r = s(g_1,\dots,g_n)$. As in the fourth paragraph of the proof, it
follows that $N_0(z_0(r),\alpha_k,\alpha_k) = \alpha$, so $\alpha\in R$.
This completes the proof.
\end{proof} 

\subsection{Entailment for everything else}
Finally, we prove that relations not ruled out by the previous entailment
theorems are also entailed. This is the result that we have been building
towards. We begin by proving an extension of
Proposition~\ref{prop:I_avoidance}.

\begin{prop} \label{prop:double_I_avoidance}  
Let $\RR\leq \AAM^m$ be computational non-halting. If $\alpha\in Y^m$ is
such that
\begin{itemize}
  \item $K = \N(\RR) \cup \big\{ i\mid \RR \text{ has } 
    \ApproxI(\alpha, i) \big\}$, and

  \item $\big| [m]\setminus K \big| \geq 3$
\end{itemize}
then for all $k\not\in K$ there is an $\ell\not\in K\cup \{k\}$ such that
$\alpha(\neq k,\ell)\not\in \big( \RR(\neq k,\ell) \big)_I$.
\end{prop}
\begin{proof}
Let $k\not\in K$, $\D_I = \D(\RR_I)$, and $L = [m]\setminus \D_I$. By
Lemma~\ref{lem:inherent_properties} item (3), we have that $k\not\in \D_I$.
Furthermore, $\alpha\not\in R_I$ since otherwise we would have $[m] = K$.
Since $\alpha\not\in R_I$ and $\alpha\in Y^m$, by Lemma~\ref{lem:RI_off_dot}
we have that either
\begin{itemize}
  \item $\alpha(K)\not\in R_I(K)$ or
  \item $\Content(\alpha(i))\not\in \{0,\Er\}$ for some $i\not\in K$.
\end{itemize}
If we are in the first situation then for \emph{all} $k,\ell\not\in K$ we
have $\alpha(\neq k,\ell)\not\in \big( R(\neq k,\ell) \big)_I$. Assume
therefore that $\alpha(K)\in R_I(K)$ and that we are in the second
situation. Fix $i\not\in K$ such that $\Content(\alpha(i))\not\in
\{0,\Er\}$. From the definition of $K$, it is not possible for $\RR$ to have
$\ApproxI(\alpha,i)$, so it must be that there is some $j\not\in K$ distinct
from $i$ such that $\Content(\alpha(j))\not\in \{0,\Er\}$. We have that
$|[m]\setminus K|\geq 3$, so it follows that for every $k\not\in K$ there is
an $\ell\not\in K$ distinct from $k$ such that $\Content(\alpha(\ell))
\not\in \{0,\Er\}$ (just choose $\ell = i$ or $\ell = j$). By
Lemma~\ref{lem:RI_off_dot} and since $k,\ell\not\in K$, this is enough to
give us $\alpha(\neq k,\ell)\not \in \big( R(\neq k,\ell) \big)_I$.
\end{proof}  

The next three lemmas are technical, but form the core of the argument in
the entailment theorem in this section. The first of these technical lemmas
is a kind of extension of Lemma~\ref{lem:f_E_monotone}.

\begin{lem} \label{lem:f_weak_monotone}  
Assume the following:
\begin{itemize}
  \item $t$ is an $n$-ary term operation,

  \item $G = \{g_1,\dots,g_n\} \subseteq \AM^m$ and 
    $\cl{E} = \{ e_1,\dots,e_n \} \subseteq \AM^m$,

  \item $\RR = \Sg_{\AAM^m}(G)$ is computational non-halting,

  \item $\SS = \Sg_{\AAM^m}(\cl{E})$ is computational,

  \item $k\in [m]$,

  \item for each $i\in [n]$ we have $e_i(\neq k) = g_i(\neq k)$ and $e_i(k)
    \leq g_i(k)$, and

  \item $t(\overline{e}) = \alpha$ and $\alpha(k)\in Y$.
\end{itemize}
Then there exists a term operation $s$ such that $\alpha\leq
s(\overline{g})$.
\end{lem}
\begin{proof}
We begin with a less formal statement of the lemma. View $G$ and $\cl{E}$ as
$m\times n$ matrices. We obtain $\cl{E}$ from $G$ by replacing the content
of some entries in the $k$-th row with $\Er$. The lemma asserts that if
$\alpha\in S$ has row $k$ in $Y$ then it is less than or equal to some
element in $\RR$.

Observe that $\RR$ being non-halting implies $\SS$ is non-halting. As usual,
the proof shall be by induction on the complexity of $t$. If $t$ is a
projection then $\alpha = e_i$ for some $i$, so $\alpha \leq g_i$. Assume
now that
\[
  t(\overline{x})
  = F\big( f_1(\overline{x}), \dots, f_\ell(\overline{x}) \big)
\]
for some $\ell$-ary fundamental operation $F$ and $n$-ary term operations
$f_i$. We will proceed by cases depending on $F$.

\Case{$F\in \{\wedge, N_0, N_\Dot, P\}$}
Such $F$ have the property that $F(\overline{a})\leq a_i$ for some $a_i$
among the $\overline{a}$, by Lemma~\ref{lem:basic_facts_relns} and since
$\SS$ and $\RR$ are computational non-halting. If $\alpha =
F(f_1(\overline{e}), \dots, f_\ell(\overline{e})) \leq f_i(\overline{e})$
then $f_i(\overline{e})(k) = \alpha(k)\in Y$, so the inductive hypothesis
applies. Thus there is $h_i$ such that $\alpha \leq f_i(\overline{e}) \leq
h_i(\overline{g})$.

\Case{$F\in \{M, M', H, S\}$}
Such $F$ are $X$-absorbing, by Lemma~\ref{lem:basic_facts_AM} item (2).
Therefore, if $\alpha = F(f_1(\overline{e}), \dots, f_\ell(\overline{e}))$
then $f_i(\overline{e})(k)\in Y$ for all $i$. The inductive hypothesis
applies, so there are $h_i$ such that $f_i(\overline{e}) \leq
h_i(\overline{g})$. It follows that $\alpha\leq F(h_1(\overline{g}), \dots,
h_\ell(\overline{g}))$.

\smallskip

The remaining (and most complicated) case is $F = I$. Suppose that $\alpha =
I(f_1(\overline{e}), f_2(\overline{e}))$. Since $\alpha(k)\in Y$, either
$\alpha(k)\in D$ or $\alpha(k)\in C$. We will examine these possibilities in
their own cases.

\Case{$F = I$, $\alpha(k)\in D$}
If $\alpha(k) \in D$ then $f_1(\overline{e})(k)\in D$, so there is a term
operation $h_1$ such that $f_1(\overline{e}) \leq h_1(\overline{g})$. This
implies $h_1(\overline{g})(k)\in D$. Since $\RR$ is computational and $I$
depends on its first input only at those coordinates with content $\Dot$, we
have $\alpha = I(h_1(\overline{g}), f_2(\overline{g}))$.

\Case{$F = I$, $\alpha(k)\in C$}
If $\alpha(k)\in C$ then $f_2(\overline{e})(k)\in C$, so there is a term
operation $h_2$ such that $f_2(\overline{e}) \leq h_2(\overline{g})$. If
$\Dot\not\in \Content(\alpha)$ then $\alpha\leq I(h_2(\overline{g}),
h_2(\overline{g}))$, and we are done. If, on the other hand, $\Dot\in
\Content(\alpha)$ then there is $j\neq k$ such that $\alpha(j)\in D$. This
implies that $f_1(\overline{e})(j)\in D$. It follows that
$f_1(\overline{g})(k)\not\in D$ since $\RR$ is computational. From the
definition of $I$ we now have $\alpha\leq I(f_1(\overline{g}),
h_2(\overline{g}))$.
\end{proof} 
\begin{lem} \label{lem:main_entailment_content} 
Assume the following:
\begin{itemize}
  \item $G = \{g_1,\dots,g_n\}\subseteq \AM^{p-1}$,

  \item $\RR = \Sg_{\AAM^{p-1}}(G)$ is computational non-halting,

  \item $t$ is an $n$-ary term operation without $I$ in its term tree,

  \item $t(g_1,\dots,g_n) = \alpha\in Y^{p-1}$, and

  \item $k\not\in \N(\RR)$ is such that $\alpha(k)\not\in D$.
\end{itemize}
Define elements of $e_i\in \AM^p$ for $i\in [n]$ by
\begin{align*}
  e_i(j) &= \begin{cases}
    g_i(j) & \text{if } j\in [p-1], \\
    g_i(k) & \text{if } j = p \text{ and }
      \Content(g_i(k)) \in \{ \Content(\alpha(k)), \Dot \}, \\
    \X(g_i(k)) & \text{otherwise},
  \end{cases} \\
  \beta(j) &= \begin{cases}
    \alpha(j) & \text{if } j\in [p-1], \\
    \alpha(k) & \text{if } j = p.
  \end{cases}
\end{align*}
Then $t(e_1,\dots,e_n) = \beta$.
\end{lem}
\begin{proof}
We begin with a less formal statement of the lemma. View $G$ as a
$(p-1)\times n$ matrix. Copy row $k$ of this matrix and put it at the
bottom, making a $p\times n$ matrix. In row $p$ (the new row), for each
entry with content not either $\Dot$ or $\Content(\alpha(k))$, replace that
content with $\Er$. Call the resulting vectors $e_1,\dots,e_n$. The Lemma
asserts that if the copied row $k$ is not in $\N(\RR)$, $\alpha\in Y^{p-1}$,
and $\alpha(k)\not\in D$ then $t(\overline{e})$ is just the vector $\alpha$
with the $k$-th row copied to the bottom.

Let $\cl{E} = \{ e_i \mid i\in [n] \}\subseteq \AM^p$ and $\SS =
\Sg_{\AAM^p}(\cl{E})$. Let us make some observations about $\SS$:
\begin{itemize}
  \item $\SS(\neq p) = \RR$,

  \item $\SS$ \emph{need not} be computational, but it is synchronized,

  \item $\SS(\N(\RR)) = \RR(\N(\RR))$ is non-halting, so $\SS$ is
    non-halting as well,

  \item for each $s\in S$ there is an $\ell\in \N(\RR)$ such that
    $s(\ell)\in D\cup X$, by Lemma~\ref{lem:inherent_properties} item (1).
\end{itemize}
Now let us examine $\alpha$ and $\beta$. Let $\beta_p = t(e_1,\dots,e_n)$
and note that $\beta_p(\neq p) = \beta(\neq p) = \alpha$ and $\Dot\in
\Content(\beta_p(\N(\RR)))$ by the last item above.

Since $t$ does not have $I$ in its term tree, we will analyze the subset of
$S$ generated by $\cl{E}$ without using $I$ in the generation. Call this
subset $S'$. Let $G_0 = \cl{E}$ and
\[
  G_n
  = \Big\{ F(\overline{h}) \mid F \text{ a fundamental $k$-ary
      operation},\ F\neq I,\ \overline{h}\in G_{n-1}^k \Big\} \cup G_{n-1}.
\]
Note that $S' = \bigcup G_n$. Since $\beta_p\in S'$, there is a least $n$
such that $\beta_p\in G_n$. We will show that $\beta_p\in G_n$ implies
$\beta\in S'$ by induction on $n$. The set $G_0 = \cl{E}$ has this property
by definition of the $e_i$, establishing the base case. Suppose now that
$n>0$, so
\[
  \beta_p = F\big( h_1,\dots,h_{\ell} \big)
\]
for some $\ell$-ary fundamental operation $F$ and $h_1,\dots,h_{\ell}\in
G_{n-1}$. We break into cases based on $F$.

\Case{$F\in \{\wedge, N_0, P\}$}
In this case, by the various parts of Lemma~\ref{lem:basic_facts_relns}, we
have that $\beta_p = F(\overline{h})\leq h_i$ for some $h_i$. Since
$\beta_p(\neq p)\in Y^{p-1}$, this implies that $\beta_p(\neq p) = h_i(\neq
p)$, so the inductive hypothesis yields $\beta = h_i\in S'$.

\Case{$F\in \{H,S\}$}
Recall that $\Dot\in \Content(\beta_p(\N(\RR)))$. Since the range of $H$ and
$S$ are disjoint from $D$, $\beta_p$ cannot be the output of one of them.

\Case{$F = M$}
Say $\beta_p = M(a,b)$. We have that $a(\neq p), b(\neq p)\in Y^{p-1}$. If
$a(k)\in D$ then $\Er\in \Content(a(\N(\RR)))$ by
Lemma~\ref{lem:inherent_properties} item (1), contradicting $a(\neq p)\in
Y^{p-1}$. Hence $a(k)\not\in D$, and so from the definition of $M$ we have
$\Content(a(k)) = \Content(b(k)) = \Content(\beta_p(k))$. Since $\beta_p(k)
= \alpha(k)$, we can use the inductive hypothesis to conclude that
$\Content(a(p)) = \Content(b(p)) = \Content(\beta(p))$. Evaluating yields
$M(a,b) = \beta$, so $\beta\in S'$.

\Case{$F = M'$}
Say $\beta_p = M'(a)$. From the definition, we have $a(\neq p)\in Y^{p-1}$
and $\Content(a(k)) = \Content(\beta_p(k))$. Immediately before the start of
the induction we observed that $\Dot\in \Content(\beta_p(\N(\RR)))$. Since
$k\not\in \N(\RR)$, it follows that $a(k)\not\in D$. The inductive
hypothesis therefore applies to $a$, and we get $\beta = M'(a)$, so $\beta\in
S'$.

\Case{$F = N_\Dot$}
Let $\beta_p = N_\Dot(a,b,c,d)$. If $|a^{-1}(D)|\leq 1$ then $\beta_p\leq
b$ or $\beta_p\leq c$ by Lemma~\ref{lem:basic_facts_relns}. Without loss of
generality say $\beta_p\leq b$. Since $\beta_p(\neq p)\in Y^{p-1}$, we have
that $\beta_p(\neq p) = b(\neq p)$ and so the inductive hypothesis gives us
$\beta = b \in S'$. From the construction of $\SS$, the only other
possibility is that $a(k) = a(p)\in D$. From the definition of $N_{\Dot}$ we
have that $b(\neq k,p) = c(\neq k,p) = \beta_p(\neq k,p)$. Every element of
$S$ has content at coordinate $p$ in $\{\Content(\beta(k)),\Dot,\Er\}$. If
$\beta_p(p)\in D$ then $\Er\in \Content(\beta_p(\N(\RR)))$ by
Lemma~\ref{lem:inherent_properties} item (1), a contradiction. Let us assume
that $\beta_p(p)\in X$, since otherwise $\beta_p = \beta$. By similar logic
we have $b(p),c(p)\not\in D$, so $b(p) = c(p)\in X$. By the contrapositive
of the inductive hypothesis, both $\Content(b(k))$ and $\Content(c(k))$ are
distinct from $\Content(\beta_p(k))$. From the definition of $N_\Dot$, this
is only possible if $\beta_p(k)\in X$, a contradiction.
\end{proof}  
\begin{lem} \label{lem:main_entailment_X} 
Assume the following:
\begin{itemize}
  \item $G = \{g_1,\dots,g_n\}\subseteq \AM^{p-1}$,

  \item $\RR = \Sg_{\AAM^{p-1}}(G)$ is computational non-halting,

  \item $K$ is a set with $\N(\RR)\subseteq K$ and $|[p-1]\setminus K| \geq
    2$,

  \item $t$ is an $n$-ary term operation without $I$ in its term tree,

  \item $t(g_1,\dots,g_n) = \alpha\in Y^{p-1}$, and

  \item for all $i\in [n]$ and each $j\not\in K$ we have
    $\Content(g_i(j))\in \{ \Content(\alpha(j)), \Dot, \Er \}$.
\end{itemize}
Fix two distinct elements $\ell_1,\ell_2\not\in K$ and define elements of
$e_i\in \AM^p$ for $i\in [n]$ by
\begin{align*}
  e_i(j) &= \begin{cases}
    g_i(j) & \text{if } j\in [p-1], \\
    g_i(\ell_1) & \text{if } j = p \text{ and } g_i(\ell_1), g_i(\ell_2)\not\in D\cup X, \\
    g_i(\ell_1) & \text{if } j = p \text{ and } g_i(\ell_1)\in D, \\
    \X(g_i(\ell_1)) & \text{otherwise},
  \end{cases} \\
  \beta(j) &= \begin{cases}
    \alpha(j) & \text{if } j\in [p-1], \\
    \alpha(\ell_1) & \text{if } j = p.
  \end{cases}
\end{align*}
Then $t(e_1,\dots,e_n) = \beta$.
\end{lem}
\begin{proof}
The proof is quite similar to the proof of
Lemma~\ref{lem:main_entailment_content}. The less formal statement of the
lemma is similar as well. View $G$ as a $(p-1)\times n$ matrix and fix a set
of coordinates $K$ such that outside of $K$ the content of the rows of $G$
is always in $\{\Content(\alpha(j)), \Dot, \Er\}$. Pick two such rows,
$\ell_1$ and $\ell_2$. Copy row $\ell_1$ to the bottom of the matrix, so
that it is now $p\times n$. For each entry in the new row (row $p$), if
above that entry at rows $\ell_1$ and $\ell_2$ we have content
$\Content(\alpha(k))\neq \Dot$ at row $\ell_1$ and content $\Er$ or $\Dot$
at row $\ell_2$ then replace the content of that entry in row $p$ with
$\Er$. Call the resulting vectors $e_1,\dots,e_n$. The lemma asserts that if
$\alpha\in Y^{p-1}$ then $t(\overline{e})$ is just the vector $\alpha$ with
the $\ell_1$-th row copied to the bottom.

Let $\cl{E} = \{ e_i \mid i\in [n] \}\subseteq \AM^p$ and $\SS =
\Sg_{\AAM^p}(\cl{E})$. The same observations made in the proof of
Lemma~\ref{lem:main_entailment_content} about $\SS$ hold here as well. The
most salient are that $\SS$ need not be computational and that for each
$s\in S$ there is $i\in \N(\RR)$ such that $s(i)\in D\cup X$. Let us now
examine $\alpha$ and $\beta$. Let $\beta_p = t(e_1,\dots,e_n)$ and note that
$\beta_p(\neq p) = \beta(\neq p) = \alpha$.

Since $t$ does not have $I$ in its term tree, we will analyze the subset of
$S$ generated by $\cl{E}$ without using $I$ in the generation. Call this
subset $S'$. Let $G_0 = \cl{E}$ and
\[
  G_n
  = \Big\{ F(\overline{g}) \mid F \text{ a fundamental $k$-ary
      operation},\ F\neq I,\ \overline{g}\in G_{n-1}^k \Big\} \cup G_{n-1}.
\]
Note that $S' = \bigcup G_n$. Since $\beta_p\in S'$, there is a least $n$
such that $\beta_p\in G_n$. We will show that $\beta_p\in G_n$ implies
$\beta\in S'$ by induction on $n$. The proof is similar to
Lemma~\ref{lem:main_entailment_content}. As in the proof of that lemma, the
base case is done by inspection of $\cl{E}$. The crux in the inductive step
is when $F = N_\Dot$, so we will leave the other cases for the reader.

\Case{$F = N_\Dot$}
Let $\beta_p = N_\Dot(a,b,c,d)$. If $|a^{-1}(D)|\leq 1$ then $\beta_p\leq b$
or $\beta_p\leq c$. In either case the inductive hypothesis gives us
$\beta\in S'$. If $|a^{-1}(D)| \geq 2$ then the only possibility is that
$a(\ell_1) = a(p)\in D$. From the definition of $N_{\Dot}$ we have that
$b(\neq \ell_1,p) = c(\neq \ell_1,p) = \beta_p(\neq \ell_1,p)$. Every
element of $S'$ has content at coordinate $p$ in
$\{\Content(\beta(\ell_1)),\Dot,\Er\}$. If $\beta_p(p)\in D$ then $\Er\in
\Content(\beta_p(\N(\RR)))$ by Lemma~\ref{lem:inherent_properties} item (1),
a contradiction. Let us assume that $\beta_p(p)\in X$, since otherwise
$\beta_p = \beta$. By similar logic we have $b(p),c(p)\not\in D$, so $b(p) =
c(p)\in X$. By the inductive hypothesis, we must have $b(\ell_1) =
c(\ell_1)\in X$. From the definition of $N_\Dot$, this forces
$\beta_p(\ell_1)\in X$ as well, a contradiction.
\end{proof}  
\begin{thm} \label{thm:main_entailment} 
Assume that $m \geq \kappa+16$ and
\begin{itemize}
  \item $\Rel_{\leq m-1}(\AAM)\Entails f$ and $f$ is $n$-ary,

  \item $G = \{g_1,\dots,g_n\}\subseteq \AM^m$,

  \item $\RR = \Sg_{\AAM^m}(G)$ is computational non-halting and
    $\big| [m]\setminus \D(\RR) \big|\leq 10$,

  \item $f(g_1,\dots,g_n) = \alpha\in Y^m$,

  \item there is at most one $k\not\in \N(\RR)$ such that
    $\ApproxI(\alpha,k)$.
\end{itemize}
Then $\alpha\in R$.
\end{thm}
\begin{proof}
The proof is by induction on the number of positions in $G$ which are in
$Y$:
\[
  \sum_{i=1}^m \Big| \big\{ k \mid g_k(i)\in Y,\ g_k\in G \big\} \Big|.
\]
If this quantity is $0$ then $G\subseteq X^m$. Since $X$ is a subuniverse of
$\AAM$, we have that $f(\overline{g})(i)\in X$, contradicting $\alpha\in
Y^m$. This establishes the basis of the induction. The next claim is the
main tool we will use in the induction.

\begin{claim} \label{claim:main_entailment__monotone}
Let $\cl{E} = \{ e_1,\dots, e_n\}\subseteq \AM^m$ be synchronized and such
that for all $i\in [n]$, some $k\in [m]$, and some $\ell\in [n]$ we have
\begin{itemize}
  \item $e_i(\neq k) = g_i(\neq k)$,

  \item $e_i(k) \leq g_i(k)$, and

  \item $e_\ell(k) < g_\ell(k)$ (i.e.\ $e_\ell(k)\in X$ and 
    $g_{\ell}(k)\in Y$).
\end{itemize}
If $f(\overline{e}) = \alpha$ then $\alpha\in R$.
\end{claim}
\begin{claimproof}
Let $\SS = \Sg_{\AAM^m}(\cl{E})$ and observe that $\cl{E}$ has fewer
positions in $Y$ than $G$ does. The generators $\cl{E}$ are synchronized and
(from the hypotheses of the theorem) $|\D(\RR)| \geq \kappa + 6$, so $\SS$
is computational.

Let us suppose (toward a contradiction) that $\SS$ is halting. This implies
that there is a halting vector $\beta\in S$, and thus a term $t$ such that
$t(\overline{e}) = \beta$. The relations $\RR$ and $\SS$, their generators,
the term operation $t$, and the element $\beta\in S$ satisfy the hypotheses
of Lemma~\ref{lem:f_weak_monotone}, so there is a term operation $s$ such
that $\beta\leq s(\overline{g})$. Since $\beta$ is a halting vector and is
thus contained in $Y^m$, this implies that $s(\overline{g}) = \beta$.
Therefore $\beta\in R$, and so $\RR$ is halting, a contradiction. Thus $\SS$
is non-halting.

Since $\SS$ is computational and non-halting the inductive hypothesis
applies, so we have $\alpha\in S$ and there is a term operation $t'$ such
that $t'(\overline{e}) = \alpha$. By Lemma~\ref{lem:f_weak_monotone}, we
obtain a term operation $s'$ such that $\alpha\leq s'(\overline{g})$. Since
$\alpha\in Y^m$, this implies $\alpha = s'(\overline{g})$, so $\alpha\in R$.
\end{claimproof}

Let
\begin{align*}
  &K'
  = \N(\RR) \cup \big\{ k \mid \RR \text{ has } \ApproxI(\alpha,k) \big\} 
    \text{ and} \\
  &K
  = K' \cup \Big\{ k\in [m]\setminus K' \mid \text{for all } 
    \ell\in [m]\setminus K',\ 
    \alpha(\neq k,\ell)\not\in \big( R(\neq k,\ell) \big)_I \Big\}.
\end{align*}
Since $\RR$ is non-halting, $|\N(\RR)\cap \D(\RR)|\leq \kappa$. By
hypothesis $\big| [m]\setminus \D(\RR) \big| \leq 10$. It follows from these
that $|\N(\RR)| \leq \kappa + 10$, so combined with the last hypothesis of
the theorem we have $|K'|\leq \kappa + 11$. By
Proposition~\ref{prop:double_I_avoidance}, we also have that $|K| \leq
\kappa + 12$. Observe that if $\ell\not\in K$ then
\begin{itemize}
  \item $\ell\not\in\N(\RR)$,

  \item $\alpha_{\ell}\not\in R_I$ since $\RR$ does not have
    $\ApproxI(\alpha,\ell)$, and

  \item for each $\ell'\not\in K$ there is a set $L$ with $|L| = m-2$,
    $K\cup \{\ell, \ell'\}\subseteq L$, and $\alpha(L)\not\in \big( R(L)
    \big)_I$ (by Proposition~\ref{prop:double_I_avoidance} and the
    construction of $K$).
\end{itemize}
There is a subtlety in the last item above. Observe that for $L'\subseteq
[m]$ we have $R_I(L') \subseteq \big( R(L') \big)_I$ from
Definition~\ref{defn:RI}. The last item above therefore implies
$\alpha(L)\not\in R_I(L)$. We are now ready to proceed with the proof.

\begin{claim} \label{claim:main_entailment__dot}
$\Dot\in \Content(\alpha(\N))$.
\end{claim}
\begin{claimproof}
Suppose that $\Dot\not\in\Content(\alpha)$ and note that the hypotheses of
the theorem imply $|\D(\RR)| > \kappa+1$. If $\alpha_i\in C^m$ for some
$i\in [m]$ then $\RR$ is halting, by Lemma~\ref{lem:capacity} item (2).
Therefore for each $i$ we have $\alpha_i(i)\in D\cup X$ and $\alpha_i(\neq
i)\in C^m$, so $i\in \H(\RR)$. It follows that $\H(\RR) = [m]$ and so
$|\H(\RR) \cap \D(\RR)| = |\D(\RR)| > \kappa+1$. Thus $\RR$ is halting by
Lemma~\ref{lem:capacity} item (4), a contradiction. Suppose now that
$\alpha(i) \in D$ but $i\not\in \N(\RR)$. This implies that
$\alpha_i(\N(\RR))\in C^{|\N(\RR)|}\cap \RR(\N(\RR))$, contradicting
$\RR(\N(\RR))$ being non-halting (Proposition~\ref{prop:inherent_nonhalt}
item (1)).
\end{claimproof}

\begin{claim} \label{claim:main_entailment__content}
If $\alpha\not\in R$ then for every row $\ell\not\in K$, the content of all
the entries is in $\{\Content(\alpha(\ell)), \Dot, \Er\}$.
\end{claim}
\begin{claimproof}
It follows from Claim~\ref{claim:main_entailment__dot} that $\alpha(i)\in D$
for some $i\in \N$. Pick some $\ell\not\in K$. By the observations after
Claim~\ref{claim:main_entailment__monotone} above, there is a set $L$ such
that $\ell\in L$, $|L| = m-2$, and $\alpha(L)\not\in R_I(L)$ (we take $\ell
= \ell'$ in the observation). Construct the elements $e_1,\dots,e_n\in
\AAM^{m-1}$ (on coordinates $L \cup \{p\}$) as in
Lemma~\ref{lem:main_entailment_content} so that $e_i(L) = g_i(L)$ for all
$i$, but the $e_i$ have an ``extra'' row $p\not\in [m]$. Let $\cl{I} = \{
e_1,\dots, e_n \}$. Since $\Rel_{\leq m-1}(\AAM)\Entails f$,
$\cl{I}\subseteq \AM^{m-1}$, and $\alpha(L)\not\in R_I(L)$, by
Proposition~\ref{prop:I_avoidance} there is a term operation $t$ without $I$
in its term tree such that $t(\overline{e}) = f(\overline{e})$. Apply
Lemma~\ref{lem:main_entailment_content} with this term operation $t$ to
obtain $f(\overline{e})(p) = t(\overline{e})(p) = \alpha(\ell)$. The $p$-th
row of $\cl{I}$ will have at most the same number of $Y$ entries as the
$\ell$-th row of $G$. Let $\cl{E} = \{ h_1,\dots, h_n \}$ be obtained by
replacing the $\ell$-th row of $G$ with the $p$-th row of $\cl{I}$. It
follows that $f(\overline{h}) = \alpha$, and if $\cl{E}$ has fewer $Y$
entries than $G$ then $\alpha\in R$, by
Claim~\ref{claim:main_entailment__monotone}. The only way for $\alpha\not\in
R$ is if for every row $\ell\not\in K$, the content of all the entries is in
$\{\Content(\alpha(\ell)), \Dot, \Er\}$.
\end{claimproof}

\begin{claim} \label{claim:main_entailment__X}
If $\alpha\not\in R$ then for every distinct $\ell_1, \ell_2\not\in K$ and
$g_j\in G$, if $g_j(\ell_2)\in X\cup D$ then $g_j(\ell_1)\in X\cup D$.
\end{claim}
\begin{claimproof}
This is similar to Claim~\ref{claim:main_entailment__content}, but for
Lemma~\ref{lem:main_entailment_X}. Towards a contradiction, pick distinct
$\ell_1, \ell_2\not\in K$ such that $g_j(\ell_2)\in X\cup D$ and
$g_j(\ell_1)\not\in X\cup D$ for some $j$. As in the proof of
Claim~\ref{claim:main_entailment__content}, we have $\alpha(i)\in D$ for
some $i\in \N$ and there is a set $L$ such that $\ell_1, \ell_2\in L$, $|L|
= m-2$, and $\alpha(L)\not\in R_I(L)$ by the observations after
Claim~\ref{claim:main_entailment__monotone} above. Construct the elements
$e_1,\dots,e_n\in \AAM^{m-1}$ (on coordinates $L\cup \{p\}$) as in
Lemma~\ref{lem:main_entailment_X} so that $e_i(L) = g_i(L)$ for all $i$, but
the $e_i$ have an ``extra'' row $p\not\in [m]$. Since $g_j(\ell_2)\in X\cup
D$ and $g_j(\ell_1)\not\in X\cup D$, from the description of $e_j$ in
Lemma~\ref{lem:main_entailment_X} we have $e_j(p)\in X$.

Let $\cl{I} = \{ e_1,\dots, e_n \}$. Since $\Rel_{\leq m-1}(\AAM)\Entails
f$, $\cl{I}\subseteq \AM^{m-1}$, and $\alpha(L)\not\in R_I(L)$, by
Proposition~\ref{prop:I_avoidance} there is a term operation $t$ without $I$
in its term tree such that $t(\overline{e}) = f(\overline{e})$. Apply
Lemma~\ref{lem:main_entailment_X} with this term operation $t$ to obtain
$f(\overline{e})(p) = t(\overline{e})(p) = \alpha(\ell_1)$. The $p$-th row
of $\cl{I}$ will have at least one fewer $Y$ entries than the $\ell_1$-th
row of $G$, by the observation at the end of the previous paragraph. Let
$\cl{E} = \{ h_1,\dots, h_n \}$ be obtained by replacing the $\ell_1$-th row
of $G$ with the $p$-th row of $\cl{I}$. It follows that $f(\overline{h}) =
\alpha$, and since $\cl{E}$ has fewer $Y$ entries than $G$ we have
$\alpha\in R$ by Claim~\ref{claim:main_entailment__monotone}, a
contradiction.
\end{claimproof}

Let
\begin{align*}
  A &= \Big\{ i\in [n] \mid \Content(g_i(\ell)) = \Content(\alpha(\ell)) 
    \text{ for all } \ell\not\in K \Big\} & \text{and} \\
  B &= \Big\{ i\in [n] \mid g_i(\ell) \in D\cup X 
    \text{ for all } \ell\not\in K \Big\}.
\end{align*}
Seeking a contradiction, suppose that $\alpha \not\in R$. Apply
Claims~\ref{claim:main_entailment__content} and \ref{claim:main_entailment__X}
to every two-element subset of $[m]\setminus K$ to obtain $A\cup B = [n]$. Since
$\big| [n]\setminus K \big| \geq 4$, there must be two distinct coordinates $\ell,
\ell'\not\in K$ such that $\alpha(\ell) = \alpha(\ell')$ with common content in
$\{0,A,B\}$. Fix $\ell$ and $\ell'$ for the remainder of the proof.

Consider $\alpha_\ell\in R$. There must be a term operation $t$ such that
$\alpha_\ell = t(\overline{g})$. There is only a single coordinate $i$ at which
$\alpha_\ell(i)$ is possibly in $X$, namely $i = \ell$. Consider all term
operations $s$ such that $\alpha_\ell \leq s(\overline{g})$ and note that
$s(\overline{g})(\neq \ell) = \alpha_\ell(\neq \ell) = \alpha(\neq \ell)$. There
is at least one such $s$ that does not contain the following in its term tree:
\begin{itemize}
  \item $\wedge$, $N_0$, $N_\Dot$, or $P$, by the various parts of
    Lemma~\ref{lem:basic_facts_relns} and since $\alpha_\ell(\neq \ell)\in
    Y^{m-1}$,

  \item $H$ or $S$ since $\Dot\in\Content(\alpha_\ell)$,

  \item $I$ since $\ell\not\in K$, so $\RR$ does not have
    $\ApproxI(\alpha,\ell)$.
\end{itemize}
This leaves us with $s$ a term in $M$ and $M'$. 

\begin{claim}
$s(\overline{g})$ depends only on $g_a$ for $a\in A$.
\end{claim}
\begin{claimproof}
Fix $b\in B$ and note that $\big| [n]\setminus K \big| \geq 4$. The relation
$\RR$ is computational, so there is at most 1 coordinate $k\not\in K$ such that
$g_b(k)\in D$. Thus there are at least 2 coordinates $k\not\in K$, $k\neq \ell$
such that $g_b(k)\in X$. The operations $M$ and $M'$ are $X$-absorbing (Lemma
\ref{lem:basic_facts_AM} item (2)), so $s$ must be as well. Since $\Er\not\in
\Content(\alpha_\ell(\neq \ell))$ and $s(\overline{g}) = \alpha_\ell$, it
follows that $s(\overline{g})$ cannot depend $g_b$, since $g_b$ has content in
$X$ for at least $2$ coordinates outside of $K\cup \{\ell\}$.
\end{claimproof}

We now have that $s(\overline{g})$ depends only on $g_a$ for $a\in A$. Recall
that $\ell, \ell'\not\in K$ were chosen so that $\ell \neq \ell'$ and
$\alpha(\ell) = \alpha(\ell')$. Since $\alpha(\ell) = \alpha(\ell')$, it follows
that $g_a(\ell) = g_a(\ell')$ for all $a\in A$, so
\[
  \alpha(\ell) = \alpha(\ell') 
  = s(\overline{g})(\ell')
  = s(\overline{g})(\ell)
  = \alpha_\ell(\ell).
\]
Hence $s(\overline{g})(\ell) = \alpha(\ell)$ and so $s(\overline{g}) = \alpha$.
Therefore $\alpha\in R$.
\end{proof} 

This completes the proofs of all the theorems referenced in the proof of
Corollary~\ref{cor:finitely_related} at the start of the section. We have
thus proven that if $\M$ halts then $\AAM$ is finitely related.

\section{Concluding remarks}\label{sec:conclusion} 
Combining Theorem~\ref{thm:M_not_halt_high_degree} and
Corollary~\ref{cor:finitely_related} yields the theorem claimed in the title
of this paper.

\begin{thm} \label{thm:halts_iff_finrel}  
The following are equivalent.
\begin{enumerate}
  \item $\M$ halts,
  \item $\deg(\AAM) < \infty$ (i.e.\ $\AAM$ is finitely related),
  \item $\M$ halts with capacity at least $\deg(\AAM) - 15$.
\end{enumerate}
\end{thm}  

\noindent Many standard results follow from this theorem. We detail a couple
of the more interesting ones below.
\begin{itemize}
  \item There exists infinitely many Minsky machines $\M$ such that the
    halting status of $\M$ is independent of ZFC (see
    Chaitin~\cite{Chaitin_Incomplete} or
    Kolmogorov~\cite{Kolmogorov_Complexity}). As a consequence of the
    theorem above, there are finite algebras $\A$ whose finite-relatedness
    is independent of ZFC.

  \item Let $\sigma$ be a fixed finite algebraic signature (name and arity
    specification of the functions) and define
    \[
      \maxdeg_{\sigma}(n)
      = \sup \Big\{ \deg(\A) \mid \A \text{ has signature $\sigma$, is
          finite degree, and } |\A| \leq n \Big\}
    \]
    If we remove the requirement that the algebras have signature $\sigma$
    then it is not too hard to show that $\maxdeg(n)$ is infinite. Let $\tau$
    be the signature of $\AAM$ and observe that $\tau$ does not depend on
    $\M$. It follows from Theorem~\ref{thm:halts_iff_finrel} that
    $\maxdeg_\tau(n)$ is not computable, and so $\maxdeg_\sigma(n)$ is not,
    in general, computable. This is essentially the Busy Beaver function of
    Rad\'{o}~\cite{Rado_BusyBeaver}.
\end{itemize}

There are several related problems which are conjectured to be undecidable
as well. We have shown that given a finite set of operations $\cl{F}$, it is
undecidable whether there is finite $\cl{R}$ such that $\Rel(\cl{F}) =
\RClo(\cl{R})$. The dual of this problem is also suspected to be
undecidable.

\begin{prob}  
Decide if a clone is finitely generated: given finite $\cl{R}$, decide
whether there is a finite $\cl{F}$ such that $\Pol(\cl{R}) = \Clo(\cl{F})$.
\end{prob}   

The most sweeping result on finitely related algebras is the following
theorem. The ``if''  portion is due to Aichinger, Mayr,
McKenzie~\cite{AichingerMayrMcKenzie_CubeFinRelated} and the ``only if''
portion is due to Barto~\cite{Barto_FinRelCube}.

\begin{thm} 
A finite algebra in a congruence modular variety is finitely related if and
only if it has a cube term.
\end{thm}  

\noindent The existence of a cube term is a weak Maltsev condition, but it
is a decidable property. This follows independently from Kazda and
Zhuk~\cite{KazdaZhuk_CubeArities} and Kearnes and
Szendrei~\cite{KearnesSzendrei_CubeBlockersInfinite}. As a consequence of
Theorem~\ref{thm:halts_iff_finrel}, there can be no decidable property which
characterizes finitely related meet-semidistributive algebras (of which
$\AAM$ is one). It is still possible, however, that there is an undecidable
weak Maltsev condition which does.

\begin{prob}  
Is there a weak Maltsev condition that characterizes finite relatedness for
finite algebras in congruence meet-semidistributive varieties?
\end{prob}   

The next two problems concern the theory of Natural Dualities, and date back
to the start of the field in the 1970s (see
McNulty~\cite{McNulty_JugglersDozen} section 3 for a history of the
problem). A good reference for the background is Clark and
Davey~\cite{ClarkDavey_DualWorking}. We produce the \emph{duality entailment
constructions} from constructions (1)--(4) of Section~\ref{sec:background}
by replacing (4) with
\begin{itemize}
  \item[($4'$)] \emph{bijective} projection of a relation onto a subset of
    coordinates.
\end{itemize}
A set of relations $\cl{R}$ duality entails a relation $\RR$ if and only if
$\RR$ can be constructed in finitely many steps from the entailment
constructions (1)--(3) of Section~\ref{sec:background} and ($4'$) above. If
this is the case then we write $\cl{R}\EntailsDual \RR$ and we refer to the
set of all such $\RR$ as $\RCloDual(\cl{R})$.

\begin{prob}  
Decide if every relation of an algebra is duality entailed by a finite
subset of them. That is, given algebra $\A$, decide whether there is finite
$\cl{R}$ such that $\Rel(\A) = \RCloDual(\cl{R})$.
\end{prob}   

\noindent It should be clear from the constructions that
$\RCloDual(\cl{R})\subseteq \RClo(\cl{R})$, so if $\A$ is finitely duality
related then it is finitely related. It follows that $\AAM$ is not finitely
duality related if $\M$ does not halt.

\begin{prob}  
If $\M$ halts, is $\AAM$ finitely duality related?
\end{prob}   

\noindent A positive answer to this problem would prove the undecidability
of the duality entailment problem. If an algebra is finitely duality
related then it is dualizable. The converse does not follow, however. This
leads us to a more general (and more important) version of the above
problem.

\begin{prob}  
Decide whether a finite algebra is dualizable.
\end{prob}   

\bibliographystyle{amsplain}   
\bibliography{fin-related-undec}
\begin{center}
  \rule{0.61803\textwidth}{0.1ex}   
\end{center}
\end{document}